\documentclass[
    aps,
    prx,
    reprint,
    superscriptaddress
]{revtex4-2}

\usepackage{amsmath,amssymb}
\usepackage{bm}
\usepackage{algorithm}
\usepackage{algpseudocode}
\usepackage{siunitx}
\usepackage{textcomp}
\usepackage{xcolor}
\usepackage[hidelinks]{hyperref}
\usepackage{comment}
\usepackage[nobottomtitles]{titlesec} 
\usepackage{bbm}

\usepackage[capitalize]{cleveref}
\usepackage{etoolbox}

\usepackage{graphicx}
\usepackage{caption}
\usepackage{subcaption}
\usepackage{ragged2e}
\usepackage{enumitem}
\usepackage{mathtools}
\usepackage{tabularx}
\usepackage{array}

\makeatletter
\def\@makecaption#1#2{%
  \vskip\abovecaptionskip
  \sbox\@tempboxa{%
    \parbox{\linewidth}{%
      \raggedright\justifying
      \textbf{#1}~#2%
    }%
  }%
  \ifdim \wd\@tempboxa >\linewidth
    \parbox{\linewidth}{\raggedright\justifying \textbf{#1}~#2}%
  \else
    \parbox{\linewidth}{\raggedright\justifying \textbf{#1}~#2}%
  \fi
  \vskip\belowcaptionskip
}
\makeatother

\usepackage[stable]{footmisc}
\usepackage{physics}
\usepackage{cancel}

\usepackage{amsthm}

\theoremstyle{plain}
\newtheorem{theorem}{Theorem}[section]
\newtheorem{proposition}[theorem]{Proposition}
\newtheorem{lemma}[theorem]{Lemma}

\newtheorem{fact}[theorem]{Fact}

\theoremstyle{definition}
\newtheorem{definition}[theorem]{Definition}

\newtheoremstyle{boldremark}
  {\topsep}      
  {\topsep}      
  {\normalfont}  
  {0pt}          
  {\bfseries}    
  {.}            
  {.5em}         
  {}             

\theoremstyle{boldremark}
\newtheorem{innerremark}[theorem]{Remark}

\newenvironment{remark}
  {%
   \pushQED{\qed}%
   \innerremark}
  {\popQED\endinnerremark}

\newcommand{\CC}{\mathbb{C}}
\newcommand{\EE}{\mathsf{E}}
\newcommand{\II}{\mathbb{I}}
\newcommand{\NN}{\mathbb{N}}
\newcommand{\PP}{\mathsf{P}}
\newcommand{\RR}{\mathbb{R}}

\newcommand{\indicator}{\mathbbm{1}}

\newcommand{\muvec}{\boldsymbol{\mu}}

\newcommand{\fvec}{\boldsymbol{f}}

\newcommand{\vvec}{\boldsymbol{v}}

\newcommand{\xvec}{\boldsymbol{x}}
\newcommand{\yvec}{\boldsymbol{y}}
\newcommand{\zvec}{\boldsymbol{z}}

\newcommand{\nuvec}{\boldsymbol{\nu}}
\newcommand{\phivec}{\boldsymbol{\phi}}

\newcommand{\Bin}{\mathrm{Bin}}
\newcommand{\Geo}{\mathrm{Geo}}
\newcommand{\TruncGeo}{\mathrm{TruncGeo}}
\newcommand{\Markov}{\mathrm{Markov}}

\newcommand{\Jeroen}[1]{\textcolor{olive}{Jeroen: #1}}
\newcommand{\sk}[1]{\textcolor{blue}{Sounak: #1}}

\hypersetup{ colorlinks=true, citecolor=blue, linkcolor=blue, filecolor=blue, urlcolor=blue }

\begin{document}

\title{A Dynamic Multiplexing Policy for a Quantum Repeater
}

\date{\today}

\author{Jeroen Grimbergen}
\email{j.grimbergen@tudelft.nl}\affiliation{QuTech, Department of EEMCS, Kavli Institute of Nanoscience, TU Delft, Lorentzweg 1, 2628 CJ Delft, The Netherlands}

\author{Sounak Kar}
\affiliation{QuTech, Department of EEMCS, Kavli Institute of Nanoscience, TU Delft, Lorentzweg 1, 2628 CJ Delft, The Netherlands}

\author{Michal van Hooft}
\affiliation{QuTech, TU Delft, Lorentzweg 1, 2628 CJ Delft, The Netherlands}

\author{Conor Bradley}
\affiliation{Delft Networks, 
Lorentzweg 1, 2628 CJ Delft, 
The Netherlands}

\author{Stephanie Wehner}
\affiliation{QuTech, Department of EEMCS, Kavli Institute of Nanoscience, TU Delft, Lorentzweg 1, 2628 CJ Delft, The Netherlands}

\begin{abstract}
We consider a multiplexed quantum repeater that distributes entanglement between two end nodes. Multiplexing is achieved through optical integration of many quantum chips. Each chip hosts an optically addressable communication qubit and a separate memory qubit. The communication qubit serves as an entanglement generation interface between different quantum chips, and the memory qubit can be used to store entanglement. 
The quantum chips on the repeater are interconnected using a reconfigurable router, which makes it possible to dynamically assign quantum chips for entanglement generation with either of the two end nodes in every end-to-end communication cycle. 
We propose a dynamic multiplexing policy in which after an entangled link has been established with one of the end nodes, all remaining quantum chips are assigned to the opposite end node. We compare this dynamic policy to a policy in which the assignment of quantum chips to end nodes is fixed. We consider a parameter regime where on average less than one entangled link is generated per end-to-end communication cycle, which is the relevant regime for near-term quantum networks. We show that in this regime, the dynamic multiplexing policy can lead to a significant improvement in fidelity over a fixed policy, while marginally improving the rate. Moreover, even though the dynamic multiplexing policy requires a deeper, and hence, more lossy, router than the fixed policy, it can still achieve higher secret key rates in the parameter regime studied. This makes dynamic multiplexing with a many-quantum-chip repeater especially relevant for the development of near-term quantum networks.
\end{abstract}

\maketitle

\section{Introduction}






Quantum networks \cite{kimble2008quantum,wehner2018quantum} may realize a range of applications including communication with information-theoretic security \cite{ekert1991quantum,renner2008security}, non-local coordination in distributed systems \cite{dasilva2025entanglement, gardiner2026learning}, and new tests of fundamental physics at the intersection of quantum mechanics and general relativity \cite{covey2025probing}. The key resource for these applications is quantum entanglement. Recent experiments have already demonstrated entanglement between quantum memories separated by more than $\SI{10}{\kilo\meter}$ of fibre \cite{yu2020entanglement, van2022entangling, knaut2024entanglement, stolk2024metropolitan, liu2024creation, luo2025entangling, liu2026long}. This was achieved using heralded entanglement generation protocols in which an entangled state is created conditional on the measurement of a photon. However, the success probability of these protocols decreases exponentially with the distance between nodes due to transmission loss in optical fibre.

These transmission losses can be overcome with quantum repeater chains in which one or more intermediate quantum repeater nodes are used to divide the distance between two end nodes into shorter segments \cite{briegel1998quantum}. Entanglement between the end nodes is created by first generating entanglement between neighbouring nodes in the chain, and then performing entanglement swaps \cite{zukowski1993event} at the repeater nodes. Since entanglement on each segment is still generated only probabilistically, quantum memories at the repeater nodes are necessary to store entanglement until it can be swapped.

A bottleneck to the realization of large-scale quantum repeater chains is that the rate at which entanglement is generated must be higher than the decoherence rate of entanglement stored in memory 
\cite{humphreys2018deterministic,zhang2024fast,van2022entangling,liu2026long}. 
The rate of entanglement generation can be improved by multiplexing the entanglement generation attempts over temporal, spectral, or spatial modes \cite{tian2017spatial,chang2019long, lago2021telecom,liu2021heralded,krutyanskiy2024multimode,li2025parallelized,ruskuc2025multiplexed, teller2025solid,cui2025metropolitan,canteri2025photon,zhu2026metropolitan}. 

We 
consider a multiplexed repeater as sketched in \cref{fig:system overview}a. Two idealized end nodes, labeled left and right, are connected by a repeater node. Multiplexing is achieved, as first proposed in Ref. \cite{lee2022quantum}, through optical integration of many quantum chips that each contain an optically addressable communication qubit and a separate memory qubit.
The quantum chip can be realized for example in color centers \cite{bradley2022robust, parker2024diamond, stas2022robust, hermans2022qubit, bartling2022entanglement, reiserer2016robust} or trapped ions \cite{negnevitsky2018repeated,drmota2023robust, huang2025realization}, and their optical integration involves a reconfigurable router, which can be built from Mach-Zehnder interferometer switches \cite{lee2022quantum}, or other switching technologies, using MEMS \cite{errando2019mems}, electro-optics \cite{hu2025integrated}, or piezoelectrics \cite{tian2024piezoelectric}. 
In contrast to other forms of multiplexing where the assignment of multiplexed modes to nodes has, to the best of our knowledge, always been considered fixed \cite{briegel1998quantum, collins2007multiplexed, razavi2009physical, razavi2009quantum,dhara2021subexponential, mantri2025comparing, haldar2025reducing,santra2019quantum, dhara2022multiplexed, kunzelmann2025multiplexed}, such a reconfigurable router provides the freedom to dynamically assign which quantum chips, i.e. which modes should attempt remote entanglement generation with which end node. It was shown in Ref. \cite{lee2022quantum} that in the limit of many quantum chips, a dynamic multiplexing policy yields similar rates of end-to-end entanglement generation to a fixed policy, and in work done in parallel to this paper, Ref. \cite{elsayed2026balancing}, one sees that dynamic multiplexing can also improve the fidelity of end-to-end entanglement generation. However, the proofs in \cite{elsayed2026balancing} rely essentially on strict asymmetry, and their numerical evaluations are also restricted only to this case. Furthermore, there was no consideration of the effect of imperfections in the repeater node, or losses in the optical router. 

Here we show that for a symmetrically placed repeater node, 
a certain dynamic assignment of quantum chips can significantly improve the fidelity of end-to-end entanglement over a fixed assignment, while marginally improving the rate, in the considered regime reflecting near-term hardware.
The dynamic policy considered in this work aims to minimize the time between a successful link generation to the left and right end nodes so that the time that a link has to be stored in memory before it can be swapped is reduced. 
It does so by assigning all quantum chips to the right once a link has been generated to the left, or vice versa. In Remark \ref{remark: oss interpretation} we highlight that the dynamic policy obtains its advantage over the fixed policy not only because it can match links faster, but also because it avoids links to the same end node to be queued up while waiting for links with the other end node to be generated.

This work is an analytical comparison of dynamic and fixed entanglement generation policies for a multiplexed quantum repeater. 
We quantify the improvement of the dynamic policy over the fixed policy in terms of the steady-state rate, fidelity, and secret key rate. We obtain exactly computable expressions for the steady-state rate and fidelity in a depolarizing noise model via a Markov chain analysis that uses the Ergodic Theorem for Markov-modulated Processes described in Lemma \ref{lemma: main replacing random variable by its expectation in ergodic theorem}.
Additionally, closed-form expressions based on an approximate Markov chain allow for efficient heuristic optimization of the rate-fidelity tradeoff for the secret key rate. 
The Markov chain analyses are validated using a discrete-event simulation in NetSquid \cite{coopmans2021netsquid}. Our main findings are summarized below.

\begin{itemize}
    \item 
    \textbf{Rate and fidelity:} Dynamic multiplexing can significantly improve the fidelity of end-to-end links generated by a quantum repeater in steady-state, especially for parameter values relevant to near-term quantum networks such as small entanglement generation probability and short memory coherence time, while also yielding a marginally higher rate. 
    This means that by implementing a dynamic multiplexing policy it is possible to achieve much better network performance than with fixed multiplexing, for the same quantum hardware. 
    \item \textbf{Secret key rate:} For an end-to-end distance on the order of $\SI{100}{\kilo\meter}$, remote entanglement generation probability per single attempt of $10^{-5}$, memory coherence time of $\SI{100}{\second}$, lossless optical switches, and other parameters as specified in \cref{fig:NILN_eta_skr_lowerbound}, we evaluate a heuristic for the secret key rate. It is found that with $32$ chips on the repeater node, dynamic multiplexing yields a $23.9$ times higher value for this heuristic than fixed multiplexing, while with $512$ chips the value is $1.65$ times higher. This means that implementing a dynamic multiplexing policy is especially relevant for near-term hardware where the number of chips may be limited.
    \item \textbf{Switching losses:} If an optical router implementation of many-quantum-chip multiplexing as in \cite{lee2022quantum} is used, then dynamic multiplexing requires a deeper optical router than fixed multiplexing (see Assumption \ref{as:policies}). At an efficiency of $90\%$ for each optical switch, which is in reach for MEMS \cite{errando2019mems}, and otherwise the same parameters as in the previous point, the heuristic estimate of the secret key rate for fixed multiplexing becomes zero up to $512$ chips, while dynamic multiplexing achieves positive secret key rates in that regime, showing that dynamic multiplexing can remain advantageous even though it incurs more switching losses. 
\end{itemize}

The outline of the rest of this paper is as follows. In \cref{section: related work} we discuss related work. The model assumptions as well as the definitions of the fixed and dynamic multiplexing policies are presented in \cref{section: system model}. The Markov chain performance analysis is done in \cref{section: performance analysis}. Numerical results for the near-term regime are presented in \cref{section: numerical evaluations} and discussed in \cref{section: discussion}. We end with an outlook in \cref{section: outlook}.

\begin{figure*}
    \centering
    \includegraphics[trim=0cm 7.5cm 0cm 0cm, clip, width=\linewidth]{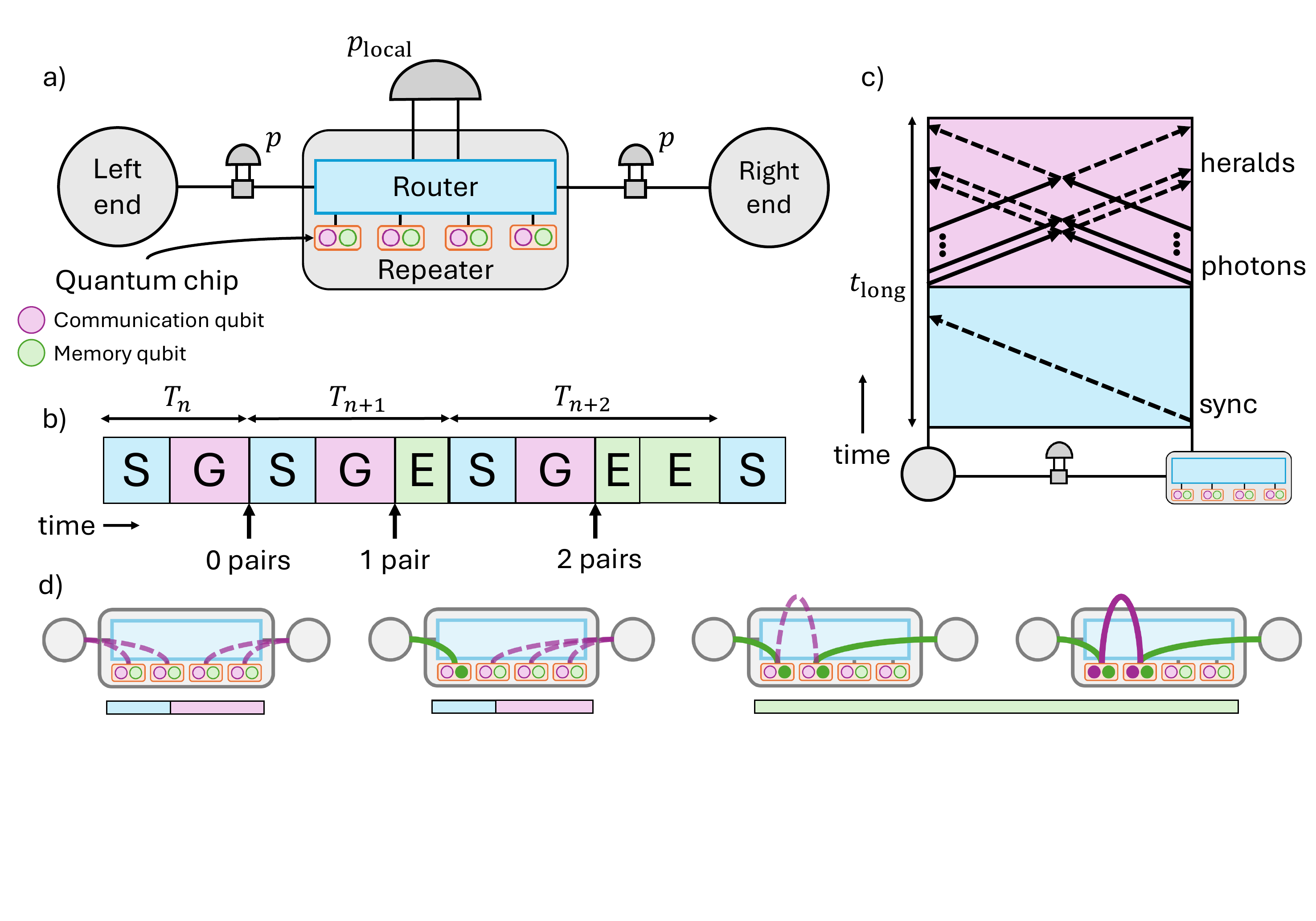}
    \caption{\textbf{Overview of many-quantum-chip repeater connecting two end nodes.} \textbf{a}) A many-quantum-chip repeater node connecting two end nodes. Quantum chips are depicted as rectangles with communication qubit (purple, left) and memory qubit (green, right). Quantum chips on the end nodes are not shown. A reconfigurable router can connect the quantum chips on the repeater node to remote entanglement generation interfaces between the repeater and either of the end nodes for which a single entanglement generation attempt has success probability $p$.
    It can also connect two chips on the repeater to a local entanglement generation interface where each attempt succeeds with probability $p_\mathrm{local}$. The local entangled links between chips on the repeater node mediate the entanglement swaps needed to create end-to-end links. \textbf{b}) Possible sequence of synchronization (S), entanglement generation (G) and entanglement swap (E) phases depending on how many pairs of links can be matched at the end of each generation phase. The total duration of time step $n$ is a random variable $T_n$ because the number of entanglement swaps and their durations are random variables. \textbf{c}) Classical communication (dashed arrows) and quantum communication (solid arrows) between repeater and left end node during synchronization and remote entanglement generation phases. The combined duration of these phases is $t_\mathrm{long}$. In the synchronization phase, the repeater node sends a synchronization message, including for example the entanglement-generation schedule for the upcoming generation phase. After receiving this message, the end node is allowed some time to process it. At the start of the entanglement generation phase, both nodes send their communication-qubit-entangled photons in a pipelined sequence to the heralding station, which measures the photons and returns the heralding messages to the nodes. \textbf{d}) Possible sequence of entanglement generation attempts under the dynamic multiplexing policy. When there are no links in memory, then during the Sync-Gen cycle half of the chips attempt to the left, and half attempt to the right, as indicated by dashed purple lines. If a link succeeds to the left, it gets stored in memory indicated by the solid green line, and in the next Sync-Gen cycle, all attempts go to the right. When a link to the right succeeds, it is matched to the link on the left. During the entanglement swap phase, local entanglement generation attempts are performed between the chips of the matched links until success or cutoff is reached. If the local link succeeds, then Bell state measurements on the respective chips project the link to the left, the local link, and the link to the right to a single end-to-end link. This last step is not shown in the figure.}
    
    \label{fig:system overview}
\end{figure*}

\section{Related Work}\label{section: related work}
As discussed in the introduction, the idea for dynamic entanglement generation between multiplexed nodes comes from Ref. \cite{lee2022quantum}. Recently, dynamic chip
assignment on such a multiplexed node was considered in Ref. \cite{elsayed2026balancing}. 
They consider a one-step mismatch minimization policy in the case that the repeater node is placed asymmetrically between the end nodes, and provide steady-state rate and fidelity \textit{bounds} under this. For the symmetric placement of the repeater node between the end nodes and in the near-term parameter regime we consider here, their one-step mismatch minimization policy reduces to our dynamic policy in Definition \ref{def: dynmux}. However, for symmetric placement of the repeater node, the lower bound on fidelity in \cite{elsayed2026balancing} applies to both the fixed and the dynamic policy, and therefore cannot be used to show the advantage of dynamic multiplexing in this case. Furthermore, the secret key rate and switching losses have not been considered in \cite{elsayed2026balancing}. Finally, the numerical evaluations in \cite{elsayed2026balancing} have been done under the assumption of entanglement generation probabilities on the order of $10^{-1}$, whereas we consider entanglement generation probabilities relevant to near-term quantum networks, which are on the order $10^{-6}$ to $10^{-5}$ \cite{knaut2024entanglement, stolk2024metropolitan, liu2024creation}.

Dynamic policies in multiplexed quantum networks have also been considered in the context of routing entanglement between users in a two-dimensional grid network \cite{van2025entanglement}. These policies dynamically selected which pairs of links to swap in order to create a path between users in a two-dimensional quantum network. We consider dynamic policies that assign resources of a single repeater node for entanglement generation in a one-dimensional multiplexed quantum repeater chain, which is a fundamentally different problem.

Except for Refs. \cite{lee2022quantum} and \cite{elsayed2026balancing}, previous analytical work on multiplexed quantum repeater chains has, to the best of our knowledge, exclusively analyzed fixed entanglement generation policies \cite{briegel1998quantum, collins2007multiplexed, sangouard2011quantum, razavi2009physical, razavi2009quantum,dhara2021subexponential, mantri2025comparing, haldar2025reducing,santra2019quantum, dhara2022multiplexed, kunzelmann2025multiplexed}. The focus of these investigations has been on policies for entanglement swapping, cutoffs, entanglement distillation, and classical communication. Other work that analyses fixed entanglement generation policies includes proposals for multiplexed midpoint sources \cite{dhara2022heralded, chen2023zero}, and multiplexing with asymmetric resource distributions between nodes \cite{munro2010quantum,propp2025quantum}.

There is a body of work that uses Markov chains for analyzing the rate and fidelity of quantum repeater chains \cite{shchukin2019waiting, inesta2023optimal, reiss2023deep, kunzelmann2025multiplexed, lee2022quantum, shchukin2022optimal, zang2024analytical}. 
There is even work that finds provably optimal policies for the delivery times of end-to-end links with some minimum fidelity using a formulation in terms of Markov Decision Processes \cite{inesta2023optimal}, or that heuristically optimizes policies for the end-to-end delivery time and fidelity \cite{haldar2024fast} or the secret key rate \cite{reiss2023deep} using reinforcement learning.
Other analytical tools used to model quantum repeater performance include elementary probability theory \cite{collins2007multiplexed, mantri2025comparing}, Fourier transforms \cite{brand2020efficient, li2021efficient}, and generating functions \cite{kar2023analysis, kamin2023exact, goodenough2025noise}. 

Our approximate Markov chain analysis is inspired by the analysis in \cite{collins2007multiplexed}, where the approximation was made that at most one entanglement generation attempt succeeds per time step. A similar approximate Markov chain has been used in \cite{kunzelmann2025multiplexed} to study the rate in a fixed multiplexing policy, but not the fidelity. 
To the best of our knowledge, an exact Markov chain analysis for the rate and fidelity of a multiplexed quantum repeater has not been considered in the literature.

\section{Assumptions on the Multiplexed Quantum Repeater}\label{section: system model}


In this section we provide the assumptions for the multiplexed quantum repeater depicted in \cref{fig:system overview}a and define the dynamic and fixed multiplexing policies that will be compared in this work. 
We perform a discrete-time analysis of the repeater; c.f. \cref{fig:system overview}b and c. The model parameters are collected in Tables \ref{tab:fundamental-parameters} and \ref{tab:derived-quantities} in Appendix \ref{app: parameters}. We give context to these parameters in the form of the numbers for color centers and trapped ions. 

\begin{enumerate}[label={\textbf{A\arabic*}},nolistsep,leftmargin=*]
\item \label{assumption:werner_state} 
\textbf{Many-quantum-chip multiplexing}

A quantum chip consists of a communication qubit and a memory qubit that support universal quantum logic and standard basis measurements. We consider a many-quantum-chip multiplexed repeater node that connects two end nodes labeled left and right, as depicted in \cref{fig:system overview}a.
There are $2m$ quantum chips on the repeater node. 
There are always sufficiently many quantum chips available on the end nodes to serve all entanglement generation requests from the repeater node.


\item \label{} \textbf{Spatial symmetry} 

We assume the system is left-right symmetric.

\item \label{}\textbf{Entanglement generation between quantum chips}

The communication qubits of different quantum chips can be entangled using a heralded entanglement generation protocol, for example with Fock-basis (``single-click'') \cite{cabrillo1999creation, hermans2023entangling} or time-bin (``double-click'') \cite{barrett2005efficient} encoding. A \emph{remote entanglement generation attempt} is a single heralded entanglement generation attempt between the communication qubits of a quantum chip on the repeater node and a quantum chip on one of the end nodes. It uses an entanglement generation interface located between the nodes, and succeeds with probability $p$. A \emph{local entanglement generation attempt} is a single heralded entanglement generation attempt between the communication qubits of two quantum chips on the repeater node. It uses an entanglement generation interface located at the repeater node, and succeeds with probability $p_\mathrm{local}$.

\item \label{ass:router} \textbf{Reconfigurable router} 

The quantum chips on the repeater node are connected to the inputs of a reconfigurable optical router, realized, for example, using Mach-Zehnder interferometer switches \cite{lee2022quantum}. We assume that the router is well-equipped to connect any quantum chip on the repeater node to either end node, and to connect any pair of quantum chips on the repeater node simultaneously to the local entanglement generation interface. Since there are $2m$ inputs (one for each quantum chip) and $4$ outputs (one for each end node and two for the local entanglement generation interface), not every connection can be made simultaneously. The router needs to be dynamically reconfigured to change the connection of the quantum chips to the output ports. The router has a depth $d(m)$, which means that there are at most $d(m)$ switches between a quantum chip and an output port. The photon loss in each switch is $\eta_\mathrm{switch}$. We assume that the end nodes require a router of the same depth as the repeater node to route photons from their quantum chips to the remote entanglement generation interfaces.

\item \label{as:timeStep} \textbf{Time steps} 

We assume that the repeater operates in discrete time steps, which consist of a synchronization phase, an entanglement generation phase, and an optional entanglement swap phase. A sequence of these three phases is depicted in \cref{fig:system overview}b, and each phase is described in detail below.

        \item[] \textit{Synchronization:} The repeater node synchronizes with the end nodes for the upcoming remote entanglement generation attempts. During the synchronization phase the repeater node shares an entanglement generation schedule with the end nodes detailing which quantum chips will attempt remote entanglement generation at what times in the upcoming entanglement generation phase, see \cref{fig:system overview}c. The end nodes receive this schedule and locally assign their own quantum chips for remote entanglement generation according to the schedule. Some time is also assumed to be allocated to the synchronization of the quantum control hardware. In a more advanced protocol one could imagine doing this synchronization only in some of the time steps.
        \item[] \textit{Entanglement generation:} The repeater node performs remote entanglement generation attempts with both of the end nodes simultaneously. In the entanglement generation phase the communication qubits of all quantum chips due for remote entanglement generation are initialized in the state required for the protocol. Subsequently, all entanglement generation attempts proceed in rapid succession, see \cref{fig:system overview}c. In a midpoint-heralded scheme, the communication qubits of the different quantum chips are sequentially optically excited, and the optical router is dynamically reconfigured to route the outgoing spin-entangled photons to the required heralding station. 
        When an entangled link is successfully heralded for some quantum chip on the repeater node it performs a \textsc{swap} gate to move its half of the entangled link from the communication qubit onto the memory qubit. The generation phase ends after the repeater node has received the heralding messages of all entanglement generation attempts, and all entangled links are stored in the memory qubits. It is possible that multiple entangled links are generated in this phase.

        The classical and quantum communication during the synchronization and entanglement generation phase is schematically depicted in \cref{fig:system overview}c. The combined duration of these two phases is $t_\mathrm{long}$. 
        
        \item[] \textit{Entanglement swap:} If at the end of the entanglement generation phase there are entangled links to both the left as well as the right end node, then the repeater node enters the entanglement swap phase. We assume that links to the left and links to the right form two separate first-in-first-out (FIFO) queues and that the repeater node forms a pair of links by matching the first link in the left queue to the first link in the right queue. An entanglement swap between a link from chip $i$ to the left end node and a link from chip $j$ to the right end node is mediated by a local entangled link between the communication qubits of chips $i$ and $j$, see \cref{fig:system overview}d. The repeater node attempts to establish this link through its local entanglement generation interface, and as soon as it succeeds, the entanglement swap can be completed through two deterministic Bell state measurements on the joint states of the communication and memory qubits of chips $i$ and $j$ respectively. 
        The outcomes of these measurements can be communicated to the end nodes during the next synchronization phase so that the end-to-end link is heralded. The repeater node performs swaps as soon as possible (asap), in the sense that it continues matching and swapping links from the two queues until one of the queues is empty. 

\item \label{} \textbf{Cutoff on local entanglement generation}

The use of a cutoff on the number of entanglement generation attempts can be used to ensure that decoherence of a memory qubit due to entanglement generation attempts on its communication qubit does not exceed a chosen threshold. We denote this cutoff by $c_\mathrm{local}$. With this cutoff, local entanglement generation durations follow independent and identically distributed (iid) truncated geometric distributions $\TruncGeo(p_\mathrm{local}, c_\mathrm{local})$, where $p_\mathrm{local}$ denotes the success probability of each trial and  $c_\mathrm{local}$ denotes the maximum number of trials. 
An entanglement swap succeeds if and only if local entanglement is generated before $c_\mathrm{local}$ trials. 
Clearly, the success probability is given by ${p_\mathrm{swap}:=1-(1-p_\mathrm{local})^{c_\mathrm{local}}}$.

\item \label{as:policies} \textbf{Entanglement generation policies}

We now formally introduce the fixed multiplexing (FxdMux) policy and the dynamic multiplexing (DynMux) policy considered in this work. These are policies for the remote entanglement generation phase. We say that a quantum chip is \emph{free} when it does not hold any entangled links. The FxdMux policy is then defined as follows.
\begin{definition}[FxdMux]\label{def: fxdmux}
    There is a fixed assignment of $m$ quantum chips to the left end node and $m$ quantum chips to the right end node across time steps. During the entanglement generation phase, all free quantum chips on the repeater node attempt entanglement generation with a quantum chip on their respective end node.
\end{definition}
Fixed multiplexing can be achieved with a router of depth $d^\mathrm{fxd}(m)=\log_2(m)$. The router is shown for $\log_2(m)=2$ in \cref{fig:routers}a. It consists of a binary tree of unit cells, where in each unit cell there is one layer of optical switches. The left switch can route the left two inputs to either of the left two outputs, and similarly for the right switch. With such a router, $m$ chips of the repeater can connect to the left end node, and the remaining $m$ chips can connect to the right end node. Moreover, any pair of a left and a right chip can be connected simultaneously to the local entanglement generation interface for an entanglement swap. 

\begin{figure}
    \centering
    \includegraphics[trim=2cm 14cm 17cm 7cm, clip, width=\linewidth]{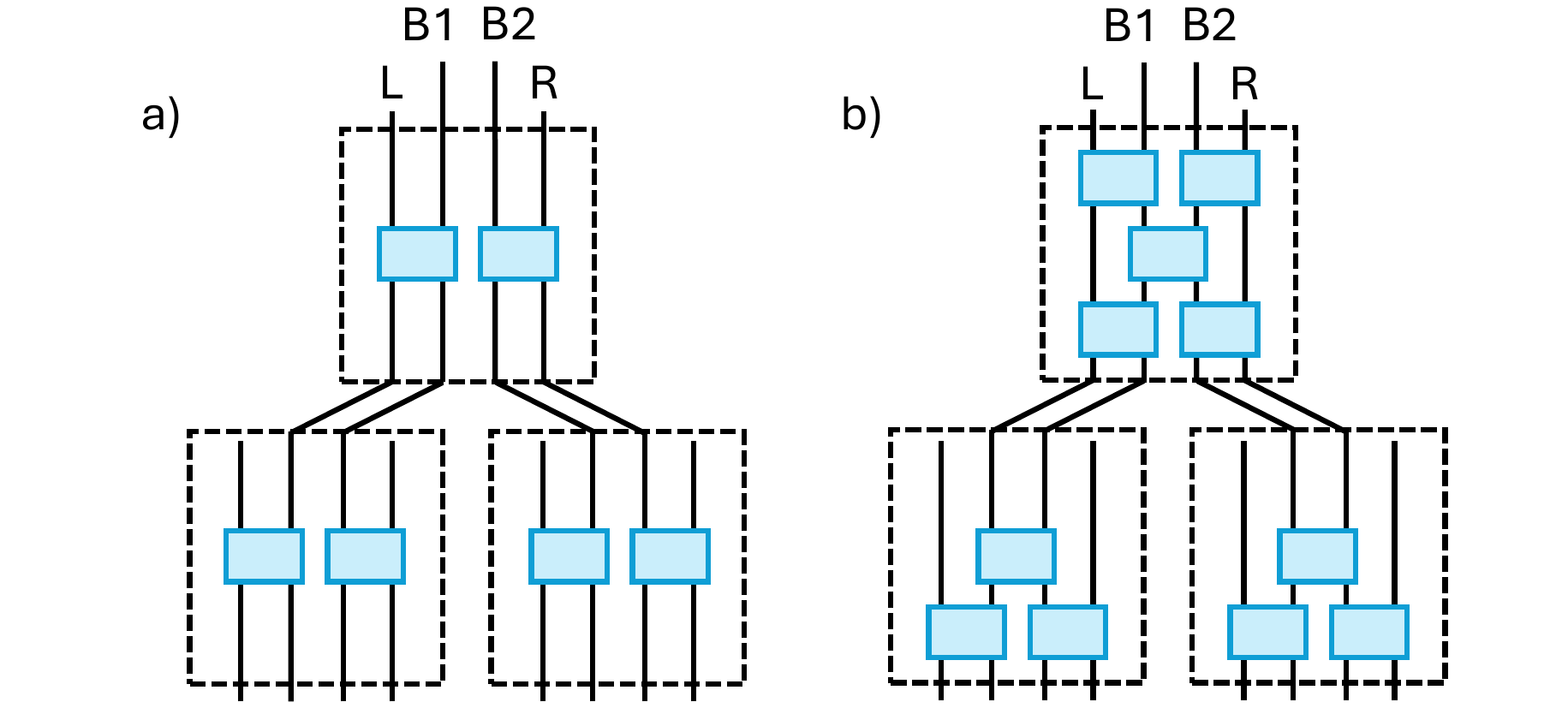}
    \caption{\textbf{Switching layouts for optical routers for a) FxdMux and b) DynMux.} Inputs of the router are at the bottom and outputs at the top. Outputs labeled L and R connect to left and right end nodes, respectively, and outputs B1 and B2 connect to the two arms of the local entanglement generation interface. Optical switches are represented by the blue rectangles. Unit cells indicated by the dashed rectangles fan out into a binary tree. At each stage, the center two outputs of a unit cell are connected to the two left most or two right most inputs of the unit cell in the stage above it. See also Assumption \ref{as:policies}.
    }
    \label{fig:routers}
\end{figure}

Recall that entanglement swaps are performed asap so that at the end of a time step, unmatched links remain between the repeater and only one of the two end nodes. The DynMux policy aims to reduce the time before a link in queue is matched. 
\begin{definition}[DynMux]\label{def: dynmux}
    If at the end of a time step an entangled link between the repeater node and one of the end nodes remains unmatched, then during the next entanglement generation phase all free quantum chips on the repeater node attempt entanglement generation with a quantum chip on the other end node. Otherwise, $m$ quantum chips attempt remote entanglement generation to the left and the rest to the right.
\end{definition}
Dynamic multiplexing can be achieved with a router of depth $d^\mathrm{dyn}(m)=2\log_2(m)+1$, as shown for $\log_2(m)=2$ in \cref{fig:routers}b. The unit cell at the root of the binary tree has depth three and can provide any permutation of its inputs onto its outputs. The unit cells in the fanout stages have depth two. They can connect any two inputs simultaneously to their central two outputs. Any input of the router can thus be mapped onto any output, so that any chip can entangle with either end node. Importantly, any pair of inputs can be simultaneously connected to the local entanglement generation interface, so that any pair of chips can establish a local link for an entanglement swap.
A possible sequence of entanglement generation attempts under DynMux is shown in \cref{fig:system overview}d.

\item \label{assumption: time scales} \textbf{Time scales}

Assumptions on the time scales associated with the different operations in the time step are as follows. 

    \item[] \textit{Classical communication:} Classical communication between the repeater node and the end nodes happens during the synchronization and the generation phase. We denote the combined duration of the synchronization and generation phase by $t_\mathrm{long}$. For an end-to-end separation on the order of a few hundred kilometers it takes on the order of milliseconds for a classical message to pass between a pair of nodes. 
    \item[] \textit{Qubit initialization:} The communication qubits of all quantum chips on the repeater node can be initialized in parallel for entanglement generation attempts. For color centers and trapped ions, initialization takes on the order of a few microseconds \cite{pompili2021realization, bradley2022robust, stephenson2020high}. For local entanglement generation attempts between quantum chips on the repeater node, the communication qubits have to be re-initialized after every attempt. The duration of a local attempt is denoted $t_\mathrm{local}$. 
    \item[] \textit{Fast switches:} We assume that the switches in the optical router can be rapidly reconfigured so that the time it takes to reconfigure the connections in the router is negligible compared to $t_\mathrm{long}$. This is possible with MEMS \cite{errando2019mems}, electro-optical \cite{hu2025integrated}, or piezoelectric \cite{tian2024piezoelectric} switches, which can achieve (sub-)microsecond reconfiguration times.
    \item[] \textit{Fast two-qubit operations:} We assume that the time of a two-qubit gate between the communication and memory qubit in a quantum chip is short compared to $t_\mathrm{long}$. This means there is no significant downtime of quantum chips when doing a \textsc{swap} gate or deterministic Bell state measurement. Two-qubit gate times in the order of a few tens of microseconds can be achieved in color centers when the nuclear spin of the defect is used as the memory qubit \cite{stas2022robust}. 
    \item[] \textit{Entanglement swap:} The expected duration of an entanglement swap is $t_\mathrm{swap}=\EE(N_\mathrm{local})t_\mathrm{local}$, where $N_\mathrm{local}\sim \TruncGeo(p_\mathrm{local},c_\mathrm{local})$ is the number of local entanglement generation attempts during an entanglement swap and we ignore the time to perform the deterministic Bell state measurement after generating the local link. The expected duration of an entanglement swap is at most on the order of a millisecond, since $t_\mathrm{swap} \leq c_\mathrm{local}t_\mathrm{local}$, we use cutoffs $c_\mathrm{local}\leq 10^3$ attempts (see Appendix \ref{app: parameters}), and $t_\mathrm{local}$ is on the order of a microsecond.

\item \label{} \textbf{Depolarizing noise}

The ideal target state for an end-to-end link produced by the repeater is a maximally entangled Bell state, which for concreteness we assume to be the state $\ket{\Psi^{-}}=(\ket{01}-\ket{10})/\sqrt{2}$. However, noise at several stages of the repeater chain protocol leads to an end-to-end state that is different from the ideal target. The main source of noise that we are concerned with here is time-dependent decoherence of entangled states stored in the memory qubits of the quantum chips. 
The dominant sources of noise are typically dephasing or amplitude damping noise, both of which increase exponentially over time \cite{azuma2023quantum}.
A tractable noise model that captures the exponential increase of the noise with time is the depolarizing noise model. 
A depolarizing channel $D_\lambda$ with depolarizing parameter $\lambda$ acts on a two-qubit state $\rho\in (\CC^{2})^{\otimes 2}$ as
\begin{equation}
    D_\lambda(\rho) = \lambda \rho + (1-\lambda) \II_4/4.
\end{equation}
We assume that the only noise acting on the qubits is depolarizing noise.
Any entangled link is then given by a Werner state 
\begin{equation}\label{eq: werner state}
    \rho_\omega=\omega\ketbra{\Psi^{-}} + (1-\omega)\II_4/4,
\end{equation}
where $\omega$ is the \textit{Werner parameter}. The Werner state form is preserved under depolarizing noise and entanglement swaps according to the following two multiplicative properties: 
\begin{itemize}
    \item Applying a depolarizing channel $D_\lambda$ with depolarization parameter $\lambda$ to a Werner state $\rho_\omega$ yields another Werner state, $D_\lambda(\rho_\omega)=\rho_{\lambda \omega}$.
    \item When an entanglement swap is performed between two Werner states $\rho_\omega$ and $\rho_{\omega'}$, the resulting state is the Werner state $\rho_{\omega\omega'}$.
\end{itemize} 
It was shown in Ref. \cite[Appendix C]{goodenough2025noise} that an analysis based on these two properties can in principle be extended to any Pauli noise. Furthermore, bounds on the accuracy of the Werner approximation for repeater chains have recently been found in \cite{davies2025accuracy}.

\item \label{} \textbf{Idle and active memory decoherence}

To model decoherence of a link stored in the memory qubit of a quantum chip we distinguish two cases. We say that a quantum chip is \emph{active} when the communication qubit is repeatedly being initialized and optically excited for local entanglement generation attempts. Otherwise the quantum chip is called \emph{idle}. The coherence time of the memory qubit is $t_\mathrm{coh-idle}$ when the quantum chip is idle. When the quantum chip is active the coherence time is expressed as $t_\mathrm{coh-active}=n_\mathrm{coh-active}t_\mathrm{local}$, where $n_\mathrm{coh-active}$ is the number of entanglement generation attempts that can be performed with the communication qubit before the state on the memory qubit has decayed by a factor $1/e$. 
The active coherence time is smaller than the idle coherence time because manipulating the communication qubit during entanglement generation attempts can introduce excess noise on the memory qubit \cite{bradley2022robust}. These coherence times lead to the following depolarization parameters for the different stages of the repeater chain protocol by connecting them to the relevant time scales:
\begin{itemize}
    \item $\lambda_\mathrm{long} := e^{-t_\mathrm{long}/t_\mathrm{coh\text{-}idle}}$ is the depolarization parameter for a link stored during a synchronization phase and entanglement generation phase when quantum chip is idle 
    \item $\lambda_\mathrm{local} := e^{-t_\mathrm{local}/t_\mathrm{coh\text{-}idle}}$ is the depolarization parameter for a link stored during one local generation cycle when quantum chip is idle.
    \item $\lambda_\mathrm{active} := e^{-1/n_\mathrm{coh\text{-}active}}$ is the depolarization parameter for a link stored during one local generation attempt when quantum chip is active.
\end{itemize}

We consider idle coherence times to be at least on the order of seconds (see Appendix \ref{app: parameters}). Since subsequent entanglement generation attempts can be performed at the optical switching time scale, which is submicrosecond (see \ref{assumption: time scales}), we ignore the effect of different storage times for links generated from different attempts in the pipelined sequence of all the attempts within an entanglement generation phase.

\item \label{} \textbf{Noise from imperfect operations:} Any end-to-end link that is produced suffers from imperfections in the heralded entanglement generation protocol and the subsequent gates and measurements, independent of how long the links before the entanglement swap had been stored in memory. 
All this static noise is combined into a single depolarization channel with parameter $\lambda_\mathrm{static}$. We set $\lambda_\mathrm{static}=1$ throughout since both policies suffer the same static noise.


\end{enumerate}

\section{Performance Analysis}\label{section: performance analysis}
Having introduced the model assumptions and policies, we now turn to the performance analysis of the multiplexed quantum repeater. We define the performance metrics in terms of long-term averages, 
and then show how these can be evaluated exactly by solving linear systems of size $O(m)$, where $m$ is half the number of chips on the repeater. We also provide an approximate Markov chain model for which the performance metrics can be given in closed form. These closed-form approximations to the performance metrics are close to the exact results in the near-term parameter regime considered for numerical evaluation in \cref{section: numerical evaluations}.

\subsection{Performance Metrics}\label{section: performance metrics}

The objective of the quantum repeater is to distribute highly entangled states at the best possible rate between the end nodes. The key metrics by which the performance of the repeater can be assessed are the rate at which it produces end-to-end states and the fidelity of the ensemble of end-to-end states to a maximally entangled Bell state. In the depolarizing noise model of \cref{section: system model}, the end-to-end states are fully characterized by their Werner parameter $\omega$ (see \cref{eq: werner state}). The fidelity $F$ of the end-to-end state to the desired Bell state is related to the Werner parameter as
\begin{equation}
    F = \frac{3\omega+1}{4}.
\end{equation}
We evaluate the long-term average rate and Werner parameter to capture the performance of the repeater when it delivers end-to-end states without interruption. 

To define the performance metrics, we first define the relevant quantities.
Here we denote and define quantities in a policy-agnostic way, which are further refined under FxdMux and DynMux policies in Appendix~\ref{app: full markov chain models}.

\textbf{The signed queue length process}:
Let $L_n, R_n\in \NN \cup \{0\}$ respectively denote the number of links between the repeater and the left end node and between the repeater and the right end node at the start of time step $n$. 
At the start of any time step there can only be entangled links connected to one of the end nodes because swaps are performed asap (see~\ref{as:timeStep}). This means that at the start of any time step there is a queue of entangled links between the repeater node and one of the end nodes. 
We define the \textit{signed queue length} $S_n:=R_n-L_n$. Under the FxdMux and DynMux policies defined in~\ref{as:policies}, the process $(S_n)_{n\geq 1}$ completely describes the evolution of the repeater since these policies only depend on which chips are free at the start of a time step. 
We further define the following quantities, which are crucial for defining the performance metrics.

\vspace{2pt}
\noindent
{\setlength{\tabcolsep}{6pt}
\begin{tabular}{ll}
$T_n$ & the duration of the $n$th time step \\
$U_n$ & the number of end-to-end links produced \\
      & in the $n$th time step \\
$V_n$ & the sum of the Werner parameters of the \\
      & end-to-end links produced in the $n$th time step \\
\end{tabular}}
We now formally define the performance metrics. 
%
%

\begin{definition}[Steady-state rate]\label{def: steady-state rate long time}
    The \emph{steady-state rate} $r$ is defined as
    \begin{equation}\label{eq: steady state rate definition}
        r := \lim_{N\to \infty} \frac{\sum_{n=1}^N U_n}{\sum_{n=1}^N T_n},
    \end{equation}
    when the limit on the RHS exists.
    We will show that the limit exists indeed and is almost surely a constant.
\end{definition}

In other words, the steady-state rate is the average number of end-to-end links produced per unit time if the repeater runs for a long time.

\begin{definition}[Steady-state Werner parameter]\label{def: steady-state werner long time}
    The \emph{steady-state Werner parameter} is defined as 
    \begin{equation}\label{eq:steadyStateWerner}
        w := \lim_{N\to \infty} \frac{\sum_{n=1}^N V_n}{\sum_{n=1}^N U_n},
    \end{equation}
    on $\{\sum_{n=1}^N U_n>0\}$.
    Like the steady-state rate, we will show that the limit on the RHS is well-defined and is almost surely a constant.
    Also, $\PP(\sum_{n=1}^N U_n>0) \to 1$ as $N \to \infty$.
\end{definition}

The physical interpretation of the steady-state Werner parameter is that if the repeater has been running for a long time, then any end-to-end link it generates appears to be a sample from an ensemble described by the Werner state $\rho_w$.



A metric that combines the rate and fidelity into a single number is the secret key rate, which is the rate at which two parties at the end nodes of the repeater chain can create bits of secret key through an entanglement-based quantum key distribution (QKD) protocol \cite{ekert1991quantum}. 
In Appendix \ref{app:skr}, we also show that the steady-state secret key rate $f_\mathrm{skr}$ can be expressed in terms of the steady-state rate $r$ and steady-state Werner parameter $w$ as 
\begin{equation}\label{eq: steady-state secret key rate main}
    f_\mathrm{skr} := \max\bigg(1-2h\bigg(\frac{1-w}{2}\bigg),0\bigg)r,
\end{equation}
where $h(x)=-x\log_2(x)-(1-x)\log_2(1-x)$ denotes the binary entropy function.
Note that the steady-state secret key rate only becomes positive for sufficiently large $w>w_\mathrm{QKD}\approx 0.780$. 

The secret key rate provides a meaningful tradeoff between steady-state rate and steady-state Werner parameter. In the multiplexed repeater
such a tradeoff is induced by the choice of the cutoff $c_\mathrm{local}$ on the maximum number of entanglement generation attempts that may be performed during an entanglement swap: a smaller cutoff improves the Werner parameter at the cost of a lower rate. The optimization of this tradeoff with respect to the secret key rate is considered in \cref{section: numerical evaluations}.



\subsection{Exact Analysis}\label{section: markov chains exact}
In this section we describe the evolution of the multiplexed quantum repeater.
Specifically, we show that the (signed) number of unmatched links on the repeater evolves as a Markov chain and completely describes the state of the repeater. 
This allows us to compute the steady-state rate and Werner parameter from Def.~\ref{def: steady-state rate long time} and Def.~\ref{def: steady-state werner long time}.
The details can be found in Appendices \ref{app: markov abstract} and \ref{app: full markov chain models}, where the former deals with the derivation in a policy-agnostic way and the latter provide specific details of the performance metrics under each policy.



We now provide an example to explain the evolution of the signed queue length process further.
We first consider the case $m=1$.
In this case, there is no difference between the FxdMux and DynMux policies. In both cases, the signed queue length process $(S_n)_{n\geq1}$ takes values in the state space
\begin{equation}
    \mathcal{S}=\{-1,0,1\}.
\end{equation}
Further, $(S_n)_{n\geq1} \sim \Markov(P)$, where the transition probability matrix $P=(P_{\ell,\ell'})_{\ell, \ell'\in \mathcal{S}}$ has elements
\begin{equation}
    P_{\ell,\ell'} = \PP\left(S_{n+1}=\ell'\mid S_n=\ell\right).
\end{equation}
The transition probabilities can be deduced from the different possible evolutions of the repeater over a single time step. For example, ${P_{0,1}=p(1-p)}$ is the probability of the event that starting from state $0$, where there are no links and hence there is one attempt to the left and one attempt to the right, only the attempt to the right succeeds. 
As another example, the probability $P_{0,0}=1-2p(1-p)$, which can be rewritten as $(1-p)^2+p^2$, describes the events that either no links are generated from the empty chain, or both a left and a right link are generated, which are then swapped, leaving no links on the repeater. 
The other probabilities follow by similar reasoning:
\begin{equation}
    P = 
    \begin{pmatrix}
        1-p & p & 0 \\
        p(1-p) & 1-2p(1-p) & p(1-p) \\
        0 & p & 1-p
    \end{pmatrix}.
\end{equation}

For $m>1$, the FxdMux and DynMux policies differ. The signed queue lengths $(S_n^\mathrm{fxd})_{n\geq1} \sim \Markov(P^\mathrm{fxd})$ and $(S_n^\mathrm{dyn})_{n\geq1} \sim \Markov(P^\mathrm{dyn})$ for the signed queue length under the FxdMux and DynMux policy, respectively, are given in Appendix \ref{app: full markov chain models}. They have finite state spaces of size $O(m)$, are irreducible, and have a stationary distribution. The stationary distribution can be computed exactly by solving a linear system with dimension equal to the size of the state space.

The steady-state rate and Werner parameter can be computed using the stationary distribution of the signed queue length. The essential lemma is an application of the ergodic theorem.


\begin{lemma}[Ergodic Theorem for Markov-modulated Processes]\label{lemma: main replacing random variable by its expectation in ergodic theorem}
    Let $(A_n)_{n\geq1}\sim \Markov(P)$ be an irreducible Markov chain on a finite state space $\mathcal{S'}$ that has a stationary distribution $\pi'$.
    Further, let $(B_n)_{n\geq1}$ be another process with $B_m|A_m\perp B_n|A_n$ for $m \ne n$ and $B_n|(A_n=a) \sim H_a$ for a known collection of discrete distributions $\{H_a\}_{a \in \mathcal{S}'}$.
    Then, $(B_n)_{n\geq1}$ has a limiting distribution and the corresponding expectation is given by
    \begin{equation}\label{eq:main YssMean}
    \mu_{B}= \sum_{a\in \mathcal{S}'}\pi_a' \EE(H_a),
    \end{equation}
     provided the expectation under $H_a$ satisfies $|\EE(H_a)|< \infty$ $\forall a \in \mathcal{S}'$.
    Further,
    \begin{equation}\label{eq: markov modulated mean}
    \lim_{N \to \infty}\frac{1}{N}\sum_{n=1}^N B_n \overset{\mathrm{a.s.}}{=}\mu_B.
    \end{equation}
\end{lemma}
\begin{proof}[Sketch of proof]
We first recall the well-known fact that, for an irreducible Markov chain on a countable state space, a stationary distribution, if it exists, is unique. We then define $C_n := (A_n, B_n)$ for $n \in \NN$, and show that $(C_n)_{n\geq1}$ is an irreducible Markov chain on the countable state space ${\mathcal{C} :=\{(a,b): a \in \mathcal{S}',b\in \mathcal{H}_a\}}$, where $\mathcal{H}_a$ denotes the support of $H_a$, $a \in \mathcal{S}'$, and its transition probabilities are given by
    \begin{multline}
        \PP \big(C_{n+1} = (a',b') \mid C_{n} = (a,b)\big) = \\
        \PP \big(A_{n+1} = a' \mid A_{n} = a\big) \PP_{H_{a'}}(b'),
    \end{multline}
    where $\PP_{H_{a'}}(b')$ denotes the probability of the singleton $\{b'\}$ under distribution $H_{a'}$. 
    The stationary distribution of $(C_n)_{n\geq 1}$ is $\rho_{(a,b)} = \pi_a' \PP_{H_a}(b)$.
    The limiting distribution of $(B_n)_{n \geq 1}$ is given by the marginal of $\rho$, i.e., $\sum_{a \in \mathcal{S}'}\pi_a' \PP_{H_a}(b)$, which implies that the corresponding expectation $\mu_B$ exists and is given by $\mu_B=\sum_{a\in \mathcal{S}'}\pi'_a \EE(H_a)$ whenever $|\EE(H_a)|< \infty$ $\forall a \in \mathcal{S}'$.
    For the long-term average of $(B_n)_{n \geq 1}$, we define $\phi(a,b)=b$ on $\mathcal{C}$ and apply the ergodic theorem ~\cite[Example 6.2.4]{durrett2019probability} to conclude that 
    \begin{align*}
        \lim_{N\to \infty} \frac{1}{N}\sum_{n=1}^N B_n = \lim_{N\to \infty}\frac{1}{N} \sum_{n=1}^N \phi(A_n,B_n)
        \overset{\mathrm{a.s.}}{=}\mu_B.
    \end{align*}
    See Appendix \ref{app: markov abstract} for the details.
\end{proof}

We apply Lemma \ref{lemma: main replacing random variable by its expectation in ergodic theorem} to $(U_n)_{n\geq 1}$ and $(T_n)_{n\geq 1}$ to get an expression for the steady-state rate in terms of the stationary distribution of the signed queue length process $(S_n)_{n\geq 1}$.
\begin{theorem}\label{theorem: steady-state rate ergodic}
    Let $(S_n)_{n\geq1}\sim \Markov(P)$ be the Markov chain denoting the signed queue length of the multiplexed repeater with $2m$ quantum chips on state space $\mathcal{S}$.
    The steady-state rate is then given by
    \begin{equation}\label{eq: steady-state rate steady-state average}
        r \overset{\mathrm{a.s.}}{=} \frac{\mu_U}{\mu_T},
    \end{equation}
    where $\mu_U$ and $\mu_T$ are given by \cref{eq: definition mu U} and \cref{eq: definition mu T}, respectively.
\end{theorem}
\begin{proof}[Sketch of proof]
    Starting from the definition of the steady-state rate in \cref{eq: steady state rate definition}, we multiply both numerator and denominator by $1/N$ and then apply \cref{eq: markov modulated mean} of Lemma \ref{lemma: main replacing random variable by its expectation in ergodic theorem} to both.
    See Appendix \ref{app: markov abstract} for the details.
\end{proof}

The steady-state expectations $\mu_U$ and $\mu_T$ can be computed using \cref{eq:main YssMean} since the distributions of $U_n|(S_n=\ell)$ and $T_n|(S_n=\ell)$ are known. They are given, respectively, as (see also \cref{eq: number of links RV ell} and \cref{eq: duration time step RV ell})
\begin{equation}
    U_n|(S_n=\ell) \overset{d}{=} \sum\nolimits_{i=1}^{E'_{\ell}} \indicator(N_{\mathrm{local},i}'\leq c_\mathrm{local}),
\end{equation}
and
\begin{equation}
    T_n|(S_n=\ell) \overset{d}{=} t_\mathrm{long} + \sum\nolimits_{i=1}^{E'_{\ell}} N_{\mathrm{local},i} t_\mathrm{local},
\end{equation}
where the random variable $E_\ell'$ is the number of entanglement swaps in a time step starting from signed queue length $\ell$, $N_{\mathrm{local},i}'\overset{\mathrm{iid}}{\sim} \Geo(p_\mathrm{local})$ are the (hypothetical) numbers of local entanglement generation attempts until success during a single entanglement swap, and $N_{\mathrm{local},i}\overset{\mathrm{iid}}{\sim} \TruncGeo(p_\mathrm{local}, c_\mathrm{local})$ are the actual numbers of local entanglement generation attempts during a single entanglement swap. The distribution of $E_\ell'$ can be deduced by considering the number of links generated with the left and right end nodes in state $\ell$, as is done for FxdMux in \cref{eq: E fxdmux} and in \cref{eq: E dynmux} for DynMux.

For the steady-state Werner parameter we apply the ergodic theorem for Markov-modulated processes to the process $(V_n)_{n\geq 1}$ as well. However, the distribution for the sum of Werner parameters produced in time step $n$ cannot be derived straightforwardly by conditioning on $S_n$. Therefore, we first show in Lemma \ref{lemma: equivalence of W and W fut} that the long-term average of $(V_n)_{n\geq 1}$ is equal to the long-term average of $(V_n^\mathrm{fut})_{n\geq 1}$, where $V_n^\mathrm{fut}$ is the sum of the Werner parameters produced in future time steps from the new links generated in time step $n$. For a detailed definition of $V_n^\mathrm{fut}$ we refer to Appendix \ref{section: steady state werner parameter}. The steady-state Werner parameter can then be given in terms of $\mu_{V^\mathrm{fut}}$ and $\mu_U$, where $\mu_{V^\mathrm{fut}}$ denotes the steady-state expectation of $(V_n^\mathrm{fut})_{n\geq 1}$.

\begin{theorem}\label{theorem: main steady-state werner fut ergodic}
    Let $(S_n)_{n\geq1}\sim \Markov(P)$ be the Markov chain denoting the signed queue length of the multiplexed repeater with $2m$ quantum chips on state space $\mathcal{S}$, and let $\pi$ be its stationary distribution. Then,
    \begin{equation}\label{eq: steady-state Werner steady-state average}
        w \overset{\mathrm{a.s.}}{=} \frac{\mu_{V^{\mathrm{fut}}}}{\mu_U},
    \end{equation}  where $\mu_U$ and $\mu_{V^\mathrm{fut}}$ are given by \cref{eq: definition mu U} and \cref{eq: definintion mu V fut} respectively.
\end{theorem}
\begin{proof}[Sketch of proof]
    The proof follows by application of Lemma \ref{lemma: main replacing random variable by its expectation in ergodic theorem} similar to the proof of \cref{theorem: steady-state rate ergodic}.
    See Appendix \ref{app: markov abstract}.
\end{proof}
The distribution of $V_n^\mathrm{fut}|(S_n=\ell)$ can be expressed by conditioning on the new links generated in time step $n$. The conditional expectations are given in Lemma \ref{lemma: condition expectation V fut}. Evaluating these expectations for FxdMux requires solving the linear system of dimension $m$ shown in Lemma \ref{lemma: Lambda a b fxdmux}. For DynMux it requires solving the $2m-2$ linear systems of dimensions $1, \dots, 2m-2$ given in Lemma \ref{lemma: Lambda a b dynmux}. The steady-state rate and Werner parameter can thus be computed exactly by solving linear systems of dimension $O(m)$.

We also prove a similar convergence result for the estimated quantum bit error rate (QBER) in Thm.~\ref{lemma:qberN} when an entanglement-based version of the BB84 protocol is executed using the multiplexed repeater.
This forms the basis of the secret key rate heuristic mentioned in \cref{eq: steady-state secret key rate main}.
We refer the reader to Appendix~\ref{sec:steadyStateQBER} for further details.


\subsection{One-Simultaneous-Success Approximation}\label{section: markov chain oss}
The above method takes into account any possible number of links generated in a time step. However, in a near-term network we expect that the remote entanglement generation probability $p$ is on the order of $10^{-5}$ \cite{stolk2024metropolitan, van2022entangling}, so that, even with a few hundred chips, less than $1$ link is generated per time step on average.
This motivates the definition of the One-Simultaneous-Success (OSS) approximation 
where at most one entangled link is generated per time step.
More precisely, the OSS approximation is obtained by approximating the transition probability matrix in the exact analysis to first order in $p$. 

Within the OSS approximation, the approximate steady-state rate $\widetilde r$ and Werner parameter $\widetilde w$ can be obtained in closed form, listed in \cref{tab:OSS approximations rate and werner}. These closed-form approximations to the steady-state rate and Werner parameter provide an efficient heuristic method for optimizing the secret key rate over the cutoff on the number of local entanglement generation attempts during a swap (see \cref{section: numerical evaluations} and Appendix \ref{section: secret key rate optimization}), and they are used in the interpretation of the numerical evaluations of the exact analysis in the near-term regime in \cref{section: discussion}. We outline the application of the OSS approximation to our exact Markov chain analysis below, and leave the details to Appendix \ref{app: oss markov chain models}. 

Starting from the exact transition matrix $P$ for the signed queue length, we obtain an approximate transition matrix $\widetilde P$ by retaining only the terms up to first order in $p$. This implements the OSS approximation for the evolution of the signed queue length. For FxdMux it leads to an irreducible Markov chain with stationary distribution on the same state space as for the exact Markov chain,
but the transition matrix is tridiagonal (see \cref{eq: oss fxd tpm}). A closed-form solution for the stationary distribution is given in \cref{eq: steady-state oss fxdmux}.

On the other hand, for DynMux we restrict the state space in the OSS approximation to $\mathcal{\widetilde S}^{\mathrm{dyn}}=\{-1,0,1\}$ since these are the only recurrent states when approximating the exact transition probability matrix to first order in $p$.
In particular, queues with $2$ or more links cannot be reached from $\{-1,0,1\}$ under the OSS approximation for DynMux since this requires multiple links to be generated in a single time step.

Approximations to $\mu_U$ and $\mu_T$ can be found by approximating the distribution of the number of entanglement swaps in state $\ell$. In state $0$ the probability of an entanglement swap in the OSS approximation is zero, since it requires at least two links to be generated. In a state $\ell\neq 0$, an entanglement swap occurs in the OSS approximation if and only if the absolute queue length decreases by $1$. Therefore, for a state $\ell\neq 0$ the approximate number of entanglement swaps is given by a Bernoulli random variable with success probability $\widetilde P_{|\ell|,|\ell|-1}$. 

Similarly, $\mu_{V^\mathrm{fut}}$ can be approximated by approximating the expectation of the sum of future Werner parameters $V_n^\mathrm{fut}|(S_n=\ell)$ conditional on the signed queue length. 
In the exact analysis, this expectation is computed by conditioning on the new links generated from state $\ell$. In the OSS approximation, the distribution of new links generated again reduces to a Bernoulli random variable that can be deduced from the transition matrix since a new link is generated if and only if the queue length increases by $1$. In Appendix \ref{app: oss markov chain models} it is shown how the OSS approximations are applied to the FxdMux and DynMux policy to arrive at the closed-form approximations to the steady-state rate and Werner parameter in \cref{tab:OSS approximations rate and werner}. In the near-term parameter regime considered in \cref{section: numerical evaluations} these closed-form results provide close approximations to the exact results.

\begin{table*}
\begin{tabular}{lll}
\hline\hline
\parbox[c]{3cm}{\vspace{4pt}\textbf{Policy}\vspace{4pt}} & \parbox[c]{6cm}{\vspace{4pt}\textbf{Steady-state rate OSS $(\widetilde r)$}\vspace{4pt}} & \parbox[c]{8cm}{\vspace{4pt}\textbf{Steady-state Werner parameter OSS $(\widetilde w)$}\vspace{4pt}} \\
\hline
\parbox[c]{3cm}{\vspace{6pt}FxdMux $(\mathrm{fxd})$\vspace{6pt}} & \parbox[c]{6cm}{\vspace{6pt}$\left[\frac{1+2\sigma(m)}{2\sigma(m)}\frac{t_\mathrm{long}}{mp} + t_\mathrm{swap}\right]^{-1}p_\mathrm{swap}$\vspace{6pt}} & \parbox[c]{8cm}{\vspace{6pt}$\displaystyle\sum_{\ell=0}^{m-1}\frac{1}{\sigma(m)}\frac{(m-1)!}{(m-\ell-1)!m^\ell}\left(\lambda_\mathrm{local}^{(1)}\right)^\ell\EE\left(\lambda_\mathrm{long}^{\widetilde D^\mathrm{fxd}}\right)^{\ell+1}\frac{\lambda_0}{p_\mathrm{swap}}
    $, where $\widetilde D^\mathrm{fxd}\sim \Geo(mp)$\vspace{6pt}} \\[6pt]
\hline
\parbox[c]{3cm}{\vspace{6pt}DynMux $(\mathrm{dyn})$\vspace{6pt}} & \parbox[c]{6cm}{\vspace{6pt}$\left[\dfrac{4m-1}{2m}\dfrac{t_\mathrm{long}}{(2m-1)p}+t_\mathrm{swap}\right]^{-1}p_\mathrm{swap}$\vspace{6pt}} & \parbox[c]{8cm}{\vspace{6pt}$ \EE\left(\lambda_\mathrm{long}^{\widetilde D^\mathrm{dyn}}\right)\frac{\lambda_0}{p_\mathrm{swap}}$, where $\widetilde D^\mathrm{dyn}\sim \Geo((2m-1)p)$\vspace{6pt}} \\[6pt]
\hline\hline
\end{tabular}
\caption{\textbf{Steady-state rate and Werner parameter for the OSS approximations of DynMux and FxdMux.} Here, $\sigma(m)\!=\!\sum_{\ell=1}^{m} m!/((m\!-\!\ell)!m^\ell)$ is a combinatorial factor, $\lambda_\mathrm{local}^{(1)}\!=\!\EE\left(\lambda_\mathrm{local}^{N_\mathrm{local}}\right)$ with ${N_\mathrm{local}\sim \TruncGeo(p_\mathrm{local},c_\mathrm{local})}$ is the expected depolarization of a link stored in queue while another pair of links is performing an entanglement swap, and $\lambda_0=\lambda_\mathrm{static}\,\EE\left(\lambda_\mathrm{active}^{2N'_\mathrm{local}}\mathbbm{1}(N'_\mathrm{local} \!\le\! c_\text{local})\right)$ with $N'_\mathrm{local}\sim \Geo(p_\mathrm{local})$ is the maximum expected contribution of any pair of links to sum of Werner parameters, i.e. the contribution without depolarizing factors due to idling in memory. The other quantities have been introduced earlier and we refer to Appendix \ref{app: parameters} for their definitions. }
\label{tab:OSS approximations rate and werner}
\end{table*}

\section{Numerical Evaluations}\label{section: numerical evaluations}
We numerically evaluate the steady-state rate and steady-state Werner parameter in the near-term regime. This means that we consider parameter values that assume improvements over the values from experiments with color centers and trapped ions listed in \cref{tab:fundamental-parameters}. Since the static noise from imperfect operations is independent of the entanglement generation policy, we do not take it into account in our analysis by setting $\lambda_\mathrm{static}=1$ throughout.

Recall that the dynamic multiplexing policy directly changes the remote entanglement generation process compared to fixed multiplexing, while entanglement swapping through local entanglement generation is the same in both policies. To isolate the effect of dynamic multiplexing we therefore first consider ideal local entanglement generation, in the sense that $p_\mathrm{local}=1$ and $n_\mathrm{coh-active}=\infty$ so that the local link is generated deterministically and no additional time-dependent decoherence is incurred during local entanglement generation. The results can be interpreted using the closed-form approximations in \cref{tab:OSS approximations rate and werner}, because the OSS approximation yields very similar results to the exact analysis for the parameter regime under consideration. 

We then discuss the effect of imperfections in the local entanglement generation process. We obtain approximations to the maximum achievable secret key rate after optimization over the cutoff $c_\mathrm{local}$ on the number of local entanglement generation attempts during an entanglement swap. This analysis includes switching losses in the optical router.

\subsection{Ideal Local Entanglement Generation}
\begin{figure*}[htbp]
    \centering
    \begin{subfigure}[t]{0.48\textwidth}
        \centering
        \includegraphics[width=\textwidth]{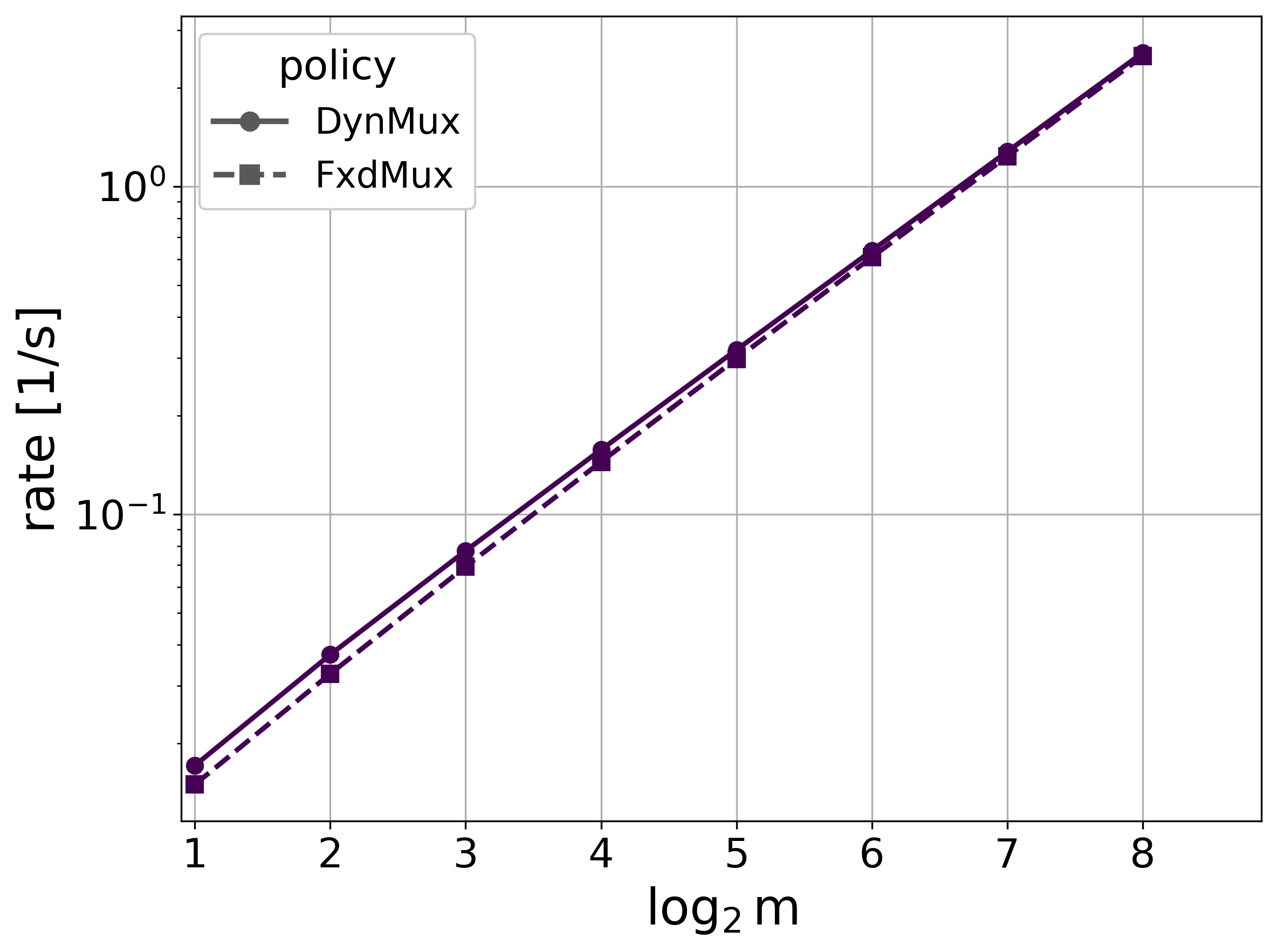}
        \caption{}
        \label{fig:IL2_rate}
    \end{subfigure}
    \hfill
    \begin{subfigure}[t]{0.48\textwidth}
        \centering
        \includegraphics[width=\textwidth]{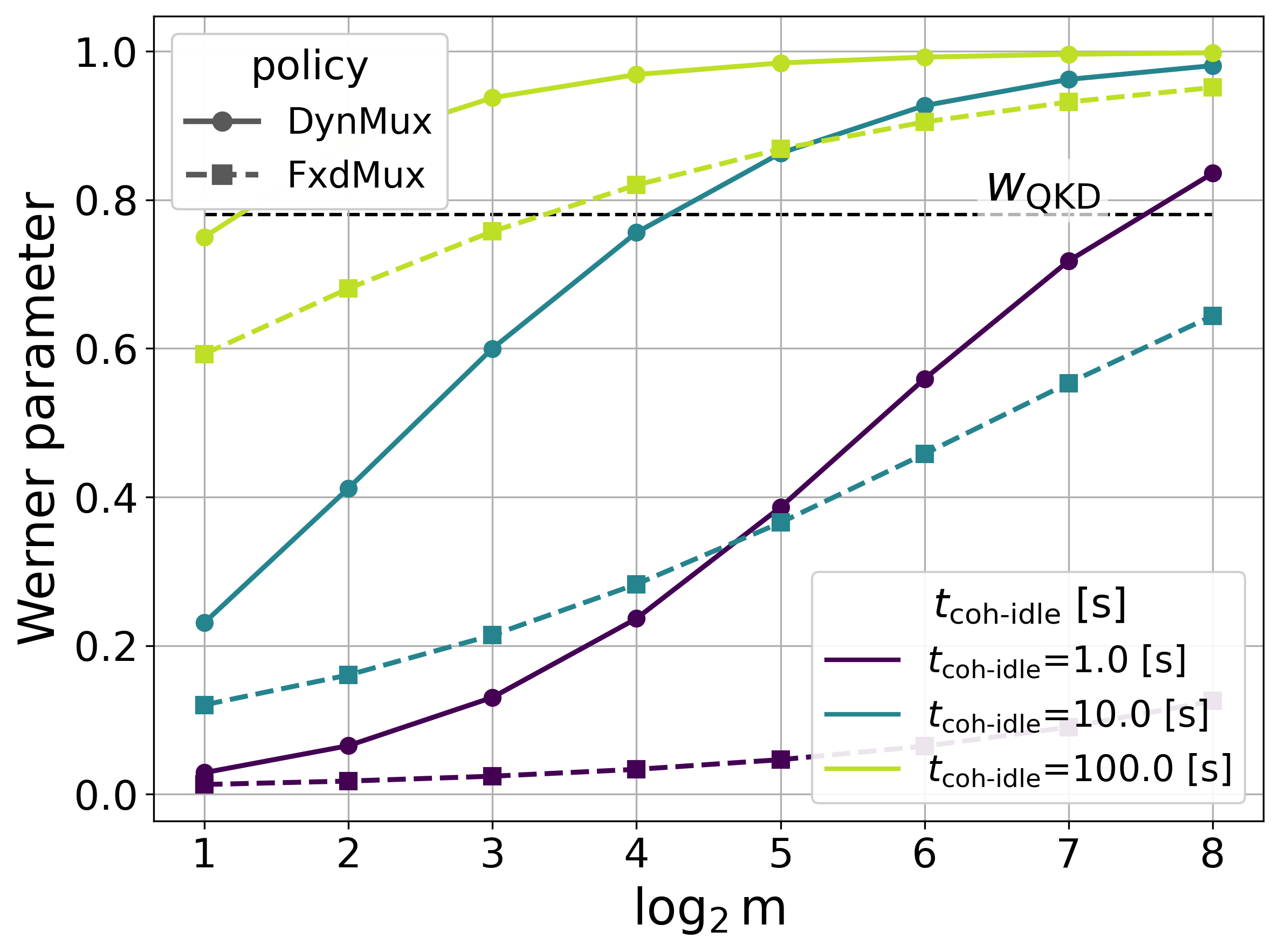}
        \caption{}
        \label{fig:IL2_werner}
    \end{subfigure}
    \caption{\textbf{Dynamic multiplexing can significantly improve the fidelity of end-to-end links generated in steady-state, while also yielding a marginally higher rate.} \textbf{(a)} The steady-state rate and \textbf{(b)} steady-state Werner parameter (which is a measure for the fidelity) for dynamic multiplexing (circles connected by solid line) are compared to those for fixed multiplexing (squares connected by dashed line) for increasing numbers of chips on the repeater node in the range $2m=4, \dots, 512$ ($\log_2(m)=1, \dots, 8$) and coherence times $t_\mathrm{coh-idle}=\SI{1}{\second},\SI{10}{\second},\SI{100}{\second}$. Other parameter values have been set to $p=10^{-5}$, $t_\mathrm{long}=\SI{1}{\milli\second}$, and $t_\mathrm{local}=\SI{1}{\micro\second}$. To isolate the effect of dynamic multiplexing on reducing decoherence due to links waiting in memory until they are swapped we have assumed ideal local entanglement generation by setting $p_\mathrm{local}=1$ and $n_\mathrm{coh-active}=\infty$, and no static noise by setting $\lambda_\mathrm{static}=1$. The threshold $w_\mathrm{QKD}\approx 0.780$ for positive steady-state secret key rate is indicated by the dashed horizontal line. Results shown are computed from the exact analysis in \cref{section: markov chains exact}, but the results from the OSS approximation completely overlap (see \cref{fig:IL2_validation_full_sim} in Appendix \ref{app: validation simulation}). }
    \label{fig:IL2_rate_werner}
\end{figure*}

In \cref{fig:IL2_rate_werner} the steady-state rate and Werner parameter under the DynMux and FxdMux policies are shown for remote entanglement generation probability $p=10^{-5}$ and $2m=4, \dots, 512$ ($\log_2(m)=1, \dots, 8$) 
chips on the repeater node. Idle coherence times ${t_\mathrm{coh-idle}\!=\!\SI{1}{\second}, \SI{10}{\second}, \SI{100}{\second}}$ are considered at a Sync-Gen duration of ${t_\mathrm{long}\!=\!\SI{1}{\milli\second}}$ and local entanglement generation attempt duration $t_\mathrm{local}=\SI{1}{\micro\second}$.
The time $t_\mathrm{long}$ corresponds to an end-to-end separation on the order of $\SI{100}{\kilo\meter}$, the values for $p$ are at the experimental state-of-the-art on the order of $ \SI{10}{\kilo\meter}$ separation \cite{knaut2024entanglement, stolk2024metropolitan, liu2024creation, liu2026long}, and idle coherence times on the order of a few seconds have been demonstrated for example in \cite{bradley2022robust, stas2022robust}. 
Local entanglement generation has been assumed to be ideal, i.e. we have set $p_\mathrm{local}=1$ and $n_\mathrm{coh-active}=\infty$, but as we discuss in \cref{section: non ideal local} the results are qualitatively similar when there are non-idealities in local entanglement generation. \cref{fig:IL2_rate_werner} only shows the exact analysis, but the OSS approximation completely overlaps with the exact analysis in this case (see \cref{fig:IL2_validation_full_sim} in Appendix \ref{app: validation simulation}).

\cref{fig:IL2_rate} shows that the rate for DynMux is marginally higher than the rate for FxdMux. The rates seemingly converge with increasing $m$, consistent with the findings of the prior study of dynamic multiplexing in~\cite{lee2022quantum}. 

On the other hand, \cref{fig:IL2_werner} shows that dynamic multiplexing significantly improves the steady-state Werner parameter compared to fixed multiplexing. Moreover, the improvement of dynamic multiplexing over fixed multiplexing is greatest for the lower end of the range of idle memory coherence times considered. For example, at $t_\mathrm{coh-idle}=\SI{10}{\second}$ and $\log_2m=8$, which corresponds to a total of $2m=512$ quantum chips on the repeater node, the Werner parameter for DynMux is close to $1$, while for FxdMux it is still around $0.6$. This means that, for these parameters, with $2m=512$ chips, DynMux has gotten rid of almost all the decoherence due to links waiting to be matched, while for FxdMux this decoherence still puts its Werner parameter below the QKD threshold $w_\mathrm{QKD}$. 

\subsection{Non-Ideal Local Entanglement Generation}\label{section: non ideal local}
The results of \cref{fig:IL2_rate_werner} assume ideal local entanglement generation, meaning $p_\mathrm{local}=1$ and $n_\mathrm{coh-active}=\infty$. The effects of local imperfections $p_\mathrm{local}<1$ and $n_\mathrm{coh-active}<\infty$ predominantly serve to rescale the rate and the Werner parameter. This can be seen in the OSS approximation from \cref{tab:OSS approximations rate and werner} where $p_\mathrm{local}$ and $n_\mathrm{coh-active}$ only impact the parameters $p_\mathrm{swap}$, $\lambda_0$, and $\lambda_\mathrm{local}^{(1)}$. The first two of these are overall scalings on the steady-state rate and Werner parameter, respectively, while contributions from the last are small when the idle coherence time is much larger than the time to perform a swap. Even when it is not, $\lambda_\mathrm{local}^{(1)}<1$ only further degrades the steady-state Werner parameter of FxdMux while not affecting DynMux. The comparison of dynamic multiplexing to fixed multiplexing shown in \cref{fig:IL2_rate_werner} for ideal local operations is therefore qualitatively similar with local imperfections. 

\begin{figure}
    \centering
    \includegraphics[width=\linewidth]{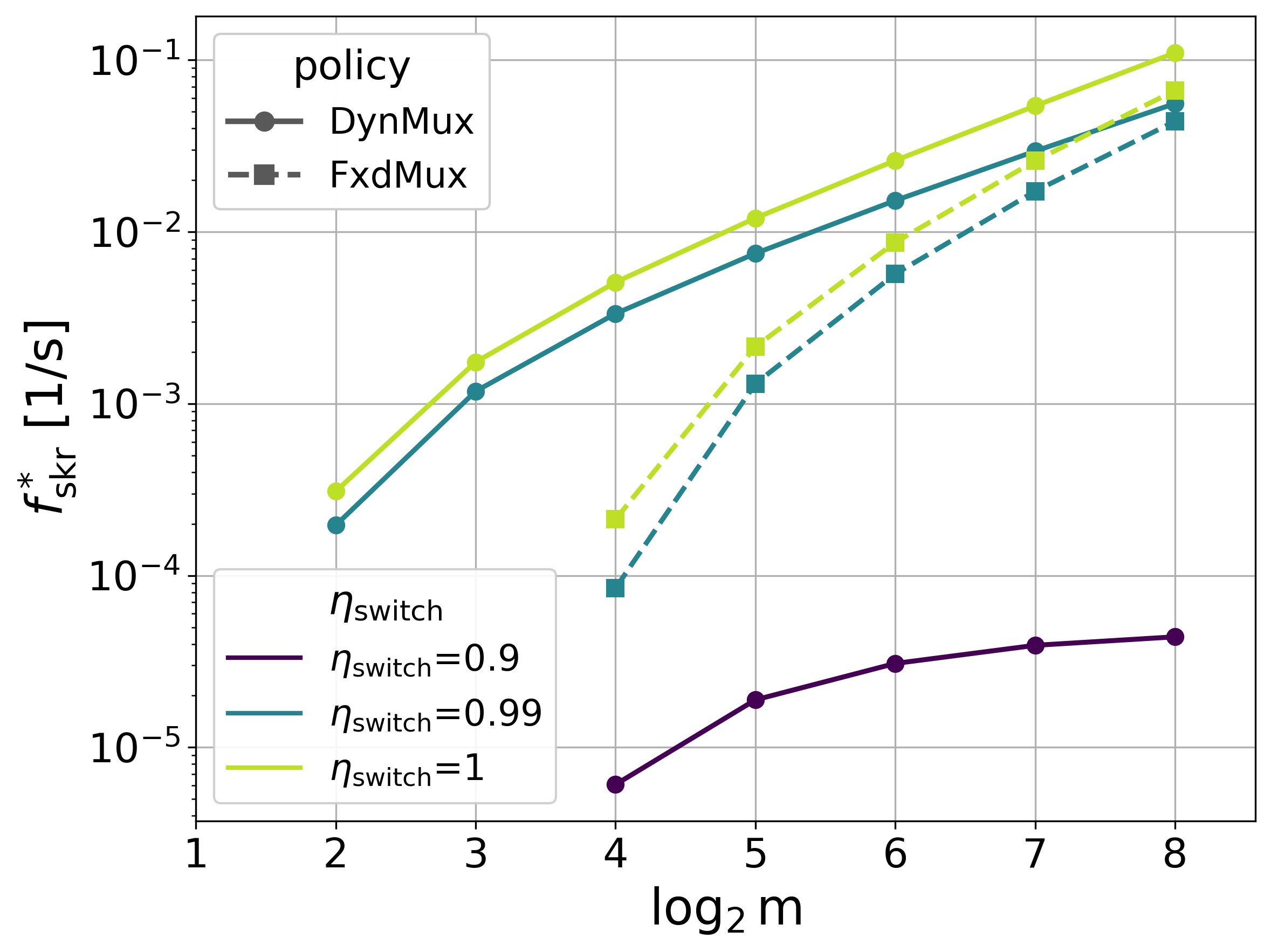}
    \caption{\textbf{Heuristic estimates to the maximum achievable secret key rate with non-ideal local entanglement generation and lossy switches.} The heuristic estimate $f^*_\mathrm{skr}$ of the maximum achievable secret key rate for DynMux and FxdMux is compared for increasing number of chips on the repeater node in the range $2m=4,\dots,512$ ($\log_2(m)=1,\dots,8$) and switching efficiencies $\eta_\mathrm{switch}\in \{0.9,0.99, 1\}$, at a remote entanglement generation probability $p=10^{-5}$, idle coherence of $t_\mathrm{coh-idle}=\SI{100}{\second}$, and Sync-Gen cycle duration $t_\mathrm{long}=\SI{1}{\milli\second}$. Local entanglement generation is non-ideal with parameters $p_\mathrm{local}=10^{-3}$, $n_\mathrm{coh-active}=1000$, and $t_\mathrm{local}=\SI{1}{\micro\second}$. Points where $f^*_\mathrm{skr}=0$ are not shown. In particular, there are no positive values for $f^*_\mathrm{skr}$ for FxdMux at $\eta_\mathrm{switch}=0.9$.}
    \label{fig:NILN_eta_skr_lowerbound}
\end{figure}

\subsection{Switching Losses}
The advantage of dynamic multiplexing shown in \cref{fig:IL2_rate_werner} does not take into account switching losses in the optical router. The effect of switching losses can be incorporated by rescaling the remote and local entanglement generation probabilities to
\begin{align}
     &p \to p(m):=p\,\eta_\mathrm{switch}^{2d(m)}, \label{eq:decrease remote enta prob}\\
     &p_\mathrm{local}\to p_\mathrm{local} (m):=p_\mathrm{local}\,\eta_\mathrm{switch}^{2d(m)}, \label{eq: decrease local enta prob}
\end{align}
where $d(m)$ is the depth of the router and $\eta_\mathrm{switch}\in [0,1]$ is the efficiency of an individual optical switch (\ref{ass:router}). The factors of two on the depth in the exponents of $\eta_\mathrm{switch}$ are due to the fact that in an entanglement generation attempt the photons from both chips have to pass through a router of depth $d(m)$. Recall from \ref{as:policies} that dynamic multiplexing requires a router of depth $d^\mathrm{dyn}(m)=2\log_2(m)+1$, while fixed multiplexing can be done with a shallower router of depth $d^\mathrm{fxd}(m) = \log_2(m)$. This means that the entanglement generation probabilities for dynamic multiplexing decrease faster with the number of chips than for fixed multiplexing.

To study the effect of this quantitatively, we consider the secret key rate heuristic introduced in \cref{eq: steady-state secret key rate main}. We consider the same near-term parameter regime as in \cref{fig:IL2_rate_werner}, but also take into account local entanglement generation imperfections $p_\mathrm{local}<1$ and $n_\mathrm{coh-active}<\infty$. However, with local imperfections there is a rate-fidelity tradeoff in the choice of local cutoff $c_\mathrm{local}$. Optimizing this tradeoff is a discrete optimization problem. Direct optimization from the exact analysis is computationally expensive. Instead, since the OSS approximation agrees closely with the exact analysis in the parameter regime considered, we approximate the optimal value of the cutoff as $\widetilde c_\mathrm{local}^*$, which is the optimal cutoff in the OSS approximation. Subsequently, the maximum achievable secret key rate is approximated as
\begin{equation}
    f_\mathrm{skr}^* := f_\mathrm{skr}(r(\widetilde c_\mathrm{local}^*), w(\widetilde c_\mathrm{local}^*)),
\end{equation}
which is a heuristic estimate of the secret key rate given the exact steady-state rate $r$ and Werner parameter $w$ for cutoff $\widetilde c_\mathrm{local}^*$. This optimization procedure is described in detail in Appendix \ref{section: secret key rate optimization}. Optimization by enumeration is possible for active coherence times in the regime $5\leq n_\mathrm{coh-active}\leq 1/p_\mathrm{local}$ because then Proposition \ref{proposition: range for n_coh_active to guarantee no key at n is c} implies that the steady-state Werner parameter is below the QKD threshold $w_\mathrm{QKD}$ for $c_\mathrm{local}>n_\mathrm{coh-active}$.

\cref{fig:NILN_eta_skr_lowerbound} shows the heuristic estimates for the maximum achievable secret key rates under dynamic and fixed multiplexing after optimizing over the cutoff $c_\mathrm{local}$ for the same parameter values $p=10^{-5}$, $t_\mathrm{long}=\SI{1}{\milli\second}$ as in \cref{fig:IL2_rate_werner}. We restrict to $t_\mathrm{coh-idle}=\SI{100}{\second}$ since for this coherence time both DynMux and FxdMux have Werner parameter above the SKR threshold for values of $\log_2(m)$ in the considered range. Additionally, there are local imperfections $p_\mathrm{local}=10^{-3}$ and $n_\mathrm{coh-active}=10^{3}$. These numbers are relevant to near-term color-center and trapped-ion hardware, as discussed in Appendix \ref{app: parameters}. We consider switching efficiencies in the range $\eta_\mathrm{switch}\in \{0.9,0.99,1\}$, where $\eta_\mathrm{switch}=1$ means lossless switches. A switching efficiency $\eta_\mathrm{switch}=0.9$ is within reach for MEMS switches \cite{errando2019mems}. For these parameters \cref{fig:NILN_eta_skr_lowerbound} shows that dynamic multiplexing yields higher estimates for the maximum achievable secret key rate than fixed multiplexing. 

With lossless switches, we see that for dynamic multiplexing the estimated maximum achievable secret key rate $f^{*,\mathrm{dyn}}_{\mathrm{skr}}$ is nonzero for $\log_2(m)\geq 2$, while with fixed multiplexing  $f^{*,\mathrm{fxd}}_\mathrm{skr}$ only becomes nonzero for $\log_2(m)\geq 4$. Moreover, the ratio $f_\mathrm{skr}^{*,\mathrm{dyn}}/f_\mathrm{skr}^{*,\mathrm{fxd}}$ of the estimates for DynMux and FxdMux is largest for smaller numbers of chips. It ranges from a factor $23.9$ at $\log_2(m)=4$ to a factor of $1.65$ at $\log_2(m)=8$. 
Similar results are seen for a lower switching efficiency of $\eta_\mathrm{switch}=0.99$, but the ratio $f_\mathrm{skr}^{*,\mathrm{dyn}}/f_\mathrm{skr}^{*,\mathrm{fxd}}$ decreases faster with $\log_2(m)$ than with perfect switches. At $\log_2(m)=4$ the ratio is $39.6$, while at $\log_2(m)=8$ it is $1.27$. 

Finally, for a switching efficiency of $\eta_\mathrm{switch}=0.9$, it is found that $f^{*,\mathrm{fxd}}_{\mathrm{skr}}$ is zero over the range of number of chips considered, while for dynamic multiplexing the estimated maximum achievable secret key rate is still positive for $\log_2(m)\geq 4$. This means that even though dynamic multiplexing requires an optical router that is twice as deep as for fixed multiplexing, its estimated maximum achievable secret key rate is still higher than for fixed multiplexing in the parameter regime considered.

\section{Discussion}\label{section: discussion}

In the near-term regime, the improvement of dynamic multiplexing over fixed multiplexing observed in \cref{section: numerical evaluations} can be explained informally using the OSS approximations $\widetilde w^\mathrm{dyn}$ and $\widetilde w^\mathrm{fxd}$ for the steady-state Werner parameter. In the following remark we do this for the case of perfect switches with $\eta_\mathrm{switch}=1$.
\begin{remark}\label{remark: oss interpretation}
    Recall from \cref{tab:OSS approximations rate and werner} that the steady-state Werner parameter under dynamic multiplexing is given by $
        \widetilde w^\mathrm{dyn} = \EE\left(\lambda_\mathrm{long}^{\widetilde D^\mathrm{dyn}}\right)\lambda_0/p_\mathrm{swap}$,
    where 
    \begin{equation}
        \widetilde D^\mathrm{dyn} \sim \Geo((2m-1)p)
    \end{equation}
    is the number of time steps until the link in queue is matched. After generating a link with one end node, all remote entanglement generation attempts are performed with the other end node. Therefore, in the OSS approximation for dynamic multiplexing, there is at any time at most one link in queue. On the other hand, the steady-state Werner parameter for fixed multiplexing is of the form
    \begin{equation}
        \widetilde w^\mathrm{fxd}=\sum_{\ell=1}^{m}\alpha_\ell\left(\EE\left(\lambda_\mathrm{long}^{\widetilde D^\mathrm{fxd}}\right)\right)^{\ell}\lambda_0/p_\mathrm{swap},
    \end{equation} 
    where the coefficients $\alpha_\ell$ can be found in \cref{tab:OSS approximations rate and werner}, and
    \begin{equation}
        \widetilde D^\mathrm{fxd} \sim \Geo(mp)
    \end{equation}
    is the number of time steps until the first link in queue is matched. 
    We thus identify two mechanisms by which dynamic multiplexing improves the steady-state Werner parameter compared to fixed multiplexing:
    \begin{itemize}
        \item Dynamic multiplexing reduces the time until matching the link that is currently first in queue, as $\widetilde D^\mathrm{dyn}$ stochastically dominates $\widetilde D^\mathrm{fxd}$.
        \item Dynamic multiplexing avoids queuing, as in the OSS approximation for dynamic multiplexing there are no contributions from queue sizes beyond $\ell=1$, while in fixed multiplexing there are contributions from queue sizes $\ell>1$. 
    \end{itemize}
\end{remark}

A caveat to the above remark is that in the presence of switching inefficiencies, it may happen that fixed multiplexing ends up matching links from queue faster than dynamic multiplexing. This is actually the case for $\eta_\mathrm{switch}=0.9$ and all values of $\log_2(m)$ in \cref{fig:NILN_eta_skr_lowerbound}. In that case, $(2m-1)p^\mathrm{dyn}(m)<mp^\mathrm{fxd}(m)$, where $p^\mathrm{dyn}(m)=p\eta_\mathrm{switch}^{2d^\mathrm{dyn}(m)}$ and $p^\mathrm{fxd}(m) = p\eta_\mathrm{switch}^{2d^\mathrm{fxd}(m)}$ are the remote entanglement generation probabilities rescaled by switching inefficiencies as in \cref{eq:decrease remote enta prob}.
Still, we see in \cref{fig:NILN_eta_skr_lowerbound} that only dynamic multiplexing yields a positive estimate for the maximum achievable secret key rate. This is because fixed multiplexing also has significant contributions from links further down the queue. Hence, even though it can match links from queue faster, the fact that some links have to wait for multiple links to be matched still causes the steady-state Werner parameter of fixed multiplexing to be below the QKD threshold $w_\mathrm{QKD}$.  

In the near-term regime, dynamic multiplexing thus improves over fixed multiplexing by matching links faster, as long as optical switching is sufficiently efficient, and, by avoiding the queueing process, which it does for any switching efficiency. 

\section{Outlook}\label{section: outlook}
As the number of chips on the repeater node increases, we have observed in \cref{fig:IL2_werner} that it is possible with many-quantum-chip multiplexing to counter the effect of decoherence due to links waiting idly in memory. 
However, even without memory decoherence, there are still the fidelity penalties due to imperfections in the heralded entanglement generation protocol and imperfections in gates, which have not been considered here. Future work may investigate how entanglement distillation can be used between many-quantum-chip multiplexed nodes to also get rid of these losses in fidelity.

Additionally, these losses in fidelity are highly dependent on the choice of quantum chip. In fact, in color centers in diamond, there is often a tradeoff between two-qubit gate speed and the active coherence properties of the memory qubit, depending on the type of memory qubit used \cite{stas2022robust, pompili2021realization,hermans2022qubit,bradley2022robust}. 
In this work we have made the assumption that gate performance, particularly speed, is preferred over coherence properties. This assumption was important to motivate our time slotted model where all operations happen sequentially so that entanglement generated during the entanglement generation phase could be used in the next entanglement swap phase. But if two-qubit gates times become longer than the communication times, then a strictly time slotted protocol in which the node has to wait for gates to finish may become inefficient. In future work, it may be investigated what the effect is of waiting for gates to finish, and how the potential negative effects of this can be mitigated by letting go of the strictly sequential operations assumed in this work.

Finally, we have presented here just a single dynamic policy for a single repeater node, with no guarantees on its optimality. It remains an open question what the optimal policy for quantum chip assignment is, especially for chains with more than a single repeater node. We hope that our present work can serve as a starting point for further analytical studies on quantum chip multiplexing for quantum networks and help guide its future development. 

\section*{Code availability}
The code for the numerical evaluations of the Markov chain models and the NetSquid simulation can be found in our GitLab repository \cite{grimbergen2026multiplexing}.

\section*{Data availability}
Data will be made available in the 4TU data repository \cite{grimbergen2026multiplexing_data}.

\bibliographystyle{apsrev4-2}
\bibliography{bib}

\section*{Acknowledgements}
J.G., S. K., and S.W. acknowledge support from NWO Vici Program under Grant VI.C.222.029.

\section*{Author Contributions}
J.G., C.B., and S.W. defined the multiplexed repeater model. J.G. and S.K. performed the analysis. M.v.H. wrote the NetSquid simulation. J.G. was the main writer of the manuscript. S.K. reviewed the manuscript extensively. S.W. provided active feedback at every stage of the project.

\newpage
\appendix
\newpage
\onecolumngrid
\section{Model Parameters}\label{app: parameters}
A summary of the model parameters introduced in \cref{section: system model} is provided in \cref{tab:fundamental-parameters} and \cref{tab:derived-quantities}. 
\cref{tab:fundamental-parameters} lists the elementary model parameters and provides an indication of their order of magnitude based on experiments with color centers and trapped ions. For the numerical evaluations in \cref{section: numerical evaluations} we consider parameter values corresponding to conceivable improvements on these experimental values. \cref{tab:derived-quantities} lists useful model parameters that can be derived from the elementary parameters. 

Here are a few observations on the elementary model parameters:
\begin{itemize}
    \item The value for the remote entanglement generation probability $p ~(\sim 10^{-6})$ comes from experiments with independently operated NV centers separated by $\sim\SI{10}{\kilo\meter}$ of deployed fibre \cite{stolk2024metropolitan}. The fidelity of the entanglement generated in this experiment was $F\sim 0.55$. However, Ref. \cite{stolk2024metropolitan} projects that with SnV centers fidelities $>0.9$ could be achieved. Such fidelities have already been achieved with trapped ions over $\sim \SI{10}{\kilo\meter}$ of spooled fibre \cite{liu2026long}. 
    \item The value of $t_\mathrm{long}~\sim \SI{1}{\milli\second}$ assumes an end-to-end separation on the order of $\sim \SI{100}{\kilo\meter}$. 
    \item The local entanglement generation probability $p_\text{local}~(\sim 10^{-4})$ is generally expected to be a few orders of magnitude higher than the remote entanglement generation probability because there is no need for quantum frequency conversion, and coupling into optical fibre can potentially be avoided all together.
    \item The active memory coherence depends strongly on the type of quantum chip. For color-center qubits, for example, long active coherence time typically also means slow two-qubit gates because both depend on the strength of the hyperfine coupling between the communication and memory qubit \cite{bradley2022robust}. The model of \cref{section: system model} assumes fast gates. In numerical evaluations, we therefore consider the lower end of the active coherence times listed in \cref{tab:fundamental-parameters}. In particular, we consider active coherence times up to $n_\mathrm{coh-active}=10^3$.
\end{itemize}

\begin{table}[ht]
\caption{\textbf{Elementary model parameters.} }
\begin{tabular}{lll}
\hline\hline
\parbox[c]{3cm}{\vspace{4pt}\textbf{Symbol}\vspace{4pt}} & \parbox[c]{8cm}{\vspace{4pt}\textbf{Description}\vspace{4pt}} & \parbox[c]{3cm}{\vspace{4pt}\textbf{Order of magnitude}\vspace{4pt}} \\
\hline
\parbox[c]{3cm}{\vspace{6pt}$2m$\vspace{6pt}} & \parbox[c]{8cm}{\vspace{6pt}Number of quantum chips on repeater node\vspace{6pt}} & \parbox[c]{3cm}{\vspace{6pt}-\vspace{6pt}} \\[6pt]
\hline
\parbox[c]{3cm}{\vspace{6pt}$p$\vspace{6pt}} & \parbox[c]{8cm}{\vspace{6pt}Success probability of each remote entanglement generation attempt\vspace{6pt}} & \parbox[c]{3cm}{\vspace{6pt}$\sim 10^{-6}$ \cite{stolk2024metropolitan}\vspace{6pt}} \\[6pt]
\hline
\parbox[c]{3cm}{\vspace{6pt}$t_\mathrm{long}$\vspace{6pt}} & \parbox[c]{8cm}{\vspace{6pt}Duration of each remote entanglement generation attempt, i.e. one Sync-Gen cycle, at end-to-end distance of $\sim \SI{100}{\kilo\meter}$ \vspace{6pt}} & \parbox[c]{3cm}{\vspace{6pt}$\sim \SI{1}{\milli\second}$\vspace{6pt}} \\[6pt]
\hline
\parbox[c]{3cm}{\vspace{6pt}$t_\mathrm{coh\text{-}idle}$\vspace{6pt}} & \parbox[c]{8cm}{\vspace{6pt}Memory coherence time when quantum chip is idle\vspace{6pt}} & \parbox[c]{3cm}{\vspace{6pt}$\sim \SI{1}{\second}$ \cite{stas2022robust, bradley2022robust}\vspace{6pt}} \\[6pt]
\hline
\parbox[c]{3cm}{\vspace{6pt}$p_\mathrm{local}$\vspace{6pt}} & \parbox[c]{8cm}{\vspace{6pt}Success probability of each local entanglement generation attempt\vspace{6pt}} & \parbox[c]{3cm}{\vspace{6pt}$\sim 10^{-4}$ \cite{stephenson2020high} \vspace{6pt}} \\[6pt]
\hline
\parbox[c]{3cm}{\vspace{6pt}$t_\mathrm{local}$\vspace{6pt}} & \parbox[c]{8cm}{\vspace{6pt}Duration of each local entanglement generation attempt\vspace{6pt}} & \parbox[c]{3cm}{\vspace{6pt}$\sim \SI{1}{\micro\second}$ \cite{stephenson2020high}\vspace{6pt}} \\[6pt]
\hline
\parbox[c]{3cm}{\vspace{6pt}$\eta_\mathrm{switch}$\vspace{6pt}} & \parbox[c]{8cm}{\vspace{6pt}Efficiency of a single optical switch \vspace{6pt}} & \parbox[c]{3cm}{\vspace{6pt}$\sim 0.9$ \cite{errando2019mems}\vspace{6pt}} \\[6pt]
\hline
\parbox[c]{3cm}{\vspace{6pt}$n_\mathrm{coh\text{-}active}$\vspace{6pt}} & \parbox[c]{8cm}{\vspace{6pt}Number of entanglement generation attempts on communication qubit before memory qubit decoheres by factor $1/e$\vspace{6pt}} & \parbox[c]{3cm}{\vspace{6pt}$\sim 10$ \cite{stas2022robust} \\
$\sim 10^3$ \cite{hermans2022qubit}
\\ $\sim 10^{5}$ \cite{bradley2022robust} \vspace{6pt}} \\[6pt]
\hline
\parbox[c]{3cm}{\vspace{6pt}$\lambda_\mathrm{static}$\vspace{6pt}} & \parbox[c]{8cm}{\vspace{6pt}Time-independent decoherence parameter for each \textit{end-to-end} link; includes initial link decoherence upon generation and gate errors\vspace{6pt}} & \parbox[c]{3cm}{\vspace{6pt}$\sim 0.5$ \cite{stolk2024metropolitan} \vspace{6pt}} \\[6pt]
\hline
\parbox[c]{3cm}{\vspace{6pt}$c_\mathrm{local}$\vspace{6pt}} & \parbox[c]{8cm}{\vspace{6pt}Cutoff to set maximum number of local entanglement generation attempts\vspace{6pt}} & \parbox[c]{3cm}{\vspace{6pt} $\leq n_\mathrm{coh-active}$ (see Appendix \ref{app:skr})\vspace{6pt}} \\[6pt]
\hline\hline
\end{tabular}
\label{tab:fundamental-parameters}
\end{table}

\begin{table}[ht]
\caption{\textbf{Derived quantities.}}
\begin{tabular}{lll}
\hline\hline
\parbox[c]{3cm}{\vspace{4pt}\textbf{Symbol}\vspace{4pt}} & \parbox[c]{5cm}{\vspace{4pt}\textbf{Definition}\vspace{4pt}} & \parbox[c]{8cm}{\vspace{4pt}\textbf{Description}\vspace{4pt}} \\
\hline
\parbox[c]{3cm}{\vspace{6pt}$p_\mathrm{swap}$\vspace{6pt}} & \parbox[c]{5cm}{\vspace{6pt}$\PP(N'_\mathrm{local}\leq c_\mathrm{local})$ with $N_\mathrm{local}'\sim \Geo(p_\mathrm{local})$\vspace{6pt}} & \parbox[c]{8cm}{\vspace{6pt}Success probability of entanglement swap for a given pair of remote links\vspace{6pt}} \\[6pt]
\hline
\parbox[c]{3cm}{\vspace{6pt}$t_\mathrm{swap}$\vspace{6pt}} & \parbox[c]{5cm}{\vspace{6pt}$\EE(N_\mathrm{local})t_\mathrm{local}$ with \\ $N_\mathrm{local}\sim \TruncGeo(p_\mathrm{local}, c_\mathrm{local})$\vspace{6pt}} & \parbox[c]{8cm}{\vspace{6pt}Expected duration of entanglement swap\vspace{6pt}} \\[6pt]
\hline
\parbox[c]{3cm}{\vspace{6pt}$\lambda_\mathrm{long}$\vspace{6pt}} & \parbox[c]{5cm}{\vspace{6pt}$e^{-t_\mathrm{long}/t_\mathrm{coh\text{-}idle}}$\vspace{6pt}} & \parbox[c]{8cm}{\vspace{6pt}Depolarization parameter for a link stored during one Sync-Gen cycle when quantum chip is idle\vspace{6pt}} \\[6pt]
\hline
\parbox[c]{3cm}{\vspace{6pt}$\lambda_\mathrm{local}$\vspace{6pt}} & \parbox[c]{5cm}{\vspace{6pt}$e^{-t_\mathrm{local}/t_\mathrm{coh\text{-}idle}}$\vspace{6pt}} & \parbox[c]{8cm}{\vspace{6pt}Depolarization parameter for a link stored during one local entanglement generation attempt when quantum chip is idle\vspace{6pt}} \\[6pt]
\hline
\parbox[c]{3cm}{\vspace{6pt}$\lambda_\mathrm{local}^{(k)}$\vspace{6pt}} & \parbox[c]{5cm}{\vspace{6pt}$\EE\left(\lambda_\mathrm{local}^{kN_\mathrm{local}}\right)$ with $N_\mathrm{local}\sim \TruncGeo(p_\mathrm{local},c_\mathrm{local})$\vspace{6pt}} & \parbox[c]{8cm}{\vspace{6pt} 
The expected depolarization in a single link (for $k=1$) or a pair of links (for $k=2$) while establishing a local link between a different pair of chips
\vspace{6pt}} \\[6pt]
\hline
\parbox[c]{3cm}{\vspace{6pt}$t_\mathrm{coh-active}$\vspace{6pt}} & \parbox[c]{5cm}{\vspace{6pt}$n_\mathrm{coh-active}t_\mathrm{local}$\vspace{6pt}} & \parbox[c]{8cm}{\vspace{6pt} 
Memory coherence time when chip is active
\vspace{6pt}} \\[6pt]
\hline
\parbox[c]{3cm}{\vspace{6pt}$\lambda_\mathrm{active}$\vspace{6pt}} & \parbox[c]{5cm}{\vspace{6pt}$e^{-1/n_\mathrm{coh\text{-}active}}$\vspace{6pt}} & \parbox[c]{8cm}{\vspace{6pt}Depolarization parameter for a link stored during one local generation entanglement attempt when quantum chip is active\vspace{6pt}} \\[6pt]
\hline
\parbox[c]{3cm}{\vspace{6pt}$\lambda_0$\vspace{6pt}} & \parbox[c]{5cm}{\vspace{6pt}$\lambda_\mathrm{static}\,\EE\left(\lambda_\mathrm{active}^{2N'_\mathrm{local}}\mathbbm{1}(N'_\mathrm{local} \!\le\! c_\text{local})\right)$ \\ with $N'_\mathrm{local}\sim \Geo(p_\mathrm{local})$\vspace{6pt}} & \parbox[c]{8cm}{\vspace{6pt}Maximum expected contribution of any pair of matched links to sum of Werner parameters, i.e. the contribution without depolarizing factors due to idling in memory\vspace{6pt}} \\[6pt]
\hline\hline
\end{tabular}
\label{tab:derived-quantities}
\end{table}
 
\begin{table}[h]
\caption{\textbf{Random variables.} The superscript $\mathrm{pol}\in \{\mathrm{fxd},\mathrm{dyn}\}$ on a symbol indicates the FxdMux or DynMux policy, respectively. It is omitted if the symbol is not used for a specific policy. Symbols for the OSS approximations are decorated with a tilde, e.g. $\widetilde S_n^\mathrm{dyn}$ is the signed queue length at the start of time step $n$ in the OSS approximation for the DynMux policy.}
\begin{tabular}{ll}
\hline\hline
\parbox[c]{3cm}{\vspace{4pt}\textbf{Symbol}\vspace{4pt}} & \parbox[c]{8cm}{\vspace{4pt}\textbf{Description}\vspace{4pt}} \\
\hline
\parbox[c]{3cm}{\vspace{6pt}$S_n^\mathrm{pol}$\vspace{6pt}} & \parbox[c]{8cm}{\vspace{6pt}Signed queue length at the start of time step $n$ \vspace{6pt}} \\[6pt]
\hline
\parbox[c]{3cm}{\vspace{6pt}$U_n^\mathrm{pol}$\vspace{6pt}} & \parbox[c]{8cm}{\vspace{6pt}Number of end-to-end links generated in time step $n$\vspace{6pt}} \\[6pt]
\hline
\parbox[c]{3cm}{\vspace{6pt}$T_n^\mathrm{pol}$\vspace{6pt}} & \parbox[c]{8cm}{\vspace{6pt}Duration of time step $n$\vspace{6pt}} \\[6pt]
\hline
\parbox[c]{3cm}{\vspace{6pt}$V_n^\mathrm{pol}$\vspace{6pt}} & \parbox[c]{8cm}{\vspace{6pt}Sum of Werner parameters of end-to-end links produced in time step $n$\vspace{6pt}} \\[6pt]
\hline
\parbox[c]{3cm}{\vspace{6pt}$V_n^\mathrm{fut,\mathrm{pol}}$\vspace{6pt}} & \parbox[c]{8cm}{\vspace{6pt}Sum of future Werner parameters of the new links produced in time step $n$\vspace{6pt}} \\[6pt]
\hline
\parbox[c]{3cm}{\vspace{6pt}$U_\ell^\mathrm{pol'}$\vspace{6pt}} & \parbox[c]{8cm}{\vspace{6pt}Number of end-to-end links generated in time step starting from state $\ell$\vspace{6pt}} \\[6pt]
\hline
\parbox[c]{3cm}{\vspace{6pt}$T_\ell^\mathrm{pol'}$\vspace{6pt}} & \parbox[c]{8cm}{\vspace{6pt}Duration of time step starting from state $\ell$\vspace{6pt}} \\[6pt]
\hline
\parbox[c]{3cm}{\vspace{6pt}$V_\ell^\mathrm{pol'}$\vspace{6pt}} & \parbox[c]{8cm}{\vspace{6pt}Sum of Werner parameters of end-to-end links produced in time step starting from state $\ell$\vspace{6pt}} \\[6pt]
\hline
\parbox[c]{3cm}{\vspace{6pt}$V_\ell^\mathrm{fut,\mathrm{pol'}}$\vspace{6pt}} & \parbox[c]{8cm}{\vspace{6pt}Sum of future Werner parameters of the new links produced in time step starting from state $\ell$\vspace{6pt}} \\[6pt]
\hline
\parbox[c]{3cm}{\vspace{6pt}$E_{\ell}^{\mathrm{pol'}}$\vspace{6pt}} & \parbox[c]{8cm}{\vspace{6pt} Number of entanglement swaps attempted in time step starting from state $\ell$\vspace{6pt}} \\[6pt]
\hline
\parbox[c]{3cm}{\vspace{6pt}$X^{\mathrm{pol'}}_{\ell}$\vspace{6pt}} & \parbox[c]{8cm}{\vspace{6pt}Number of new links matched with a link from the queue in time step starting from state $\ell$\vspace{6pt}} \\[6pt]
\hline
\parbox[c]{3cm}{\vspace{6pt}$Y^{\mathrm{pol'}}_{\ell}$\vspace{6pt}} & \parbox[c]{8cm}{\vspace{6pt}Number of pairs of new links matched in time step starting from state $\ell$\vspace{6pt}} \\[6pt]
\hline
\parbox[c]{3cm}{\vspace{6pt}$Z^{\mathrm{pol'}}_{\ell}$\vspace{6pt}} & \parbox[c]{8cm}{\vspace{6pt}Number of new links that remain unmatched in time step starting from state $\ell$\vspace{6pt}} \\[6pt]
\hline
\parbox[c]{3cm}{\vspace{6pt}$N'_\mathrm{local}$\vspace{6pt}} & \parbox[c]{8cm}{\vspace{6pt} (Hypothetical) number of local entanglement generation attempts until success during an entanglement swap, distributed as $\Geo(p_\mathrm{local})$ \vspace{6pt}} \\[6pt]
\hline
\parbox[c]{3cm}{\vspace{6pt}$N_\mathrm{local}$\vspace{6pt}} & \parbox[c]{8cm}{\vspace{6pt} Actual number of local entanglement generation attempts during an entanglement swap, distributed as $\TruncGeo(p_\mathrm{local}, c_\mathrm{local})$\vspace{6pt}} \\[6pt]
\hline\hline
\end{tabular}
\label{tab:random variables}
\end{table}


\section{Derivation of Performance Metrics as Long-term Averages}\label{app: markov abstract}

In \cref{section: performance metrics} we defined the steady-state rate and steady-state Werner parameter for the multiplexed quantum repeater in terms of time averages. 
In \cref{section: markov chains exact} we described the evolution of the signed queue lengths using discrete-time Markov chains, which completely described the state of the repeater.
A list of the main random variables used to model the multiplexed repeater is provided in \cref{tab:random variables}. 
In this appendix, we first prove Lemma~\ref{lemma: main replacing random variable by its expectation in ergodic theorem} about the long-term behavior of a generic Markov-modulated process $(B_n)_{n\geq1}\sim \Markov(P)$, 
which we use to compute the steady-state rate and Werner parameter.


\subsection{Ergodic Theorem for Markov-modulated Processes}

\begin{proof}[Proof of Lemma~\ref{lemma: main replacing random variable by its expectation in ergodic theorem}]
    We first recall the well-known fact that, for an irreducible Markov chain on a countable state space, a stationary distribution, if it exists, is unique.
    We now define $C_n := (A_n, B_n)$ for $n \in \NN$.
    Then,  $(C_n)_{n\geq1}$ is a Markov chain on the countable state space ${\mathcal{C} :=\{(a,b): a \in \mathcal{S}',b\in \mathcal{H}_a\}}$, where $\mathcal{H}_a$ denotes the support of $H_a$, $a \in \mathcal{S}'$, and its transition probabilities are given by
    \begin{equation}\label{eq: tpm product chain lemma}
        \PP \big(C_{n+1} = (a',b') \mid C_{n} = (a,b)\big) = \PP \big(A_{n+1} = a' \mid A_{n} = a\big) \PP_{H_{a'}}(b'),
    \end{equation}
    where $\PP_{H_{a'}}(b')$ denotes the probability of the singleton $\{b'\}$ under distribution $H_{a'}$.
    We show that $(C_n)_{n\geq1}$ is an irreducible Markov chain that has a stationary distribution $\rho$.
    We consider $(a,b), (a',b') \in \mathcal{C}$.
    Since $(A_n)_{n\geq1}$ is irreducible, $\exists~ k \in \NN$ such that ${\PP \big(A_{n+k} \!=\! a' \mid A_{n} \!=\! a\big)\!>\!0}$.
    Irreducibility of $(C_n)_{n\geq1}$ now follows by observing that 
    \begin{equation}\label{eq: tpm k product chain lemma}
        \PP \big(C_{n+k} = (a',b') \mid C_{n} = (a,b)\big) = \PP \big(A_{n+k} = a' \mid A_{n} = a\big) \PP_{H_{a'}}(b')>0,
    \end{equation}
    as $\PP_{H_{a'}}(b')>0$ by definition of $\mathcal{C}$.
    The distribution $\rho_{(a,b)} = \pi_a' \PP_{H_a}(b)$ is the stationary distribution of $(C_n)_{n\geq1}$, since
    \begin{align*}
        \sum_{(a,b) \in \mathcal{C}}\rho_{(a,b)} \PP \big(C_{n+1} = (a',b') \mid C_{n} = (a,b)\big) &= \sum_{(a,b)\in \mathcal{C}}\pi_a' \PP_{H_{a}}(b)\PP \big(A_{n+1} = a' \mid A_{n} = a\big) \PP_{H_{a'}}(b') \\
        &= \underbrace{\sum_{a \in \mathcal{S}'}\pi_a'\PP \big(A_{n+1} = a' \mid A_{n} = a\big)}_{\pi_{a'}'} \Big(\underbrace{\sum_{b \in \mathcal{H}_a}  \PP_{H_{a}}(b)}_{1}\Big)  \PP_{H_{a'}}(b') \\
        &= \rho_{(a',b')},
    \end{align*}
    and 
    \begin{align*}
        \sum_{(a,b) \in \mathcal{C}} \rho_{(a,b)} =
        \sum_{a \in \mathcal{S}'} \pi_a' \Big(\sum_{b \in \mathcal{H}_a} \PP_{H_{a}}(b)\Big)=1.
    \end{align*}
    The limiting distribution of $(B_n)_{n \geq 1}$ is given by the marginal of $\rho$, i.e., $\sum_{a \in \mathcal{S}'}\pi_a' \PP_{H_a}(b)$, which implies that the corresponding expectation $\mu_B$ exists and is given by $\mu_B=\sum_{a\in \mathcal{S}'}\pi'_a \EE(H_a)$ whenever $|\EE(H_a)|< \infty$ $\forall a \in \mathcal{S}'$.
    
    For the long-term average of $(B_n)_{n \geq 1}$, we define $\phi(a,b)=b$ on $\mathcal{C}$ 
    so that
    \begin{align*}
        \lim_{N\to \infty} \frac{1}{N}\sum_{n=1}^N B_n = \lim_{N\to \infty}\frac{1}{N} \sum_{n=1}^N \phi(A_n,B_n)
        \overset{\mathrm{a.s.}}{=} \sum_{(a,b) \in \mathcal{C}} \rho_{(a,b)} \phi(a,b)= \sum_{a \in \mathcal{S}'} \pi_{a}' \sum_{b \in \mathcal{H}_a} \PP_{H_{a}}(b)b = \sum_{a\in \mathcal{S}'}\pi_a' \mu_{H_a} = \mu_B.
    \end{align*}
    Here, the second equality follows by the ergodic theorem~\cite[Example 6.2.4]{durrett2019probability}, 
    which is possible due to the assumption that $|\mu_{H_a}| < \infty$ $\forall a \in \mathcal{S}'$.
    
\end{proof}
\begin{remark}
    For analyzing the performance metrics, we consider processes corresponding to the duration ($(T_n)_{n \ge 1}$), number of link generations ($(U_n)_{n \ge 1}$), and sum of Werner parameters of the links generated ($(V_n)_{n \ge 1}$) in a time step, where the underlying modulating process is the signed queue length $(S_n)_{n \ge 1}$. 
    The process $(S_n)_{n \ge 1}$ is a Markov chain defined on a finite state space and is irreducible.
    As the proof above shows, the modulated processes considered have a unique stationary distribution, which is their limiting distribution as well.
    We thus use these two terms interchangeably for our analysis below.
\end{remark}

\subsection{Steady-state Rate}
In this section and the next, we apply the ergodic lemma for Markov-modulated processes (Lemma \ref{lemma: main replacing random variable by its expectation in ergodic theorem}) to express the steady-state rate and steady-state Werner parameter of the multiplexed repeater, respectively, in terms of state averages.
For the rest of this section we let $\mathcal{S}$ denote the state space of the signed queue length, i.e., we adopt a policy-agnostic notation, as our analysis in this section applies to both FxdMux and DynMux policies.

Recall from Definition \ref{def: steady-state rate long time} that the steady-state rate is defined as
\begin{equation}
    r = \lim_{N\to \infty} \frac{\sum_{n=1}^N U_n}{\sum_{n=1}^N T_n},
\end{equation}
where $U_n$ is the number of end-to-end links produced in time step $n$ and $T_n$ denotes its duration. 
In \cref{theorem: steady-state rate ergodic} we use the ergodic theorem to compute the steady-state rate as the ratio of respective steady-state averages.
More precisely, we show that the steady-state rate is the ratio of the expectations of the random variables $U$ and $T$, where $U|S=\ell \sim U'_\ell$ and $T|S=\ell \sim T'_\ell$ with $U'_\ell$ and $T'_\ell$ respectively denoting the number of end-to-end links produced in a time step and the duration of a time step starting with signed queue length $\ell$.
Furthermore, $S$ here denotes the steady-state signed queue length.
We first illustrate the definitions of $U$ and $T$. 

Let $\ell\geq 0$ denote the signed queue length at the beginning of a time step,
where positive queue length implies that there is no link on the left segment.
Further, let $K_{L,\ell}'$ and $K_{R,\ell}'$ denote the number of remote links generated to the left and the right in this state, respectively, so that there are a total of $E'_{\ell}:=\min(K_{L,
\ell}', \ell + K_{R,\ell}')$ entanglement swaps to be attempted in the entanglement swap phase. 
On the other hand, starting from signed queue length $\ell<0$, we have $E'_{\ell}:=\min(-\ell+K_{L,\ell}', K_{R,\ell}')$. Denoting by $N'_\mathrm{local}$ the (hypothetical) number of local entanglement generation attempts until success during an entanglement swap, we have $N'_\mathrm{local}\sim\Geo(p_\mathrm{local})$ and each entanglement swap spans ${N_\mathrm{local}:=\min(N'_\mathrm{local},c_\mathrm{local})}$ local attempts. 
That is $N_\mathrm{local}\sim\TruncGeo(p_\mathrm{local},c_\mathrm{local})$. Then, the number of end-to-end links generated from state $\ell$ is distributed as
\begin{equation}\label{eq: number of links RV ell}
    U'_\ell \overset{d}{=} \sum\nolimits_{i=1}^{E'_{\ell}} \indicator(N_{\mathrm{local},i}'\leq c_\mathrm{local}),
\end{equation}
where $\indicator(\cdot)$ denotes the indicator function and $N'_{\mathrm{local},i}\overset{\mathrm{iid}}{\sim} \Geo(p_\mathrm{local})$. Taking expectations in \cref{eq: number of links RV ell}, 
by the iid property of the entanglement swap durations and their independence with number of swap attempts $E'_\ell$ we have that
\begin{equation}\label{eq: U to E}
    \EE(U'_\ell) = \sum_{k\geq 1}\PP\left(E_{\ell}'=k\right)k \, \underbrace{\EE\left(\indicator(N'_\mathrm{local}\leq c_\mathrm{local})\right)}_{p_\mathrm{swap}} = \EE(E'_{\ell}) \,p_\mathrm{swap}. 
\end{equation}
Also, the duration of the time step starting from state $\ell$ is distributed as
\begin{equation}\label{eq: duration time step RV ell}
    T'_\ell \overset{d}{=} t_\mathrm{long} + \sum\nolimits_{i=1}^{E'_{\ell}} N_{\mathrm{local},i}\, t_\mathrm{local},
\end{equation}
where $N_{\mathrm{local},i}\overset{\mathrm{iid}}{\sim} \TruncGeo(p_\mathrm{local},c_\mathrm{local})$, $t_\mathrm{long}$ is the combined duration of the synchronization and remote generation cycles that are present in every time step, and $N_{\mathrm{local},i}\, t_\mathrm{local}$ is the duration of the $i$th entanglement swap.
The expected duration of the time step starting with signed queue length $\ell$ can be expressed in terms of the expected number of end-to-end links produced in that step. 
Taking expectations in \cref{eq: duration time step RV ell}  and using the iid property of the entanglement swaps and their independence with $E'_\ell$, it follows that
\begin{equation}
    \EE(T_\ell') = t_\mathrm{long} + \EE(E'_{\ell}) \EE(N_{\mathrm{local}}) t_\mathrm{local}.
\end{equation}
By combining the above equation with \cref{eq: U to E} and setting $t_\mathrm{swap}:=\EE(N_{\mathrm{local}}) t_\mathrm{local}$ it follows that
\begin{equation}\label{eq: expected time step from U}
    \EE(T'_\ell) = t_\mathrm{long} + \EE(U'_\ell)\frac{t_\mathrm{swap}}{p_\mathrm{swap}}.
\end{equation}

We now provide a computable expression for the steady-state rate using Lemma~\ref{lemma: main replacing random variable by its expectation in ergodic theorem} in terms of the steady-state expectations of the number of end-to-end link generations and the duration of a time step, defined as below:
\begin{align}
    \mu_U &= \sum_{\ell\in \mathcal{S}}\pi_\ell \EE(U'_\ell) \quad  \label{eq: definition mu U} \\
    \text{and} \quad \mu_T &= \sum_{\ell\in \mathcal{S}}\pi_\ell \EE(T'_\ell), \label{eq: definition mu T}
\end{align}
where $\pi$ denotes the steady-state distribution of the signed queue length process $(S_n)_{n \ge 1}$.
\begin{remark}\label{rem:finiteUT}
    Since $U'_\ell \le m$, it is clear from \cref{eq: U to E} and \cref{eq: expected time step from U} that both $\mu_U$ and $\mu_T$ are finite.    
\end{remark}

We now prove \cref{theorem: steady-state rate ergodic} which expresses the steady-state rate in terms of $\mu_U$ and $\mu_T$.
\begin{proof}[Proof of \cref{theorem: steady-state rate ergodic}]
    Let $U_n$ denote the number of end-to-end links produced in time step $n$.
    Then, the conditional distribution of $U_n$ given entire history is the same as that of $U_n$ given $S_n$ and $U_n|(S_n=\ell) \overset{d}{=} U'_\ell$; see \cref{eq: number of links RV ell}.
    Hence, by Lemma \ref{lemma: main replacing random variable by its expectation in ergodic theorem} and Remark~\ref{rem:finiteUT}, 
        \begin{align}
            \lim_{N\to \infty} \frac{1}{N}\sum_{n=1}^{N}U_n & \overset{\mathrm{a.s.}}{=} \mu_U  \label{eq:Ulim} \\
            \text{and} \quad \lim_{N\to \infty} \frac{1}{N}\sum_{n=1}^{N}T_n & \overset{\mathrm{a.s.}}{=} \mu_T.
        \end{align}
        Therefore,
        \begin{equation}
            r=\lim_{N\to \infty} \frac{\sum_{n=1}^N U_n}{\sum_{n=1}^N T_n} = \lim_{N\to \infty} \frac{\frac{1}{N}\sum_{n=1}^N U_n}{\frac{1}{N}\sum_{n=1}^N T_n} \overset{\mathrm{a.s.}}{=} \frac{\mu_U}{\mu_T}.
        \end{equation}
\end{proof}

\subsection{Steady-state Werner Parameter}\label{section: steady state werner parameter}
The steady-state Werner parameter as defined in Definition \ref{def: steady-state werner long time} is the average Werner parameter of the end-to-end links produced in the long term. 
Similar to the steady-state rate, we want to express the steady-state Werner parameter in terms of averages with respect to the steady-state distribution $\pi$ of the signed queue lengths $(S_n)_{n\geq1}$.
However, 
the expected sum of the Werner parameters of the end-to-end links produced in a time step starting with signed queue length $\ell$ 
also depends on the time that its constituent links have been stored in memory.
Therefore, instead of computing the expected Werner parameter of the end-to-end links produced in a time step starting from state $\ell$, 
we compute the Werner parameters of the end-to-end links that \textit{will be produced} in future time steps from the newly generated links in the time step starting from state $\ell$. 
This is made precise below.

When new links are generated in state $\ell$, we have the following possibilities with respect to matching:
\begin{itemize}
    \item  a new link is matched with a link from the queue, and we denote the \textit{pairs} as $X$-pairs.
    \item two new links are matched with each other, and we denote the \textit{pairs} as $Y$-pairs.
    \item a new link is not matched and has to be stored in the queue, and we denote the \emph{links} as $Z$-links.
\end{itemize}
Examples of different types of link generation events are shown in \cref{fig:link generation events}. 
\begin{figure}
     \centering
     \begin{subfigure}[b]{0.3\textwidth}
         \centering
         \includegraphics[trim=8cm 5cm 8cm 0cm, clip, width=\textwidth]{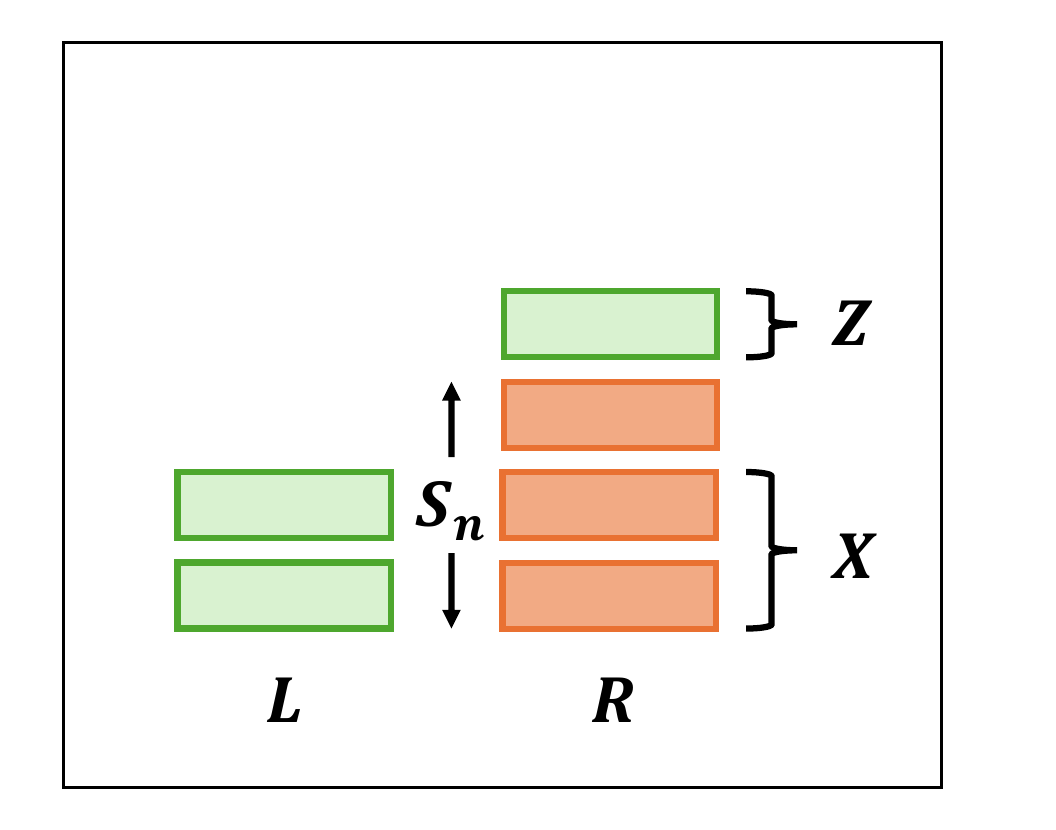}
         \caption{}
     \end{subfigure}
     \begin{subfigure}[b]{0.3\textwidth}
         \centering
         \includegraphics[trim=8cm 5cm 8cm 0cm, clip, width=\textwidth]{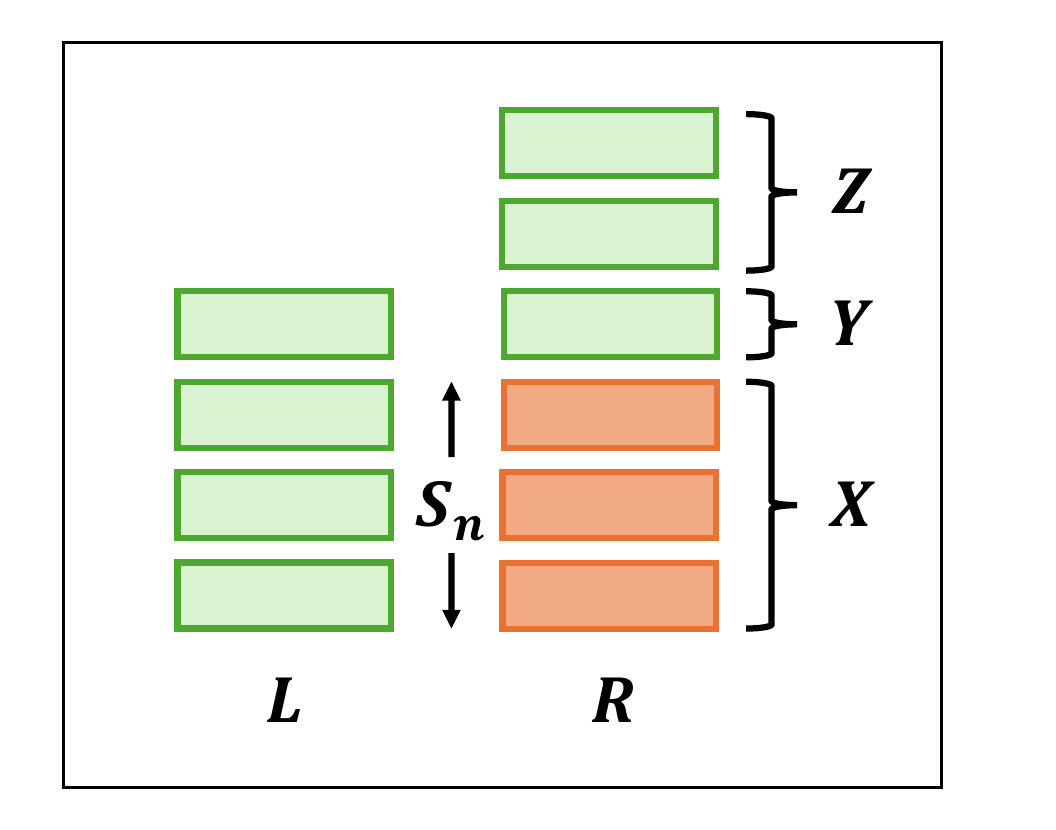}
         \caption{}
     \end{subfigure}
     \begin{subfigure}[b]{0.3\textwidth}
         \centering
         \includegraphics[trim=8cm 5cm 8cm 0cm, clip, width=\textwidth]{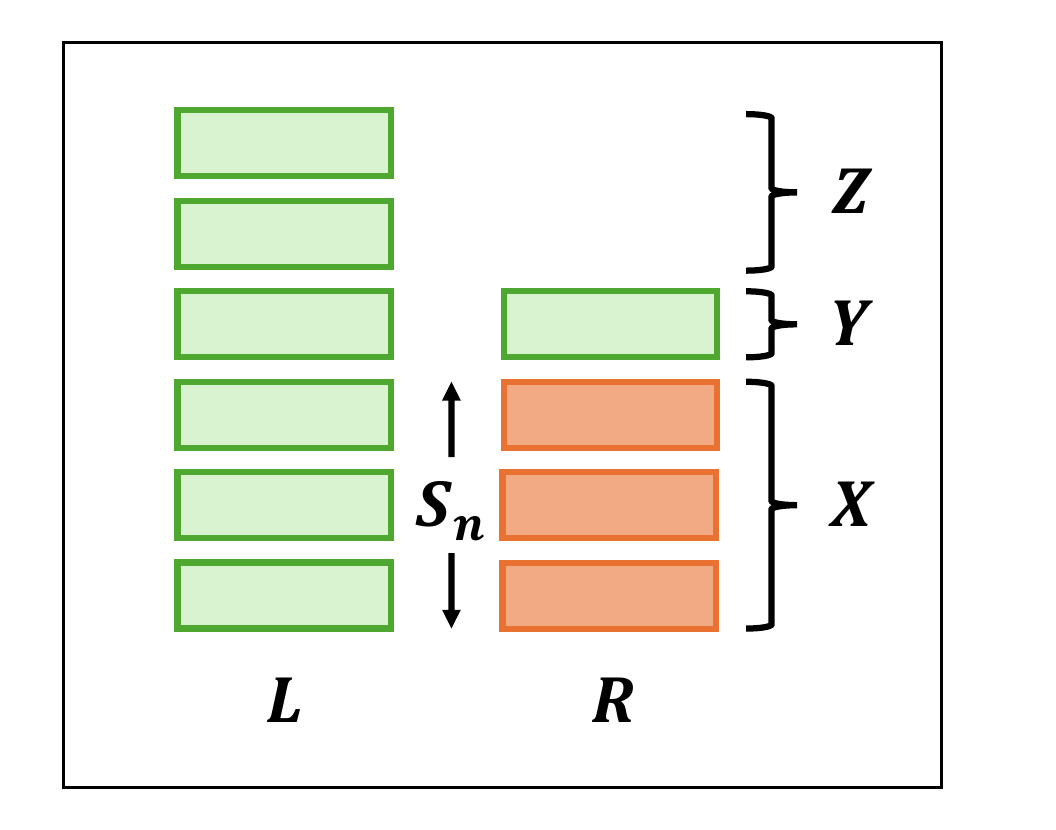}
         \caption{}
     \end{subfigure}
     \caption{\textbf{Examples of $X$-pairs, $Y$-pairs and $Z$-links in a time step starting with signed queue length $S_n = 3$.} The links that were already in queue are indicated by dark orange rectangles, and the new links, by light green rectangles. Subfigures (a), (b), and (c) are representatives for three different types of link generation events over the segments $L$ and $R$.}
    \label{fig:link generation events}
\end{figure}
The sum of the Werner parameters of the end-to-end links produced in a time step that starts with signed queue length $\ell$, would have contributions from the $X$ and $Y$ pairs. 
However, determining the Werner parameter of an $X$-pair requires the age of the queued link and finding that age is not straightforward only with the information on the signed queue length. 

The key observation to the computation of the steady-state Werner parameters is to note that any $X$-pair, i.e., a pair of a newly generated link and a queued link, must have started as a $Z$-link. 
Conversely, any $Z$-link will in the future become part of an $X$-pair. 
This motivates us to define random variables for the sum of \textit{future Werner parameters}: one based on the \textit{state} from which the $Y$-pairs and the $Z$-links are generated and the other based on the \textit{time index} when the $Y$-pairs and the $Z$-links are generated:
\begin{itemize}
    \item[] $V_\ell^{\mathrm{fut}'}$: the sum of the Werner parameters of the end-to-end links produced from the $Y$-pairs and the $Z$-links generated in a time step starting from \textit{state} $\ell$.
    \item[] $V^\mathrm{fut}_n$: the sum of the Werner parameters of the end-to-end links produced from the $Y$-pairs and $Z$-links generated in \textit{time step} $n$.
\end{itemize}

The next lemma relates the sum of Werner parameters of the end-to-end links produced in the time step $n$ to the time-indexed sum of future Werner parameters.
\begin{lemma}\label{lemma: equivalence of W and W fut}
    Let $V_n$ be the sum of the Werner parameters of the end-to-end links produced in time step $n$ and let $V_n^\mathrm{fut}$ be the sum of the Werner parameters of the end-to-end links produced from the new links generated in time step $n$, as defined above. Then,
    \begin{equation}
    \lim_{N\to \infty} \frac{1}{N}\sum_{n=1}^N V_n^\mathrm{fut} =
    \lim_{N\to \infty} \frac{1}{N}\sum_{n=1}^N V_n.
\end{equation}
\end{lemma}
\begin{proof}
    Let $X_n$, $Y_n$, and $Z_n$ respectively denote the number of $X$-pairs, $Y$-pairs, and $Z$-links observed in time step $n$.
    Then,
    \begin{equation}\label{eq: w decomp}
        V_n=\sum_{x=1}^{X_n} V_{x,n} + \sum_{y=1}^{Y_n} V_{y,n},
    \end{equation}
    where $V_{x,n}$ is the Werner parameter of the end-to-end link from the $x$-th $X$-pair and $V_{y,n}$ is the Werner parameter of the end-to-end link from the $y$-th $Y$-pair in time step $n$.
    Similarly, we can express $V_n^\mathrm{fut}$ as
    \begin{equation}\label{eq: wfut decomp}
        V_n^\mathrm{fut}=\sum_{y=1}^{Y_n} V_{y,n} + \sum_{z=1}^{Z_n} V_{z,n}^\mathrm{fut},
    \end{equation}
    where $V_{z,n}^\mathrm{fut}$ is the Werner parameter that will be produced in a future time-step from the $z$-th $Z$-link.
    The $Y$-terms of \cref{eq: w decomp} and \cref{eq: wfut decomp} agree in every step since $Y$-links are generated and matched in the same step. 
    Further, any $X$-link that is matched in step $n$ must have been generated as a $Z$-link in the past, i.e.,
    $V_{x,n}=V^\mathrm{fut}_{z,n'}$ for some $n'<n$.
    Therefore,
    \begin{align}
        \Big|\frac{1}{N}\sum_{n=1}^{N} V_n - \frac{1}{N}\sum_{n=1}^N V_n^\mathrm{fut}\Big| &= \Big|\frac{1}{N}\sum_{n=1}^N \sum_{x=1}^{X_n} V_{x,n} - \frac{1}{N}\sum_{n=1}^N \sum_{z=1}^{Z_n} V_{z,n}^\mathrm{fut}\Big|
        \leq \frac{1}{N}|S_{N+1}|,
    \end{align}
    since the queue size $|S_{N+1}|$ at the start of step $N+1$ is the number of $Z$-links that have remained unmatched until the end of step $N$, and the Werner parameters of these unmatched $Z$-links are bounded by $1$. As the queue length $|S_{N+1}|$ is bounded by the number of chips, $2m$, it follows that
    \begin{equation}
        \lim_{N\to \infty} \Big|\frac{1}{N}\sum_{n=1}^N V_n - \frac{1}{N}\sum_{n=1}^N V_n^\mathrm{fut}\Big| \leq \lim_{N\to \infty} \frac{2m}{N} = 0,
    \end{equation}
    and we conclude $V_n$ and $V_n^\mathrm{fut}$ have the same long-term average.
\end{proof}

We now consider the expected sum of future Werner parameters produced in a time step starting with signed queue length $\ell$, which can be computed by conditioning on the counts of $X$-pairs, $Y$-pairs, and $Z$-links generated in that time step.
Since the distribution of these counts are completely determined by the signed queue length, we denote them as $(X'_\ell,Y'_\ell,Z'_\ell)$ and compute the expected sum of future Werner parameters as
\begin{equation}\label{eq: expected future Werner parameter from state}
    \EE\left(V^{\mathrm{fut}'}_\ell\right) = \sum_{(x,y,z) \in \NN_{\ge 0}^3} p_\ell(x,y,z)\, \EE\left(V^{\mathrm{fut}'}_\ell\mid(X'_\ell,Y'_\ell,Z'_\ell) = (x,y,z)\right),
\end{equation}
where $p_\ell(x,y,z) \!:=\! \PP\big((X'_\ell,Y'_\ell,Z'_\ell) = (x,y,z)\big)$.
The quantities $p_\ell(x,y,z)$ can be determined directly from the entanglement generation policy and is done for FxdMux in \cref{eq: prob xyz for fxd} and for DynMux in \cref{eq: prob xyz for dyn}. The computation of the expected sum of future Werner parameters given $(X'_\ell,Y'_\ell,Z'_\ell) = (x,y,z)$ is explained below, where we use the shorthand $V^{\mathrm{fut}'}_\ell|(x,y,z)$ for ${V^{\mathrm{fut}'}_\ell|(X'_\ell,Y'_\ell,Z'_\ell) = (x,y,z)}$.

\begin{lemma}\label{lemma: condition expectation V fut}
    Let $V_\ell^\mathrm{fut'}|(x,y,z)$ be the sum of future Werner parameters produced from the new links generated in a time step starting with signed queue length $\ell$ given the event $(X_\ell',Y_\ell',Z_\ell')=(x,y,z)$. 
    Denoting by $\Lambda(a,b)$ the future idle depolarization parameter incurred by a link at position $a$ in a queue of size $b$ at the end of a time step until it enters the entanglement swap phase, we have 
    \begin{align}
    & \EE\left(V^{\mathrm{fut}'}_\ell\mid(x,y,z)\right) =\left(\lambda_\mathrm{local}^{(2)}\right)^{x}\sum_{i=1}^y\left(\lambda_\mathrm{local}^{(2)}\right)^{i-1}\lambda_0 + \left(\lambda_\mathrm{local}^{(1)}\right)^{x+y} \sum_{i=1}^z \EE\left(\Lambda(\ell-x+i, \ell-x+z)\right)\lambda_0, \quad \text{where} \label{eq: conditional W future}\\
    & \lambda_\mathrm{local}^{(k)}:=\EE\left(\lambda_\mathrm{local}^{kN_\mathrm{local}}\right)\thickspace \mathrm{with}\thickspace N_\mathrm{local}\!\sim\! \TruncGeo(p_\mathrm{local}, c_\mathrm{local}), \quad \mathrm{and} \label{eq: def lambda local k}\\ 
    &\lambda_0 := \lambda_\mathrm{static}\,\EE\left(\lambda_\mathrm{active}^{2N'_\mathrm{local}} \mathbbm{1}(N'_\mathrm{local} \!\le\! c_\mathrm{local})\right)\thickspace \mathrm{with}\thickspace  N'_\mathrm{local}\!\sim\! \Geo(p_\mathrm{local}) \label{eq: definition of lambda 0}.
\end{align}
\end{lemma}
\begin{proof}
Recall from \cref{section: system model} that we assume pairs of entangled links are swapped sequentially in FIFO order during the entanglement swap phase.
In this phase, there are at most $c_\mathrm{local}$ independent local entanglement generation attempts, and entanglement swap fails if no success is observed until this point.
We denote by $N'_\mathrm{local}$ the (hypothetical) number of local entanglement generation attempts until success, and by ${N_\mathrm{local}:=\min(N'_\mathrm{local},c_\mathrm{local})}$ the actual number of local attempts. 
Then for the link generation event ${(X'_\ell,Y'_\ell,Z'_\ell) \!=\! (x,y,z)}$ in a time step starting with signed queue length $\ell$, the sum of future Werner parameters is given by

\begin{multline}\label{eq: conditional W future as random variable}
    V^{\mathrm{fut}'}_\ell|(x,y,z) \overset{d}{=}
    \sum_{i=1}^y \Big(\prod_{j=1}^{x+i-1} \lambda_\mathrm{local}^{2N_{\mathrm{local},j}}\Big)
    \lambda_\mathrm{active}^{2N'_{\mathrm{local},x+i}} \lambda_\mathrm{static} \mathbbm{1}(N'_{\mathrm{local},x+i}\leq c_\mathrm{local}) + \\
    \Big(\prod_{j=1}^{x+y}\lambda_\mathrm{local}^{N_{\mathrm{local},j}}\Big)\sum_{i=1}^z \big(\Lambda(\ell-x+i, \ell-x+z)\lambda_\mathrm{active}^{2N'_{\mathrm{local},x+y+i}} \lambda_\mathrm{static} \mathbbm{1}(N'_{\mathrm{local},x+y+i}\leq c_\mathrm{local}),
\end{multline}
where the $N'_{\mathrm{local},i}\overset{\mathrm{iid}}{\sim} \Geo(p_\mathrm{local})$ for $i=1, \dots, x + y + z$ and $N_{\mathrm{local},i}:=\min(N'_{\mathrm{local},i},c_\mathrm{local})$ are the number of local entanglement generation attempts observed for the entanglement swap of pair $i$. 
Further, $\mathbbm{1}(\cdot)$ denotes the indicator function. The first term in \cref{eq: conditional W future as random variable} corresponds to contributions from $Y$-pairs and the second term to contributions from $Z$-links.

\textbf{$Y$-pairs in \cref{eq: conditional W future as random variable}:} 
The first term in \cref{eq: conditional W future as random variable} accounts for the $Y$-pairs that are generated and swapped in the current time step. All the $X$-pairs and the first $(i-1)$ $Y$-pairs have to be swapped before the $i$-th $Y$-pair can be swapped.
This gives the depolarization factor $\prod_{j=1}^{x+i-1} \lambda_\mathrm{local}^{2N_{\mathrm{local},j}}$ since both links of the $i$-th $Y$-pair are depolarized by $\lambda_\mathrm{local}$ during each local entanglement generation attempt between other pairs of chips.
The factors of $\lambda_\mathrm{active}$ account for the active depolarization experienced by links when local entanglement generation attempts are performed between the communication qubits of \textit{their own} chips during entanglement swapping. The factor $\lambda_\mathrm{static}$ accounts for all time-independent decoherence processes.
The indicator functions of the form $\mathbbm{1}(N'_{\mathrm{local},i}\leq c_\mathrm{local})$ 
make sure that the contribution is only added to the future Werner parameter when the entanglement swap is successful. 

\textbf{$Z$-links in \cref{eq: conditional W future as random variable}:} For the second term accounting for the $Z$-links,
recall that the second link of the pair is generated in a future time step. The overall factor $\prod_{j=1}^{x+y}\lambda_\mathrm{local}^{N_{\mathrm{local},j}}$ accounts for the depolarization of each $Z$-link while waiting for the $X$-pairs and $Y$-pairs in the step that they were generated to be swapped.
Conditioned on the event that $(X'_\ell,Y'_\ell, Z'_\ell)=(x,y,z)$, the $i$th $Z$-link produced in state $\ell$ starts the next time step in position $\ell-x+i$ in a queue of size $\ell-x+z$. The factor $\Lambda(a,b)$ denotes the future idle depolarization parameter incurred by a link when it starts at position $a$ in a queue of size $b$ until it enters the entanglement swap. This factor depends on the entanglement generation policy. Finally, the factors of $\lambda_\mathrm{active}$ account for depolarization when the $Z$-links are swapped in some future time step, $\lambda_\mathrm{static}$ gives the time-independent depolarization, and the indicator functions make sure the contributions are only added when the entanglement swap is successful.

Since $N_{\mathrm{local},i}$s are independent, the result follows by taking expectations in \cref{eq: conditional W future as random variable}.
\end{proof}
%

%
Note that $\lambda_\mathrm{local}^{(k)}$ in \cref{eq: def lambda local k} is the expected depolarization in a single link (for $k=1$) or a pair of links (for $k=2$) when establishing a local link between a different pair of chips. Moreover, $\lambda_0$ in \cref{eq: definition of lambda 0} is the maximum expected contribution of any pair of matched links to sum of Werner parameters, i.e. the contribution without depolarizing factors due to links idling in memory. It is shown in Lemma \ref{lemma: Lambda a b fxdmux} that the expectations $\EE(\Lambda(a,b))$ of future idle depolarization parameters for all starting positions $a$ and total queue lengths $b$ for the FxdMux policy can be obtained from a single linear system of dimension $m$, and Lemma \ref{lemma: Lambda a b dynmux} shows that for DynMux they can be obtained from $2m-2$ linear systems of dimensions $1,\dots,2m-2$.

Since $V_\ell^\mathrm{fut}$ has bounded outcomes, we can compute its steady-state expectation from Lemma \ref{lemma: main replacing random variable by its expectation in ergodic theorem} as
\begin{equation}\label{eq: definintion mu V fut}
    \mu_{V^\mathrm{fut}} = \sum_{\ell\in \mathcal{S}}\pi_\ell \EE\left(V^\mathrm{fut'}_\ell\right).
\end{equation}
The steady-state Werner parameter $w$ can then be given in terms of $\mu_U$ and $\mu_{V^\mathrm{fut}}$.


\begin{proof}[Proof of \cref{theorem: main steady-state werner fut ergodic}]
    Starting from Definition~\ref{def: steady-state werner long time}, we have 
    $$w=\lim_{N\to\infty}\frac{\sum_{n=1}^N V_n}{\sum_{n=1}^N U_n} \overset{\text{(i)}}{=}\lim_{N\to\infty}\frac{\sum_{n=1}^N V^\mathrm{fut}_n}{\sum_{n=1}^N U_n} = \lim_{N\to\infty}\frac{\frac{1}{N}\sum_{n=1}^N V^\mathrm{fut}_n}{\frac{1}{N}\sum_{n=1}^N U_n},$$
    where (i) follows from Lemma~\ref{lemma: equivalence of W and W fut}.
    Since $V_\ell^\mathrm{fut}$ and $U_\ell$ have bounded outcomes for all states $\ell$, we can now apply Lemma \ref{lemma: main replacing random variable by its expectation in ergodic theorem} on the numerator and the denominator of the last expression, which gives $\PP(\frac{1}{N}\sum_{n=1}^N V^\mathrm{fut}_n \to \mu_{V^{\mathrm{fut}}}) = 1$ and $\PP(\frac{1}{N}\sum_{n=1}^N U_n \to \mu_U) = 1$.
    Since $\mu_U>0$, the ratio is eventually well-defined on a probability $1$ set, and we indeed have $w \overset{\mathrm{a.s.}}{=} \mu_{V^{\mathrm{fut}}}/\mu_U$.
\end{proof}

We make the following observation regarding the maximum steady-state Werner parameter, which will be used in Appendix \ref{app:skr} when optimizing the secret key rate over the cutoff $c_\mathrm{local}$.
\begin{lemma}\label{lemma: upper bound on w}
    The steady-state Werner parameter $w$ is at most $w_\mathrm{max}:= \lambda_0/p_\mathrm{swap}$.
\end{lemma}
\begin{proof}
    By \cref{theorem: main steady-state werner fut ergodic} the steady-state Werner parameter is $w \overset{\mathrm{a.s}}= \mu_{V^{\mathrm{fut}}}/\mu_{U}$. 
    In a time step starting with signed queue length $\ell$ and given the random variables $(X_\ell',Y_\ell',Z_\ell')$, the number of entanglement swaps in state $\ell$ is given by $E_\ell'=X_\ell' + Y_\ell'$. It thus follows by \cref{eq: U to E} that
    \begin{equation}
        \EE(U_\ell') = \EE(E_\ell')p_\mathrm{swap}=\sum_{(x,y,z)\in \NN^3_{\geq 0}}p_\ell(x,y,z) (x+y) p_\mathrm{swap}. 
    \end{equation}
    Moreover, by Lemma \ref{lemma: main replacing random variable by its expectation in ergodic theorem} and \ref{lemma: equivalence of W and W fut} we also have that $\mu_{V^\mathrm{fut}}\overset{\mathrm{a.s}}{=}\mu_V$.
    Since the expected Werner parameter produced from any entanglement swap is at most $\lambda_0$ (see \cref{eq: definition of lambda 0}), we can upper bound $\EE(V'_\ell)$ as
    \begin{equation}
        \EE(V'_\ell) \leq \sum_{(x,y,z)\in \NN^3_{\geq 0}}p_\ell(x,y,z) (x+y) \lambda_0.
    \end{equation}
    It follows that
    \begin{equation}
        w \overset{\mathrm{a.s}}{=} \frac{\mu_{V}}{\mu_{U}}= \frac{\sum_{\ell\in \mathcal{S}}\pi_\ell \EE(V'_\ell)}{\sum_{\ell\in \mathcal{S}}\pi_\ell \EE(U'_\ell)} \leq \frac{\lambda_0}{p_\mathrm{swap}}.
    \end{equation}
\end{proof}

\section{The Signed Queue Length under FxdMux and DynMux}\label{app: markov applied}
In this appendix, we describe the evolution of the signed queue length Markov chains under the FxdMux and DynMux policies. 
The exact steady-state rate and steady-state Werner parameter can be obtained from these Markov chains by solving at most $O(m)$ linear systems of dimension at most $O(m)$, where $2m$ is the number of quantum chips on the repeater node.
Additionally, we derive approximations to the steady-state rate and Werner parameter through the OSS approximation. The transition probability matrix for the underlying Markov chains for the signed queue length in the OSS approximation retains only the terms to the first order in $p$ from the exact transition probability matrices. In this sense, the OSS approximation assumes that at most one remote entangled link can be produced in each time step. Within the OSS approximation closed-form solutions for the steady-state rate and steady-state Werner parameter can be derived. In the parameter regime studied in \cref{section: numerical evaluations}, these closed-form solutions provide close approximations to the exact results. 

\subsection{Exact Analysis}\label{app: full markov chain models}
We first treat the signed queue length under the FxdMux policy and then proceed to the DynMux policy. 
Also, we use the following shorthand for the binomial probability mass function (pmf):
\begin{equation}\label{eq: binomial pmf shorthand}
    B(k;n,q)=
    \begin{cases}
        \binom{n}{k}q^k (1-q)^{n-k} & \mathrm{if\,}0\leq k\leq n \\
        0 & \mathrm{otherwise}
    \end{cases}~.
\end{equation}

\subsubsection*{FxdMux}
\textbf{Markov chain:} 
We denote the signed queue length at the beginning of the time step $n$ under the FxdMux policy as $S_n^\mathrm{fxd}$.
Then, $(S_n^\mathrm{fxd})_{n\geq1} \sim \Markov(P^\mathrm{fxd})$ with state space $\mathcal{S}^\mathrm{fxd} := \{-m,\dots, -1, 0, 1, \dots, m\}$ and the transition matrix $P^\mathrm{fxd}$ is given by
\begin{align}
      P^\mathrm{fxd}_{\ell,\ell'} &= \sum_{k=0}^{m} B(k; m,p)B(k+\ell'-\ell;m-\ell,p)\quad\mathrm{for}~\ell,\ell'\geq0, \label{eq:pfixed1}\\
      P^\mathrm{fxd}_{\ell,\ell'} &= \sum_{k=0}^{m} B(k-\ell'; m,p)B(k-\ell;m-\ell,p)\quad\mathrm{for}~\ell\geq0,\ell'\leq 0 \label{eq:pfixed2}\\
      P^\mathrm{fxd}_{\ell,\ell'} &=   P^\mathrm{fxd}_{-\ell,-\ell'} \label{eq:pfixed3}.
\end{align}
We now explain~\eqref{eq:pfixed1}--\eqref{eq:pfixed3}.
For $\ell=0$ the chain is empty, and there are $m$ remote entanglement generation attempts to the left and $m$ to the right.
The chain evolves to state $\ell'\geq 0$ if $k$ attempts succeed on the left and $\ell'+k$ attempts succeed on the right, for some $k \in \{0,1,\dots,m\}$.
In the entanglement swap phase, $k$ links will be matched and swapped, or lost in case of swap failure, so that at the start of the next time step $\ell'$ links will remain unmatched on the right segment. 
Similarly, 
when the chain starts with links present on the right segment (i.e., $\ell>0$), 
there are $m$ remote link generation attempts on the left and $m-\ell$ attempts on the right. 
The transition to state $\ell'\geq 0$ at the start of the next time step happens whenever $k$ links are generated on the left and $k+\ell'-\ell$ links are generated on the right. 
At the start of the entanglement swap phase, there will then be $k$ links on the left and $k+\ell'$ links on the right as there were already $\ell$ links on the right. After entanglement swaps, $\ell'$ links will remain on the right segment. 

The queue can also flip from the right to the left segment.
In particular, the queue flips from $\ell\geq 0$ to $\ell'\leq 0$ when $k-\ell'$ links succeed on the left, and $k-\ell$ links succeed on the right, such that the entanglement swap phase starts with a total of $k+|\ell'|$ links on the left segment and $k$ links on the right segment. Finally, by the left-right symmetry of the system all transition probabilities starting from states $\ell\leq 0$ can be determined.

\textbf{Stationary distribution:} The Markov chain $(S_n^\mathrm{fxd})_{n\geq 1}$ is irreducible and aperiodic on a finite state space. Irreducibility and aperiodicity follow from the fact that state $\ell=0$ can be reached in one step from any state
and any state can reach state $\ell=0$ in one step. The stationary distribution $\pi^\mathrm{fxd}$ can be found by solving a linear system of dimension $|\mathcal{S}^\mathrm{fxd}|=2m+1$ \cite[Eq. 9.31]{van2009performance}.

\textbf{Steady-state rate:} 
For a state $\ell\geq 0$, the number of new links generated on the left segment during the generation phase under the FxdMux policy is $K_{L,\ell}^{\mathrm{fxd}'} \sim \Bin(m,p)$, while the number of new links on the right segment is $K_{R,\ell}^{\mathrm{fxd}'} \sim \Bin(m-\ell,p)$. With
\begin{equation}\label{eq: E fxdmux}
E_{\ell}^{\mathrm{fxd}'}=\min(K_{L,\ell}^{\mathrm{fxd}'}, \ell + K_{R,\ell}^{\mathrm{fxd}'}),
\end{equation} we then have similar to \cref{eq: U to E} that
\begin{align}\label{eq: U fxd step 1}
    \EE\left(U_\ell^\mathrm{fxd'}\right) =\EE\left(E^{\mathrm{fxd'}}_\ell\right)p_\mathrm{swap}.
\end{align}
Using the fact that for a non-negative integer-valued random variable $A$, the expectation of $A$ can be expressed as $\EE(A) = \sum_{a=1}^\infty \PP(A\geq a)$, we have
\begin{align}
    \EE\left(E^{\mathrm{fxd'}}_\ell\right)&=\sum_{k=1}^m \PP(\min(K^{\mathrm{fxd}'}_{L,\ell},\ell + K^{\mathrm{fxd}'}_{R,\ell})\geq k)= 
    \sum_{k=1}^{m}\sum_{k_L,k_R=0}^m B(k_L;m,p) B(k_R;m-\ell,p) \mathbbm{1}(\min(k_L,\ell+k_R)\geq k),\label{eq: U fxd step 2}
\end{align}
where $B(k;n,q)$ is the binomial pmf shorthand from \cref{eq: binomial pmf shorthand} and $\mathbbm{1}(\cdot)$ denotes the indicator function. The case $\ell<0$ can be derived from the left-right symmetry. 
The expected duration of a time step starting in state $\ell$ follows similar to \cref{eq: expected time step from U},
\begin{equation}
    \EE\left(T_\ell^{\mathrm{fxd'}}\right) = t_\mathrm{long} + \EE\left(U^{\mathrm{fxd'}}_\ell\right)\frac{t_\mathrm{swap}}{p_\mathrm{swap}}.
\end{equation}
Using the numerical solution for $\pi^\mathrm{fxd}$, the steady-state rate $r^\mathrm{fxd}$ under the FxdMux policy can thus be computed exactly from \cref{eq: steady-state rate steady-state average} as
\begin{equation}
    r^\mathrm{fxd}\overset{\mathrm{a.s.}}{=} \frac{\sum_{\ell\in \mathcal{S}^\mathrm{fxd}}\pi^\mathrm{fxd}_\ell\EE\left(U_{\ell}^{\mathrm{fxd'}}\right)}{\sum_{\ell\in \mathcal{S}^\mathrm{fxd}}\pi^\mathrm{fxd}_\ell\EE\left(T_{\ell}^{\mathrm{fxd'}}\right)}.
\end{equation}

\textbf{Steady-state Werner parameter:} Under the FxdMux policy, we denote the number of $X$-pairs, $Y$-pairs, and $Z$-links produced in a time step starting with signed queue length $\ell$ as $(X^{\mathrm{fxd}'}_\ell,Y^{\mathrm{fxd}'}_\ell,Z^{\mathrm{fxd}'}_\ell)$.
The expected sum of future Werner parameters starting from this state can be calculated similar to \cref{eq: expected future Werner parameter from state}, i.e.,
\begin{equation}\label{eq: expected future Werner parameter from state fxd}
    \EE\left(V_\ell^\mathrm{fut,\, fxd'}\right) = \sum_{(x,y,z) \in \NN_{\ge 0}^3}p_\ell^\mathrm{fxd}(x,y,z) \EE\left(V_\ell^\mathrm{fut,\, fxd'}\mid(x,y,z)\right),
\end{equation}
where $p_\ell^\mathrm{fxd}(x,y,z)$ and $V_\ell^\mathrm{fut,\, fxd'}|(x,y,z)$ are respectively shorthands for $\PP\big((X^{\mathrm{fxd}'}_\ell,Y^{\mathrm{fxd}'}_\ell,Z^{\mathrm{fxd}'}_\ell) = (x,y,z)\big)$ and ${V_\ell^\mathrm{fut,\, fxd'}|(X^{\mathrm{fxd}'}_\ell,Y^{\mathrm{fxd}'}_\ell,Z^{\mathrm{fxd}'}_\ell) = (x,y,z)}$ for $(x,y,z) \in \NN_{\ge 0}^3$.
The probabilities $p_\ell^\mathrm{fxd}(x,y,z)$ 
can be calculated as
\begin{equation}\label{eq: prob xyz for fxd}
    p_\ell^\mathrm{fxd}(x,y,z) = 
    \begin{cases}
    B(x;m,p) B(z;m-\ell,p) \mathbbm{1}(y=0) &    \mathrm{if\,}x<\ell, \\
    B(k_L; m,p) B(k_R;m-\ell,p) \Big(\mathbbm{1}\big(k_L\leq \ell + k_R, k_L=\ell+y, k_R = y+z \big)\,+ \\
    \quad \mathbbm{1}\big(k_L> \ell + k_R, k_L = \ell+y+z, k_R = y\big) \Big) & \mathrm{if\,} x = \ell.
    \end{cases}
\end{equation}
This is because for $x<\ell$, matching links forming $X$-pairs are generated on the empty left segment from $m$ attempts, and $Z$-links are added to the queue on the right segment. 
The number of generated $Y$-pairs is necessarily zero in this case. 
For $x=\ell$, all queued links are matched. 
In that case, the first term within the parentheses represents the event of having unmatched links on the right segment, while the second term represents unmatched links on the left segment.

The expected sum of future Werner parameters from state $\ell$ conditioned on the event ${\{(X^{\mathrm{fxd}'}_\ell,Y^{\mathrm{fxd}'}_\ell,Z^{\mathrm{fxd}'}_\ell) \!=\! (x,y,z)\}}$, i.e. $\EE(V_\ell^\mathrm{fut,\, fxd'}|(x,y,z))$, can be calculated following \cref{eq: conditional W future}.
Evaluating it requires the computation of the expected future idle depolarization $\EE\left(\Lambda^\mathrm{fxd}(a,b)\right)$ of the link in position $a$ in a queue of size $b$ until it enters entanglement swap. Under the FxdMux policy, $\Lambda^\mathrm{fxd}(a,b)$ only depends on the starting position $a$ in the queue. To see this, note that since links are matched FIFO the number of time steps $N_a^\mathrm{fxd}$ until the link starting at position $a$ is matched is a random variable whose distribution can be computed as
\begin{equation}
    \PP(N_a^\mathrm{fxd}\geq n) = \PP\Big(\sum_{i=1}^n B_i\geq a\Big), 
\end{equation}
where $B_i \overset{\mathrm{iid}}{\sim} \Bin(m,p)$ are 
binomial random variables corresponding to the number of links generated on the empty segment in step $i$ until the links are matched. This does not depend on the initial queue size $b$. The additional depolarization within the time step in which the link is matched is also independent of $b$. Hence, we may define $\Lambda^\mathrm{fxd}(a):=\Lambda^\mathrm{fxd}(a,m)$ for the idle depolarization starting from position $a$ in the queue until entanglement swap, and subsequently define the vector $(f_a)_{a=1}^m$ of the expected idle depolarization parameters from all possible starting positions through 
\begin{equation}\label{eq: depolarization vector fxd}
    f_a := \EE\left(\Lambda^\mathrm{fxd}(a)\right)\quad\mathrm{for}\quad a=1,\dots, m.
\end{equation}

It is shown in Lemma \ref{lemma: Lambda a b fxdmux} below that the vector $(f_a)_{a=1}^m$ can be obtained as the solution to a linear system of dimension $m$. After solving this linear system numerically, the expected sum of future Werner parameters conditional on $(X^{\mathrm{fxd}'}_\ell,Y^{\mathrm{fxd}'}_\ell,Z^{\mathrm{fxd}'}_\ell)$ can be computed from \cref{eq: conditional W future}, so that \cref{eq: expected future Werner parameter from state fxd} for the expected sum of future Werner parameters from state $\ell$ can be evaluated.
Using the numerical solution for $\pi^\mathrm{fxd}$, the steady-state Werner parameter $w^\mathrm{fxd}$ under the FxdMux policy can thus be computed exactly from \cref{eq: steady-state Werner steady-state average} as
\begin{equation}\label{eq:steadyStateWernerFxd}
    w^\mathrm{fxd} \overset{\mathrm{a.s.}}{=} \frac{\sum_{\ell\in \mathcal{S}^\mathrm{fxd}}\pi^\mathrm{fxd}_\ell\EE\left(V_{\ell}^{\mathrm{fut,\, fxd'}}\right)}{\sum_{\ell\in \mathcal{S}^\mathrm{fxd}}\pi^\mathrm{fxd}_\ell\EE\left(U_{\ell}^{\mathrm{fxd'}}\right)}.
\end{equation}
\begin{lemma}\label{lemma: Lambda a b fxdmux}
    Let $\fvec = (f_a)_{a=1}^m$ be the vector of expected idle depolarization parameters for links starting from position $a$ in the queue until entering entanglement swap under the FxdMux policy, as defined in \cref{eq: depolarization vector fxd}. Then,
    \begin{equation}\label{eq: linear system for depolarization fxdmux}
        (I-G)\fvec = \vvec,
    \end{equation}
    where the matrix $G=(G_{a,a'})_{a,a'=1}^m$ is given by
    \begin{equation}
    G_{a,a'} = \begin{cases}
        B(a-a';m,p)\lambda_\mathrm{long}\left(\lambda_\mathrm{local}^{(1)}\right)^{a-a'} & \mathrm{if\,} a\geq a' \\
        0 &\mathrm{otherwise}
    \end{cases}
\end{equation}
and the vector $\vvec = (v_a)_{a=1}^m$ is given by
\begin{equation}
    v_a = \sum_{k=a}^m B(k;m,p) \lambda_\mathrm{long}\left(\lambda_\mathrm{local}^{(2)}\right)^{a-1}.
\end{equation}
\end{lemma}
\begin{proof}
    Let $\fvec = (f_a)_{a=1}^m$ be the vector of expected future idle depolarization parameters entering entanglement swap starting from position $a$ in the queue under the FxdMux policy, as defined in \cref{eq: depolarization vector fxd}.
We then have the following backward recursion by conditioning on the number of links generated in the current time step:
\begin{equation}\label{eq: dep from pos a in fxdmux}
    f_a =\sum_{k=0}^{a-1} B(k;m,p)\lambda_\mathrm{long}\left(\lambda_\mathrm{local}^{(1)}\right)^{k}f_{a-k} +\sum_{k=a}^m B(k;m,p) \lambda_\mathrm{long}\left(\lambda_\mathrm{local}^{(2)}\right)^{a-1}.
\end{equation}
The first term on the RHS corresponds to $0\leq k<a$ links being generated so that the link at position $a$ remains unmatched, but moves down to position $a-k$ in the queue. In this case, the expected depolarization in the current step is given by $\lambda_\mathrm{long}(\lambda_\mathrm{local}^{(1)})^{k}$.
The additional future depolarization of this link has the same distribution as $\Lambda^\mathrm{fxd}(a-k)$.
The second term gives the depolarization when $a$ or more links are generated.
The factor $\lambda_\mathrm{long}$ takes into account depolarization of the link at position $a$ during the generation phase, whereas the factor $(\lambda_\mathrm{local}^{(2)})^{a-1}$ gives the expected depolarization of the matched pair until it enters the entanglement swap.  
The backward recursion in \cref{eq: dep from pos a in fxdmux} can be reformulated as the $m$-dimensional linear system \cref{eq: linear system for depolarization fxdmux} given in the statement. 
\end{proof}

\subsubsection*{DynMux}
\textbf{Markov chain:} We denote the signed queue length at the beginning of time step $n$ under the DynMux policy as $S^\mathrm{dyn}_n$. Then, $(S_n^\mathrm{dyn})_{n\geq1} \sim \Markov(P^\mathrm{dyn})$ on state space
$\mathcal{S}^\mathrm{dyn} = \{-2m+2,\dots, -1, 0, 1, \dots, 2m-2\}.$
The maximum queue size $2m-2$ under the DynMux policy is achieved as follows. 
If there is one entangled link to the left or right end node, then $2m-1$ attempts are done to the other end node. If all these attempts succeed, then the queue size at the start of the next time step is $2m-2$ because one pair of links was consumed by an entanglement swap. The transition probabilities for the signed queue length under the DynMux policy are as follows:
\begin{align}
      P^\mathrm{dyn}_{0,\ell'} &= \sum_{k=0}^{m} B(k; m,p)B(k+\ell';m,p) \quad\mathrm{for}\quad \ell'\geq0, \\
      P^\mathrm{dyn}_{\ell,\ell'} &= B(\ell-\ell'; 2m-\ell,p)\quad\mathrm{for}\quad \ell>0 \\
      P^\mathrm{dyn}_{\ell,\ell'} &=   P_{-\ell,-\ell'}^\mathrm{dyn}.
\end{align}
Here, the case $\ell=0$ is the same as for FxdMux, since in that case both policies evenly distribute $m$ attempts to the left and $m$ attempts to the right. If a time step starts in state $\ell>0$ so that there are links in queue on the right segment, then all remaining $2m-\ell$ chips attempt entanglement generation with the left end node. 
The chain transitions to state $\ell'$ when $\ell-\ell'$ of these attempts succeed, where we note that $\ell'$ may take on negative values as well. The case $\ell\leq 0$ follows by left-right symmetry.

\textbf{Stationary distribution:} The Markov chain $(S^\mathrm{dyn}_n)_{n\geq 1}$ is irreducible and aperiodic. Aperiodicity is immediate since any state can transition to itself in one step. For irreducibility, note that any state $\ell\in \{0,1,\dots, 2m-2\}$ can be reached in a single step from state $-1$ since in this case $2m-1$ attempts are performed to the right and any number may succeed yielding at most $2m-2$ links on the right segment after entanglement swapping. By symmetry any state in $\ell\in \{-2m+2,\dots, -1,0\}$ can be reached from $\ell=1$. Moreover, there is a positive probability to transition from state $\ell$ with $|\ell|\geq 1$ to state $0$ in $|\ell|$ steps by generating a single link in every time step. Since both state $-1$ and state $1$ can be reached from state $0$ in one step, irreducibility follows. Since the state space is finite, the stationary distribution $\pi^\mathrm{dyn}$ is the limiting distribution, and it can be found as the solution to a linear system of dimension $|\mathcal{S}^\mathrm{dyn}|=4m-3$ \cite[Eq. 9.31]{van2009performance}.


\textbf{Steady-state rate:}
The number of new links generated on the left segment under the DynMux policy in state $\ell$ is distributed as 
\begin{equation}
    K_{L,\ell}^{\mathrm{dyn'}} \sim 
    \begin{cases}
        \Bin(m,p) & \mathrm{if\,} \ell=0 \\
        \Bin(2m-\ell, p) & \mathrm{if\,} \ell>0
    \end{cases} \quad \text{and} \quad K_{L,\ell}^{\mathrm{dyn'}} = 0~~ \text{if}~\ell<0,
\end{equation}
and the number of new links generated on the right segment is distributed as
\begin{equation}
    K_{R,\ell}^{\mathrm{dyn'}} \sim 
    \begin{cases}
        \Bin(2m+\ell, p) & \mathrm{if}\,\ell<0 \\
        \Bin(m,p) & \mathrm{if\,} \ell=0 
    \end{cases}
     \quad \text{and} \quad K_{R,\ell}^{\mathrm{dyn'}} = 0~~ \text{if}~\ell>0,
\end{equation}
The number of entanglement swaps attempted in state $\ell$ is then 
\begin{equation}\label{eq: E dynmux}
E_{\ell}^{\mathrm{dyn'}} =
    \begin{cases}
        \min(-\ell, K_{R,\ell}^{\mathrm{dyn'}}) & \mathrm{if}\,\ell<0 \\
        \min(K_{L,0}^{\mathrm{dyn'}},K_{R,0}^{\mathrm{dyn'}}) & \mathrm{if\,} \ell=0 \\
        \min(K_{L,\ell}^{\mathrm{dyn'}}, \ell)  & \mathrm{if\,} \ell>0
    \end{cases},
\end{equation}
and the expected number of end-to-end links produced in state $\ell$ follows similarly to \cref{eq: U to E},
\begin{equation}
    \EE\left(U_\ell^{\mathrm{dyn'}}\right) =  \EE\left(E_{\ell}^{\mathrm{dyn'}}\right)
    p_\mathrm{swap}.
\end{equation}
The expected duration of the time step starting in state $\ell$ is given in terms of $\EE(U_\ell^{\mathrm{dyn'}})$ through \cref{eq: expected time step from U} to be
\begin{equation}
    \EE\left(T_\ell^{\mathrm{dyn'}}\right) = t_\mathrm{long} + \EE\left(U^{\mathrm{dyn'}}_\ell\right)\frac{t_\mathrm{swap}}{p_\mathrm{swap}}.
\end{equation}
After numerically solving for the stationary distribution $\pi^\mathrm{dyn}$, the steady-state rate $r^\mathrm{dyn}$ under the DynMux policy can be computed exactly from \cref{eq: steady-state rate steady-state average} as
\begin{equation}
    r^\mathrm{dyn}\overset{\mathrm{a.s.}}{=} \frac{\sum_{\ell\in \mathcal{S}^\mathrm{dyn}}\pi^\mathrm{dyn}_\ell\EE\left(U_{\ell}^{\mathrm{dyn'}}\right)}{\sum_{\ell\in \mathcal{S}^\mathrm{dyn}}\pi^\mathrm{dyn}_\ell\EE\left(T_{\ell}^{\mathrm{dyn'}}\right)}.
\end{equation}

\textbf{Steady-state Werner parameter:} Under the DynMux policy, we denote the number of $X$-pairs, $Y$-pairs, and $Z$-links produced in a time step starting with signed queue length $\ell$ as $(X^{\mathrm{dyn}'}_\ell,Y^{\mathrm{dyn}'}_\ell,Z^{\mathrm{dyn}'}_\ell)$.
The expected sum of future Werner parameters starting from this state can then be calculated similar to \cref{eq: expected future Werner parameter from state}, i.e.,
\begin{equation}\label{eq: expected future Werner parameter from state dyn}
    \EE\left(V_\ell^\mathrm{fut,\, dyn'}\right) = \sum_{(x,y,z) \in \NN_{\ge 0}^3}p_\ell^\mathrm{dyn}(x,y,z) \EE\left(V_\ell^\mathrm{fut,\, dyn'}\mid(x,y,z)\right),
\end{equation}
where $p_\ell^\mathrm{dyn}(x,y,z)$ and $V_\ell^\mathrm{fut,\, dyn'}|(x,y,z)$ are respectively shorthands for $\PP\big((X^{\mathrm{dyn}'}_\ell,Y^{\mathrm{dyn}'}_\ell,Z^{\mathrm{dyn}'}_\ell) = (x,y,z)\big)$ and ${V_\ell^\mathrm{fut,\, dyn'}|(X^{\mathrm{dyn}'}_\ell,Y^{\mathrm{dyn}'}_\ell,Z^{\mathrm{dyn}'}_\ell) = (x,y,z)}$ for $(x,y,z) \in \NN_{\ge 0}^3$.
For $\ell>0$ the probabilities $p_\ell^\mathrm{dyn}(x,y,z)$ can be expressed as
\begin{equation}\label{eq: prob xyz for dyn}
    p_\ell^\mathrm{dyn}(x,y,z) = 
    \begin{cases}
    B(k_L;2m-\ell,p) \indicator(x=k_L)\indicator(y=0) \indicator(z=0) &    \mathrm{if\,}x<\ell, \\
    B(k_L;2m-\ell,p) \indicator(y=0)\indicator(\ell+z = k_L) &    \mathrm{if\,}x=\ell.
    \end{cases}
\end{equation}
This is because for $x<\ell$, matching links for $X$-pairs are generated on the empty left segment from $2m-\ell$ attempts, and $y$ and $z$ are necessarily zero since no new links can be generated on the right segment when $\ell>0$. 
For $x=\ell$, all links in the queue on the right segment are matched. 
Any additional link is a $Z$-link because it starts a new queue on the left segment. 
The case $\ell<0$ follows by left-right symmetry.
Since FxdMux and DynMux coincide for the empty chain, we have 
\begin{equation}
    p_\ell^\mathrm{dyn}(x,y,z)=p_\ell^\mathrm{fxd}(x,y,z) \quad \text{for} \quad \ell=0.
\end{equation}

The expected sum of future Werner parameters from state $\ell$ conditioned on the event ${\{(X^{\mathrm{dyn}'}_\ell,Y^{\mathrm{dyn}'}_\ell,Z^{\mathrm{dyn}'}_\ell) \!=\! (x,y,z)\}}$, i.e. $\EE(V_\ell^\mathrm{fut,\, dyn'}|(x,y,z))$, can be calculated following Eq.~(\ref{eq: conditional W future}--\ref{eq: definition of lambda 0}).
Evaluating it requires the computation of the expected future idle depolarization $\EE\left(\Lambda^\mathrm{dyn}(a,b)\right)$ of the link in position $a$ in a queue of size $b$ under the DynMux policy until it enters entanglement swap. 
Unlike for FxdMux, the $\Lambda^{\mathrm{dyn}}(a,b)$ does depend on the initial queue size $b$. However, as the queue is emptied in FIFO order, the difference $d:=b-a$ between the starting position of the link and the total queue length will stay constant until the link is matched. For each $d=0,\dots, 2m-3$, we therefore define a vector $\phivec^{(d)}=(\phi^{(d)}_a)_{a=1}^{2m-2-d}$ through
\begin{equation}\label{eq: depolarization vector dyn}
    \phi_{a}^{(d)} := \EE\left(\Lambda^\mathrm{dyn}(a, a+d)\right)\qquad \mathrm{for}\quad a=1,\dots, 2m-2-d.
\end{equation}
The dimension of the vectors $\phivec^{(d)}$ follows from the fact that the total queue size is at most the maximum queue size, i.e. $a+d$ is at most $2m-2$.

In Lemma \ref{lemma: Lambda a b dynmux} we show that each $\phivec^{(d)}$ is the solution to an independent linear system of dimension $2m-2-d$. After solving these $2m-2$ linear system numerically, the expected sum of future Werner parameters conditional on $(X^{\mathrm{dyn}'}_\ell,Y^{\mathrm{dyn}'}_\ell,Z^{\mathrm{dyn}'}_\ell)$ can be computed from \cref{eq: conditional W future}, so that \cref{eq: expected future Werner parameter from state dyn} for the expected sum of future Werner parameters from state $\ell$ can be evaluated.
Using the numerical solution for $\pi^\mathrm{dyn}$, the steady-state Werner parameter $w^\mathrm{dyn}$ under the DynMux policy can thus be computed exactly from \cref{eq: steady-state Werner steady-state average} as
\begin{equation}\label{eq:steadyStateWernerDyn}
    w^\mathrm{dyn} \overset{\mathrm{a.s.}}{=} \frac{\sum_{\ell\in \mathcal{S}^\mathrm{dyn}}\pi^\mathrm{dyn}_\ell\EE\left(V_{\ell}^{\mathrm{fut,\,dyn'}}\right)}{\sum_{\ell\in \mathcal{S}^\mathrm{dyn}}\pi^\mathrm{dyn}_\ell\EE\left(U_{\ell}^{\mathrm{dyn'}}\right)}.
\end{equation}

\begin{lemma}\label{lemma: Lambda a b dynmux}
    Let $d\in \{0,\dots, 2m-3\}$ and let $\phivec^{(d)} = (\phi^{(d)}_a)_{a=1}^{2m-2-d}$ be the vector of expected future idle depolarization parameters incurred by a link starting from position $a$ in the queue with initial queue size $a+d$ until entering entanglement swapping under the DynMux policy, as defined in \cref{eq: depolarization vector dyn}. Then,
    \begin{equation}\label{eq: linear system for depolarization dynmux}
        (I-\Gamma^{(d)})\phivec^{(d)} = \nuvec^{(d)},
    \end{equation}
    where the matrix $\Gamma^{(d)}=(\Gamma^{(d)}_{a,a'})_{a,a'=1}^{2m-2-d}$ is given by
    \begin{equation}
    \Gamma^{(d)}_{a,a'} = \begin{cases}
        B(a-a';2m-a-d,p)\lambda_\mathrm{long}\left(\lambda_\mathrm{local}^{(1)}\right)^{a-a'} & \mathrm{if\,} a\geq a' \\
        0 &\mathrm{otherwise}
    \end{cases}
    \end{equation}
and the vector $\nuvec^{(d)} = (\nu^{(d)}_a)_{a=1}^{2m-2-d}$ is given by
\begin{equation}
    \nu^{(d)}_a = \sum_{k=a}^{2m-a-d} B(k;2m-a-d,p) \lambda_\mathrm{long}\left(\lambda_\mathrm{local}^{(2)}\right)^{a-1}.
\end{equation}
\end{lemma}
\begin{proof}
Let $d\in \{0,\dots, 2m-3\}$ and let $\phivec^{(d)} = (\phi^{(d)}_a)_{a=1}^{2m-2-d}$ be the vector of expected future idle depolarization parameters incurred by a link starting from position $a$ in the queue with initial queue size $a+d$ until entering entanglement swapping under the DynMux policy.
Then we have the backward recursion
\begin{equation}\label{eq: depolarization recursion for dynmux}
    \phi_a^{(d)}= 
    \sum_{k=0}^{a-1}B(k;2m-a-d,p)\lambda_\mathrm{long}\left(\lambda_\mathrm{local}^{(1)}\right)^{k} \phi_{a-k}^{(d)}
    +
    \sum_{k=a}^{2m-a-d} B(k;2m-a-d,p)\lambda_\mathrm{long}\left(\lambda_\mathrm{local}^{(2)}\right)^{a-1},
\end{equation}
where the first term on the RHS corresponds to the case that less than $a$ new links are generated so that in the next time step the link of interest starts at position $a-k$ in a queue of size $a-k+d$ from where the expected idle depolarization until entanglement swapping is $\phi_{a-k}^{(d)}$. The second term corresponds to the case that $a$ or more links are generated so that the link at position $a$ is matched. It depolarizes by $\lambda_\mathrm{long}$ because it was stored for a Sync-Gen cycle, and it is expected to experience a further depolarization $(\lambda_\mathrm{local}^{(2)})^{a-1}$ while waiting for the pairs that precede it in queue to be swapped. 
The backward recursion in \cref{eq: depolarization recursion for dynmux} can be reformulated as the linear system in the statement.
\end{proof}

\subsection{OSS Approximation}\label{app: oss markov chain models}

In this section we provide the details of the OSS approximation introduced in \cref{section: markov chain oss}. Closed-form solutions for the steady-state rate and steady-state Werner parameter in the OSS approximations for FxdMux and DynMux are obtained and the results are collected in \cref{tab:OSS approximations rate and werner}. We assume that $2mp<1$.


\subsubsection*{FxdMux OSS}
\textbf{Markov chain:} We start by describing the OSS Markov chain for the signed queue length under the FxdMux policy.
We denote the signed queue length at the beginning of the time step $n$ under the FxdMux policy in the OSS approximation by $\widetilde S_n^\mathrm{fxd}$.
Then, $(\widetilde S_n^\mathrm{fxd})_{n\geq1} \sim \Markov(\widetilde P^\mathrm{fxd})$ with state space $\mathcal{\widetilde S}^\mathrm{fxd} := \{-m,\dots, -1, 0, 1, \dots, m\}$ and the transition probability matrix $\widetilde P^\mathrm{fxd}$ is obtained from the transition matrix of the exact Markov chain by retaining only the terms up to first order in $p$. In this sense only transitions involving at most one new link are taken into account. Retaining only the terms up to first order in $p$ in \cref{eq:pfixed1} and \cref{eq:pfixed2}, the nonzero transition probabilities in the OSS approximation for $\ell\geq 0$ are
\begin{equation}\label{eq: oss fxd tpm}
    \widetilde P_{\ell,\ell-1}^\mathrm{fxd} = mp, \qquad \widetilde P_{\ell,\ell}^\mathrm{fxd} = 1-(2m-\ell)p, \qquad\mathrm{and}\qquad \widetilde P_{\ell,\ell+1}^\mathrm{fxd} = (m-\ell)p,
\end{equation}
and the nonzero transition probabilities for $\ell<0$ are defined by the symmetry $\widetilde P^\mathrm{fxd}_{\ell,\ell'}=\widetilde P^\mathrm{fxd}_{-\ell,-\ell'}$. 

\textbf{Stationary distribution:} The stationary distribution for the OSS approximation of the FxdMux policy is given by
\begin{equation}\label{eq: steady-state oss fxdmux}
        \widetilde\pi_\ell^\mathrm{fxd} = \left[1+2\sum_{k=1}^m\frac{m!}{(m-k)!m^{k}}\right]^{-1}  \frac{m!}{(m-\ell)!m^{\ell}} \qquad \mathrm{for}\quad \ell \geq 0,
\end{equation}
and for $\ell<0$ it follows from the left-right symmetry. This can be verified by directly checking the stationary equation $\widetilde\pi_{\ell'}^\mathrm{fxd} = \sum_{\ell=-m}^m\widetilde\pi_{\ell}^\mathrm{fxd} \widetilde P^\mathrm{fxd}_{\ell,\ell'}$ and the normalization condition $\sum_{\ell={-m}}^m \widetilde\pi_\ell^\mathrm{fxd}=1$. 
For further use we define the shorthand 
\begin{equation}\label{eq: sigma shorthand}
    \sigma(m) := \sum_{\ell=1}^m\frac{m!}{(m-\ell)!m^{\ell}}.
\end{equation}

\textbf{Steady-state rate:}
We let $\widetilde U_\ell^\mathrm{fxd'}$ denote the number of end-to-end links produced in a time step starting in state $\ell$ under the FxdMux policy in the OSS approximation. In the OSS approximation there can be at most one entanglement swap in a time step. Hence, similar to \cref{eq: number of links RV ell}, we have 
\begin{equation}
    \widetilde U^{\mathrm{fxd}'}_\ell \overset{d}{=}\indicator(\widetilde E^{\mathrm{fxd}'}_\ell=1)\indicator(N'_{\mathrm{local}}\leq c_\mathrm{local}),
\end{equation}
where $\widetilde E_\ell^\mathrm{fxd'}$ is the number of entanglement swaps.
In the OSS approximation for the FxdMux policy there is an entanglement swap if and only if the absolute queue size decreases so that 
\begin{equation}
    \widetilde E_{\ell}^{\mathrm{fxd}'} \sim
    \begin{cases}
        \mathrm{Bern}\left(\widetilde P^\mathrm{fxd}_{\ell,\ell+1}\right) & \mathrm{if}\,\ell < 0 \\[8pt] 
        \mathrm{Bern}\left(\widetilde P^\mathrm{fxd}_{\ell,\ell-1}\right) & \mathrm{if}\,\ell > 0
    \end{cases}  \quad \text{and} \quad \widetilde E_{\ell}^{\mathrm{fxd}'} = 0~~ \text{if}~\ell=0.
\end{equation}

From \cref{eq: U to E} and \cref{eq: oss fxd tpm} the expected number of end-to-end links produced from a time step starting in state $\ell$ in the OSS approximation is then 
\begin{equation}
    \EE\left(\widetilde U_\ell^\mathrm{fxd'}\right) =
    \begin{cases}
        (mp) p_\mathrm{swap}  & \mathrm{if\,} \ell\neq 0 \\
        0 & \mathrm{otherwise}
    \end{cases}.
\end{equation}
Using the stationary distribution from \cref{eq: steady-state oss fxdmux} the expected number of end-to-end links produced in the steady-state in the OSS approximation is 

\begin{equation}\label{eq: U oss fxd}
    \mu_{\widetilde U^{\mathrm{fxd}}}= \sum_{\ell=-m}^m \widetilde \pi_\ell^\mathrm{fxd} \EE\left(\widetilde U_\ell^\mathrm{fxd'}\right)= \frac{2\sigma(m)}{1+2\sigma(m)}(mp)p_\mathrm{swap}
\end{equation}
The duration of the time step starting in state $\ell$ follows from \cref{eq: expected time step from U} as 
\begin{equation}\label{eq: T oss fxd}
    \mu_{\widetilde T^\mathrm{fxd}} = t_\mathrm{long} + \mu_{\widetilde U^\mathrm{fxd}}\frac{t_\mathrm{swap}}{p_\mathrm{swap}} .
\end{equation}
Using \cref{theorem: steady-state rate ergodic}, \cref{eq: U oss fxd}, and \cref{eq: T oss fxd}, the final equation for the steady-state rate in the OSS approximation for FxdMux then becomes 
\begin{align}
    \widetilde r^\mathrm{fxd} &\overset{\mathrm{a.s.}}{=} \frac{\mu_{\widetilde U^\mathrm{fxd}}}{ \mu_{\widetilde T^\mathrm{fxd}}} = \left[\frac{1+2\sigma(m)}{2\sigma(m)}\frac{t_\mathrm{long}}{mp} + t_\mathrm{swap}\right]^{-1}p_\mathrm{swap}. \label{eq: r fxd oss}
\end{align}
We note that a similar result for the steady-state rate in the OSS approximation of FxdMux was obtained previously in \cite{kunzelmann2025multiplexed}.

\textbf{Steady-state Werner parameter:} Under the FxdMux policy in the OSS approximation, we denote the number of $X$-pairs, $Y$-pairs, and $Z$-links produced in a time step starting with signed queue length $\ell$ as $(\widetilde X^{\mathrm{fxd}'}_\ell,\widetilde Y^{\mathrm{fxd}'}_\ell,\widetilde Z^{\mathrm{fxd}'}_\ell)$.
The expected sum of future Werner parameters starting from this state then can be calculated similar to \cref{eq: expected future Werner parameter from state}, i.e.,
\begin{equation}\label{eq: expected future Werner parameter from state fxd oss}
    \EE\left(\widetilde V_\ell^\mathrm{fut,\, fxd'}\right) = \sum_{(x,y,z) \in \NN_{\ge 0}^3}\widetilde p_\ell^\mathrm{fxd}(x,y,z) \EE\left(\widetilde V_\ell^\mathrm{fut,\, fxd'}\mid(x,y,z)\right),
\end{equation}
where $\widetilde p_\ell^\mathrm{fxd}(x,y,z)$ and $\widetilde V_\ell^\mathrm{fut,\, fxd'}|(x,y,z)$ are respectively shorthands for $\PP\big((\widetilde X^{\mathrm{fxd}'}_\ell,\widetilde Y^{\mathrm{fxd}'}_\ell,\widetilde Z^{\mathrm{fxd}'}_\ell) = (x,y,z)\big)$ and ${\widetilde V_\ell^\mathrm{fut,\, fxd'}|(\widetilde X^{\mathrm{fxd}'}_\ell,\widetilde Y^{\mathrm{fxd}'}_\ell,\widetilde Z^{\mathrm{fxd}'}_\ell) = (x,y,z)}$ for $(x,y,z) \in \NN_{\ge 0}^3$. In the OSS approximation, at most one link is generated in every time step so that $\widetilde Y^{\mathrm{fxd}'}_\ell=0$ for all $\ell$, while $\widetilde X^{\mathrm{fxd'}}_\ell$ and $\widetilde Z^{\mathrm{fxd'}}_\ell$ are Bernoulli random random variables.
Since the future Werner parameter only has contributions from $Z$-links, we have that
\begin{equation}
    \EE\left(\widetilde V_\ell^\mathrm{fut,\, fxd'}\right) =\widetilde p_\ell^{\mathrm{fxd}}(0,0,1)\EE\left(\widetilde V^\mathrm{fut,\, fxd'}_\ell\mid(0,0,1)\right).
\end{equation}
The probability to generate one $Z$-link is equal to the probability that the queue length increases by $1$. Hence, 
\begin{equation}
    \widetilde p_\ell^\mathrm{fxd}(0,0,1)=\begin{cases}
        \widetilde P^\mathrm{fxd}_{\ell,\ell-1} & \mathrm{if}\,-m<\ell<0\\
        \widetilde P_{0,1}^\mathrm{fxd}+\widetilde P_{0,-1}^\mathrm{fxd} & \mathrm{if\,}\ell=0\\
        \widetilde P^\mathrm{fxd}_{\ell,\ell+1} & \mathrm{if}\,0<\ell<m
    \end{cases}.
\end{equation}
The expected future Werner parameter conditional on one $Z$-link, $\EE(\widetilde V^\mathrm{fut,\, fxd'}_\ell\mid(0,0,1))$, is of the form of \cref{eq: conditional W future} and in the OSS approximation it becomes
\begin{equation}
    \EE\left(\widetilde V^\mathrm{fut,\, fxd'}_\ell\mid(0,0,1)\right) = \EE\left(\widetilde\Lambda^\mathrm{fxd}(|\ell|+1)\right) \lambda_0,
\end{equation}
where $\EE(\widetilde\Lambda^{\mathrm{fxd}}(|\ell|+1))$ is the expected idle depolarization for a link starting in position $|\ell|+1$ in the queue until it gets matched in the OSS approximation for FxdMux. We have that 
\begin{equation}
\EE\left(\widetilde\Lambda^\mathrm{fxd}({|\ell|+1})\right) = \left(\lambda_\mathrm{local}^{(1)}\right)^{|\ell|}\left(\EE\left(\lambda_\mathrm{long}^{\widetilde D^\mathrm{fxd}}\right)\right)^{|\ell|+1},
\end{equation}
where $\widetilde D^\mathrm{fxd}\sim \Geo(mp)$ is the number of time steps required to generate a matching link for the first link in the queue in the OSS approximation for FxdMux, and $\lambda_\mathrm{local}^{(1)} $ is the depolarization incurred by a link while waiting for entanglement swaps between another pair of links to be completed. After matching and swapping the first $|\ell|$ links in the queue, the link that started at position $|\ell|+1$ has become the first link in the queue. This gives the factors $(\lambda_\mathrm{local}^{(1)})^{|\ell|}(\EE(\lambda_\mathrm{long}^{\widetilde D^\mathrm{fxd}}))^{|\ell|}$. It then has to wait another $\widetilde D^\mathrm{fxd}$ steps until being matched itself, which gives the final factor of $\EE(\lambda_\mathrm{long}^{\widetilde D^\mathrm{fxd}})$.

The steady-state state average of the future Werner parameter in the OSS approximation for FxdMux then becomes
\begin{align}\label{eq: V oss fxd}
    \mu_{\widetilde V^\mathrm{fut,\, fxd}} &= \sum_{\ell=-m}^m \widetilde\pi_{\ell}^\mathrm{fxd}\EE\left(\widetilde V_\ell^\mathrm{fut,\, fxd'}\right) = 2\sum_{\ell=0}^{m-1} \widetilde\pi_{\ell}^\mathrm{fxd}\widetilde P^\mathrm{fxd}_{\ell,\ell+1}\left(\lambda_\mathrm{local}^{(1)}\right)^\ell\left(\EE\left(\lambda_\mathrm{long}^{\widetilde D^\mathrm{fxd}}\right)\right)^{\ell+1}\lambda_0,
\end{align}
where we have used the left-right symmetry.
It follows from \cref{theorem: main steady-state werner fut ergodic}, \cref{eq: U oss fxd}, and \cref{eq: V oss fxd}, that the steady-state Werner parameter in the OSS approximation for FxdMux is
\begin{align}
    \widetilde w^\mathrm{fxd} &\overset{\mathrm{a.s}}{=} \frac{\mu_{\widetilde V^{\mathrm{fut,\, fxd}}}}{\mu_{\widetilde U^\mathrm{fxd}}} \\
    &= \frac{2\sum_{\ell=0}^{m-1}\widetilde\pi_{\ell}^\mathrm{fxd} \widetilde P^\mathrm{fxd}_{\ell,\ell+1}\left(\lambda_\mathrm{local}^{(1)}\right)^\ell\EE\left(\lambda_\mathrm{long}^{\widetilde D^\mathrm{fxd}}\right)^{\ell+1}\lambda_0}{\frac{2\sigma(m)}{1+2\sigma(m)}(mp)p_\mathrm{swap}}
    \\
    &=   \sum_{\ell=0}^{m-1}\frac{1}{\sigma(m)}\frac{(m-1)!}{(m-\ell-1)!m^\ell}\left(\lambda_\mathrm{local}^{(1)}\right)^\ell\EE\left(\lambda_\mathrm{long}^{\widetilde D^\mathrm{fxd}}\right)^{\ell+1}\frac{\lambda_0}{p_\mathrm{swap}} \label{eq: w fxd oss}
\end{align}
where for the final equality we used that $\widetilde\pi_\ell^\mathrm{fxd}(1+2\sigma(m)) = m!/((m-\ell)!m^\ell)$ by \cref{eq: steady-state oss fxdmux} and $\widetilde P^\mathrm{fxd}_{\ell,\ell+1}= (m-\ell)p$ for $\ell\geq 0$ by \cref{eq: oss fxd tpm}.


\subsubsection*{DynMux OSS} 

\textbf{Markov chain:} We denote the signed queue length at the beginning of the time step $n$ under the DynMux policy in the OSS approximation as $\widetilde S_n^\mathrm{dyn}$. Then, $(\widetilde S_n^\mathrm{dyn})_{n\geq1} \sim \Markov(\widetilde P^\mathrm{dyn})$ with state space $\mathcal{\widetilde S}^\mathrm{dyn} := \{-1, 0, 1\}$ and the transition matrix $\widetilde P^\mathrm{dyn}$ is given by
\begin{equation}\label{eq: tpm oss dyn}
    \widetilde P^\mathrm{dyn} = 
    \begin{pmatrix}
    1-(2m-1)p & (2m-1)p & 0 \\ 
    mp & 1 - 2mp & mp \\ 
    0 & (2m-1)p & 1 - (2m-1)p    
    \end{pmatrix}.
\end{equation}
This transition probability matrix is the first order $p$ approximation of the exact transition probability matrix $P^\mathrm{dyn}$ restricted to $\{-1,0,1\}$. The reason for restricting the state space to the states $\{-1,0,1\}$ in the OSS approximation for DynMux is that these are the only recurrent states after approximating $P^\mathrm{dyn}$ to first order in $p$. 
    
\textbf{Stationary distribution:} The Markov chain $(\widetilde S^\mathrm{dyn})_{n\geq 1}$ is irreducible on a finite state space. The stationary distribution is given by 
    \begin{equation}
        \widetilde \pi_0^\mathrm{dyn} = \frac{2m-1}{4m-1} \qquad\mathrm{and}\qquad \widetilde \pi_{1}^\mathrm{dyn}=\widetilde \pi_{-1}^\mathrm{dyn} = \frac{m}{4m-1},
    \end{equation}
    as may be directly verified by checking the stationarity equation $\widetilde\pi_{\ell'}^\mathrm{dyn} = \sum_{\ell=-1}^1\widetilde\pi_{\ell}^\mathrm{dyn} \widetilde P^\mathrm{dyn}_{\ell,\ell'}$ and the normalization condition $\sum_{\ell=-1}^1 \widetilde\pi_\ell^\mathrm{dyn}=1$. 

\textbf{Steady-state rate:}
In the OSS approximation for DynMux, there is an entanglement swap if and only if the absolute queue size decreases. Hence, the number of entanglement swaps in a time step starting from state $\ell$ under the DynMux policy is
\begin{equation}
    \widetilde E_{\ell}^{\mathrm{dyn}'} \sim 
      \begin{cases}
        \mathrm{Bern}\left(\widetilde P^\mathrm{dyn}_{-1,0}\right) & \mathrm{if}\,\ell=-1 \\
        \mathrm{Bern}\left(\widetilde P^\mathrm{dyn}_{1,0}\right) & \mathrm{if}\,\ell =1
    \end{cases}   \quad \text{and} \quad \widetilde E_{\ell}^{\mathrm{dyn}'} = 0~~ \text{if}~\ell=0,
\end{equation}
and similar to \cref{eq: number of links RV ell} we have $\widetilde U^{\mathrm{dyn}'}_\ell \overset{d}{=}\indicator(\widetilde E^{\mathrm{dyn}'}_\ell=1)\indicator(N'_{\mathrm{local}}\leq c_\mathrm{local})$.
From \cref{eq: U to E} and \cref{eq: tpm oss dyn} the expected number of end-to-end links produced from a time step starting in state $\ell$ in the OSS approximation is then 
\begin{equation}
    \EE\left(\widetilde U_\ell^\mathrm{dyn'}\right) =
    \begin{cases}
        ((2m-1)p) p_\mathrm{swap}  & \mathrm{if\,} \ell\neq 0 \\
        0 & \mathrm{otherwise}
    \end{cases}.
\end{equation}
    The expected number of end-to-end links produced in the steady state in the OSS approximation is thus
    \begin{equation}\label{eq: oss U dyn}
        \mu_{\widetilde U^\mathrm{dyn}} = \sum_{\ell=-1}^1 \widetilde\pi^\mathrm{dyn}_\ell \EE\left(\widetilde U_\ell^\mathrm{dyn'}\right)= \frac{2m(2m-1)p}{4m-1}p_\mathrm{swap}
    \end{equation}
    Further, the expected duration of the time step follows using \cref{eq: expected time step from U} as
\begin{equation}\label{eq: oss T dyn}
    \mu_{\widetilde T^\mathrm{dyn}} = t_\mathrm{long} + \mu_{\widetilde U^\mathrm{dyn}}\frac{t_\mathrm{swap}}{p_\mathrm{swap}}.
\end{equation}
    It follows from \cref{theorem: steady-state rate ergodic}, \cref{eq: oss U dyn}, and \cref{eq: oss T dyn} that the steady-state rate for DynMux in the OSS approximation is
    \begin{align}
        \widetilde r^\mathrm{dyn} &\overset{\mathrm{a.s.}}{=} \frac{\mu_{\widetilde U^\mathrm{dyn}}}{\mu_{\widetilde T^\mathrm{dyn}}} = 
        \left[\frac{4m-1}{2m}\frac{t_\mathrm{long}}{(2m-1)p}+t_\mathrm{swap}\right]^{-1}p_\mathrm{swap}. \label{eq: r dyn oss}
    \end{align}

\textbf{Steady-state Werner parameter:}
Under the DynMux policy in the OSS approximation, we denote the number of $X$-pairs, $Y$-pairs, and $Z$-links produced in a time step starting with signed queue length $\ell$ as $(\widetilde X^{\mathrm{dyn}'}_\ell,\widetilde Y^{\mathrm{dyn}'}_\ell,\widetilde Z^{\mathrm{dyn}'}_\ell)$.
The expected sum of future Werner parameters starting from this state then can be calculated similar to \cref{eq: expected future Werner parameter from state}, i.e.,
\begin{equation}\label{eq: expected future Werner parameter from state dyn oss}
     \EE\left(\widetilde V_\ell^\mathrm{fut,\, dyn'}\right) = \sum_{(x,y,z) \in \NN_{\ge 0}^3}\widetilde p_\ell^\mathrm{dyn}(x,y,z) \EE\left(\widetilde V_\ell^\mathrm{fut,\, dyn'}\mid(x,y,z)\right),
\end{equation}
where $\widetilde p_\ell^\mathrm{dyn}(x,y,z)$ and $\widetilde V_\ell^\mathrm{fut,\, dyn'}|(x,y,z)$ are respectively shorthands for $\PP\big((\widetilde X^{\mathrm{dyn}'}_\ell,\widetilde Y^{\mathrm{dyn}'}_\ell,\widetilde Z^{\mathrm{dyn}'}_\ell) = (x,y,z)\big)$ and ${\widetilde V_\ell^\mathrm{fut,\, dyn'}|(\widetilde X^{\mathrm{dyn}'}_\ell,\widetilde Y^{\mathrm{dyn}'}_\ell,\widetilde Z^{\mathrm{dyn}'}_\ell) = (x,y,z)}$ for $(x,y,z) \in \NN_{\ge 0}^3$. In the OSS approximation, at most one link is generated in every time step so that $\widetilde Y^{\mathrm{dyn}'}_\ell=0$ for all $\ell$, while $\widetilde X^{\mathrm{dyn'}}_\ell$ and $\widetilde Z^{\mathrm{dyn'}}_\ell$ are Bernoulli random random variables.
Since the future Werner parameter only has contributions from $Z$-links, we have that
\begin{equation}
     \EE\left(\widetilde V_\ell^\mathrm{fut,\, dyn'}\right) =\widetilde p_\ell^{\mathrm{dyn}}(0,0,1)\EE\left(\widetilde V^\mathrm{fut,\,dyn'}_\ell\mid(0,0,1)\right).
\end{equation}
The probability to generate one $Z$-link is equal to the probability that the queue length increases by $1$. For the DynMux policy in the OSS approximation this can only happen in state $\ell=0$, so that
\begin{equation}
    \widetilde p_\ell^\mathrm{dyn}(0,0,1)=\begin{cases}
    \widetilde P_{0,1}^\mathrm{dyn}+\widetilde P_{0,-1}^\mathrm{dyn} & \mathrm{if\,}\ell=0\\
        0 & \mathrm{otherwise}
    \end{cases}.
\end{equation}
The expected future Werner parameter conditional on one $Z$-link, $\EE(\widetilde V^\mathrm{fut,\, dyn'}_\ell\mid(0,0,1))$, is of the form of \cref{eq: conditional W future}. For $\ell=0$ in the OSS approximation it becomes
\begin{equation}
    \EE\left(\widetilde V^\mathrm{fut,\,dyn'}_0\mid(0,0,1)\right) = \EE\left(\widetilde\Lambda^\mathrm{dyn}(1,1)\right) \lambda_0,
\end{equation}
where $\EE(\widetilde\Lambda^{\mathrm{dyn}}(1,1))$ is the expected future idle depolarization for a link starting in position $1$ in a queue of size $1$ until it gets matched under the OSS approximation for the DynMux policy. The number of time steps until the link in queue gets matched is $\widetilde D^\mathrm{dyn}\sim \Geo((2m-1)p)$. Hence,
\begin{equation}
\EE\left(\widetilde\Lambda^\mathrm{dyn}(1,1)\right) = \EE\left(\lambda_\mathrm{long}^{\widetilde D^\mathrm{dyn}}\right).
\end{equation}

The steady-state state average of the future Werner parameter in the OSS approximation for DynMux then becomes
\begin{align}\label{eq: oss V dyn}
    \mu_{\widetilde V^\mathrm{fut,\, dyn}} &= \widetilde \pi_0^\mathrm{dyn}  \EE\left(\widetilde V_0^\mathrm{fut,\, dyn'}\right) = \frac{(2m-1)2mp}{4m-1}\EE\left(\lambda_\mathrm{long}^{\widetilde D^\mathrm{dyn}}\right)\lambda_0,
\end{align}
where we've used that $\widetilde P_{0,1}^{\mathrm{dyn}} +\widetilde P_{0,-1}^{\mathrm{dyn}} =2mp.$
It follows from \cref{theorem: main steady-state werner fut ergodic}, \cref{eq: oss U dyn}, and \cref{eq: oss V dyn} that the steady-state Werner parameter in the OSS approximation for DynMux is
\begin{align}
    \widetilde w^\mathrm{dyn} &\overset{\mathrm{a.s}}{=} \frac{\mu_{\widetilde V^{\mathrm{fut,\,dyn}}}}{\mu_{\widetilde U^\mathrm{dyn}}} \\
    &= \frac{\frac{(2m-1)2mp}{4m-1}\EE\left(\lambda_\mathrm{long}^{\widetilde D^\mathrm{dyn}}\right)\lambda_0}{\frac{2m(2m-1)p}{4m-1}p_\mathrm{swap}}
    \\
    &= \EE\left(\lambda_\mathrm{long}^{\widetilde D^\mathrm{dyn}}\right)\frac{\lambda_0}{p_\mathrm{swap}}. \label{eq: w dyn oss}
\end{align}

\section{Steady-state QBER and Secret Key Rate Heuristic}\label{app:skr}
In this appendix, we derive the steady-state estimated Quantum Bit Error Rate (QBER) for the multiplexed repeater.
We specifically show that, like the steady-state rate and Werner parameter, the steady-state estimated QBER is also almost surely a constant.
We then use it to provide a heuristic estimate of the secret key rate of the repeater when we run the key distribution protocol described in Appendix \ref{sec:steadyStateQBER}, which is an entanglement-based implementation of the BB84 protocol~\cite{bennett1984quantum}.
We also explore the dependence of the estimated secret key rate on the cutoff for local entanglement generation attempts. In particular, we combine the exact analysis and the OSS approximation of Appendix \ref{app: markov applied} to heuristically optimize the secret key rate over the cutoff.

\subsection{Steady-state QBER}\label{sec:steadyStateQBER}

We consider the following asymmetric implementation of the protocol~\cite{tomamichel2012tight}, which is run for time steps $n \le N$ and we estimate the QBER at the end of time step $N$.
At time step $n$, user Alice at the left end node measures each of the $U_n$ generated links in basis $\mathbb{Z}$ with probability $p_n^\text{z}$ and in basis $\mathbb{X}$ with probability $1-p_n^\text{z}$, where $p_n^\text{z} \to 0$ as $n\to \infty$.
User Bob at the right end node does the same.
The measurement outcomes from the links for which both Alice and Bob choose basis $\mathbb{Z}$ are used for QBER estimation, whereas the measurement outcomes for which both choose basis $\mathbb{X}$ are used to make the key.
The QBER converges to the desired limit even when $p_n^\text{z}$ is kept constant across time steps, although we prove the result for the following efficient implementation of the protocol where $p_n^\text{z} \to 0$.

\begin{description}

\item[State preparation and distribution]
At time step $n$, $U_n$ end-to-end links are produced, each described by a Werner state of the form \cref{eq: werner state}.

\item[Measurement]
For each of the $U_n$ links produced, Alice chooses the measurement basis $\mathbb{Z}$ with probability $p_n^\text{z}$ where $(p_n^\text{z})^2=1/(2 \sqrt{n})$ and basis $\mathbb{X}$ with probability $p_n^\text{x}:=1-p_n^\text{z}$ for her half of the pair.
Bob does the same for his half, independently of Alice.
Since the ideal link is given by the singlet state, Bob flips his measurement outcomes.
Both store their outcomes in their respective classical registers.
We assume that the measurement process is instantaneous, i.e., there is no further depolarization of the generated links during this process.

\item[Sifting]
After time step $N$, Alice and Bob broadcast their basis choices over the classical channel and keep only the outcomes for which their bases match. 
We define the following quantities.
\begin{itemize}[itemsep=1pt, topsep=1.5pt]
    \item[] $B^\text{z}_n$: number of outcomes measured in $\mathbb{Z}$ basis in time step $n$ for which the bases agree.
    \item[] $Q_n$: number of disagreements between Bob's flipped measurements and Alice's measurement among the $B^\text{z}_n$ outcomes.
\end{itemize}
%
%
%
\item[QBER estimation]
Alice and Bob estimate QBER as 
\begin{equation}\label{eq:qber}
    \bar{Q}_N = \frac{\sum_{n=1}^N Q_n}{\sum_{n=1}^N B^\mathrm{z}_n}.
\end{equation}
The protocol aborts if $\bar{Q}_N$ exceeds a predetermined threshold.
Otherwise, they proceed to the next step.

\item[Classical post processing]
Alice and Bob perform error correction and privacy amplification to obtain a key out of the outcomes where both measured their half of the pair in $\mathbb{X}$ basis.
\end{description}

The key result of this appendix is that the estimated QBER $\bar{Q}_N$ converges almost surely to the anticipated constant limit as formalized below.

\begin{theorem}\label{lemma:qberN}
The estimated QBER satisfies
\begin{align*}
    \bar{Q}_N \overset{\mathrm{a.s.}}{\longrightarrow} \frac{1-w}{2}.
\end{align*}
where $w$ is the steady-state Werner parameter as defined in \cref{eq:steadyStateWerner}.
\end{theorem}

Observe that the result is independent of the multiplexing policy, i.e., under the fixed (resp. dynamic) multiplexing policy, $w$ can be replaced by $w_\mathrm{fxd}$ from \cref{eq:steadyStateWernerFxd} (resp. $w_\mathrm{dyn}$ from \cref{eq:steadyStateWernerDyn}).
The rest of the appendix is devoted to proving Thm.~\ref{lemma:qberN} and providing a secret key rate heuristic using it.
The main idea of the proof is that both the numerator and denominator in \cref{eq:qber}, when normalized by $\sqrt{N}$, converge to respective constants almost surely, thereby establishing convergence of the estimated QBER.
To establish the result formally, we first express $\bar{Q}_N$ in terms of Werner parameters of individual links generated until time step $N$.
Recall that the signed queue length process $(S_n)_{n\geq 1}$ is defined on the state space $\mathcal{S}$ and its stationary distribution is $\pi$.
Further, conditional on $S_n=\ell$, independently across time step $n$,
\begin{align*}
    (U_n,W_{n,1},W_{n,2},\dots) \overset{d}{=} (U'_\ell,W'_{\ell,1},W'_{\ell,2},\dots),    
\end{align*}
where $W_{n,i}$ ($1 \le i \le U_n$) denotes the Werner parameters of the $i$th link generated during time step $n$ and $W'_{\ell,i}$ denotes the Werner parameters of the $i$th link generated during a time step starting with signed queue length $\ell$. 
We emphasize that the Werner parameters of individual links are considered at the end of the time step to account for depolarization during the time step.

For $n\geq 1$ and $1\leq i\leq U_n$, let $I_{n,i}$ denote the Bernoulli variable indicating if the $i$th link generated during time step $n$ is used for QBER estimation, i.e.,
\begin{align}\label{eq:Ini}
  I_{n,i}\sim \mathrm{Bern}((p_n^\mathrm{z})^2),~~ \text{with}~~
(p_n^\mathrm{z})^2:=\frac{1}{2\sqrt n}, 
\end{align}
independent of everything else. 
Then, $B^\mathrm{z}_n=\sum_{i=1}^{U_n}I_{n,i}$ and $V_n=\sum_{i=1}^{U_n}W_{n,i}$.
Given $I_{n,i} = 1$, let $Q_{n,i}$ be the Bernoulli variable taking value $1$ in case Alice and Bob have different outcomes after Bob has flipped his. Then, by properties of the Werner state in~\cref{eq: werner state},
\begin{align}\label{eq:Qni}
  Q_{n,i}|(I_{n,i} = 1, W_{n,i}) \sim \mathrm{Bern}\Big(\frac{1-W_{n,i}}{2}\Big).
\end{align}
For notational convenience, we also define
$Q_{n,i}|(I_{n,i} = 0) =0$, which leads to following expressions for the summands in the numerator and denominator of the estimated QBER (Eq.~\ref{eq:qber}):
\begin{align}\label{eq:MnQn}
  Q_n=\sum_{i=1}^{U_n}Q_{n,i}, \quad  B^\mathrm{z}_n=\sum_{i=1}^{U_n}I_{n,i}.
\end{align}
We now consider the asymptotic behavior of $\sum_{n=1}^N Q_n/\sqrt{N}$ and $\sum_{n=1}^N B^\mathrm{z}_n/\sqrt{N}$, where the scaling is motivated by $1/\sqrt{n}$ Bernoulli thinning in \cref{eq:Ini}.
We first prove a lemma for a generic Markov chain, which will be useful in establishing the corresponding limits.


\begin{lemma}\label{lemma:sqrtLimit}
Let \((X_n)_{n\ge 1}\) be an irreducible aperiodic Markov chain on a finite state $\mathcal{S}'$ space with steady-state distribution $\pi'$.
For a bounded deterministic function $\alpha:\mathcal S'\to\mathbb R$, we define 
\begin{align}\label{eq:muAlpha}
    \mu_{(\alpha)}:=\sum_{\ell \in \mathcal{S}'} \pi'_\ell \alpha(\ell)   
\end{align}
Then,
\[
\frac{1}{\sqrt N}\sum_{n=1}^N \frac{\alpha(X_n)}{\sqrt n}
\overset{\mathrm{a.s.}}{\longrightarrow} 2\mu_{(\alpha)}.
\]
\end{lemma}

\begin{proof}
The main idea is to consider the deviation of the summands around the asymptotic mean and show that the cumulative deviation scaled by $\sqrt{N}$ converges to zero almost surely.
We have
\begin{align}\label{eq:alphaXn}
\frac{1}{\sqrt N}\sum_{n=1}^N \frac{\alpha(X_n)}{\sqrt n}
=
\mu_{(\alpha)} \; \underbrace{\frac{1}{\sqrt N}\sum_{n=1}^N \frac{1}{\sqrt n}}_{{\longrightarrow}\;2}
+
\frac{1}{\sqrt N}\sum_{n=1}^N \,\underbrace{\frac{\alpha(X_n)-\mu_{(\alpha)}}{\sqrt n}}_{=:A_n/\sqrt{n}}.
\end{align}
Now letting
$R_n:=\sum_{k=1}^n A_k$, 
by ergodic theorem for Markov chains,
\begin{align}\label{eq:ergodicR}
    \frac{R_n}{n}=\frac1n\sum_{k=1}^n(\alpha(X_k)-\mu_{(\alpha)})\overset{\mathrm{a.s.}}{\longrightarrow}0.    
\end{align}
Rewriting $A_n = R_n-R_{n-1}$ with the convention that $R_0=0$, summation by parts gives
\begin{align}\label{eq:AnBySqrtN}
     \sum_{n=1}^N \frac{A_n}{\sqrt n} = \frac{R_N}{\sqrt N} + \sum_{n=1}^{N-1} R_n\!\left(\frac1{\sqrt n}-\frac1{\sqrt{n+1}}\right)
    \implies  \frac1{\sqrt N}\sum_{n=1}^N \frac{A_n}{\sqrt n} = \underbrace{\frac{R_N}{N}}_{\substack{\overset{\mathrm{a.s.}}{\longrightarrow}0 \\ \eqref{eq:ergodicR}}} + \frac1{\sqrt N}\sum_{n=1}^{N-1} R_n\!\left(\frac1{\sqrt n}-\frac1{\sqrt{n+1}}\right)
\end{align}
%
We now consider the second term for each $\omega$ in the set where the first term converges to zero, which has probability~$1$.
Given $\varepsilon>0$, we have $|R_n(\omega)|<\varepsilon n$ for $n>N_\varepsilon(\omega)$. 
Since $n\!\left(\frac1{\sqrt n}-\frac1{\sqrt{n+1}}\right)\le \frac1{2\sqrt n}$,
\begin{align}\label{eq:remainderSqrtLemma}
    \frac1{\sqrt N}\sum_{n=1}^{N-1} |R_n(\omega)|\!\left(\frac1{\sqrt n}-\frac1{\sqrt{n+1}}\right) \le \frac1{\sqrt N}\sum_{n=1}^{N_\varepsilon(\omega)-1} |R_n(\omega)|\!\left(\frac1{\sqrt n}-\frac1{\sqrt{n+1}}\right) + \frac\varepsilon{2\sqrt N}\sum_{n=N_\varepsilon(\omega)}^{N-1} \frac1{\sqrt n}.
\end{align}
Since,
$\frac1{\sqrt N}\sum_{n=1}^N \frac{1}{\sqrt n}\to 2$,
taking limsup shows that the RHS in \cref{eq:remainderSqrtLemma} is smaller than $\varepsilon$.
Plugging back in \cref{eq:AnBySqrtN} and in turn in \cref{eq:alphaXn}, the claim follows.
\end{proof}

We will also use the following well-known fact for martingale difference sequences.

\begin{fact}\label{fact:kronckerMDS}
    Let $(D'_k|\mathcal F'_k)$ be a martingale difference sequence (MDS), i.e., $\EE(|D'_k|) < \infty$ and $\EE(D'_k|\mathcal{F}'_{k-1}) = 0$, where $\mathcal{F}'_k$ denotes the corresponding filtration (i.e., the information available until time step $k$).
    Further let $\{b_k\}$ be a positive sequence with $b_k \uparrow \infty$ such that $\sum_{k=1}^\infty \EE({D'}_k^{2})/b_k^2 < \infty$. Then,
    \begin{align}\label{eq:MDSkronecker}
        \frac{1}{b_n} \sum_{k=1}^n D'_k \overset{\mathrm{a.s.}}{\longrightarrow} 0 \quad \text{as} \quad n \to \infty.
    \end{align}
\end{fact}
The fact is based on the idea that if $(D'_k|\mathcal F'_k)$ is an MDS, so is $(\frac{D'_k}{b_k}|\mathcal F'_k)$. 
Defining $Y'_n =\sum_{k \le n} D'_k/b_k$, it can be verified that $(Y'_k|\mathcal{F}'_k)$ is a martingale with $\EE({Y'}_n^2) = \sum_{k=1}^n \EE({D'}_k^{2})/b_k^2$, i.e., $\sup_n \EE({Y'}_n^2)<\infty$. Then by martingale convergence theorem~\cite[Sec. 12.1]{williams1991probability}, $Y'_n \overset{\mathrm{a.s.}}{\longrightarrow} Y'_\infty$ for some $Y'_\infty$ with $|Y'_\infty|<\infty$ a.s.
Applying Kronecker's lemma~\cite[Sec. 12.7]{williams1991probability} pathwise, we see that $\sum_{k=1}^n D'_k /b_n \to 0$ on this probability $1$ set, validating \cref{eq:MDSkronecker}.

To derive the asymptotic behavior of $\sum_{n=1}^N Q_n/\sqrt{N}$, we first define the sum of the Werner parameters of the end-to-end links produced in a time step starting from state $\ell$ as
\begin{align}
    V'_\ell:=\sum_{i=1}^{U'_\ell}W'_{l,i}.    
\end{align}
and the corresponding steady-state mean 
\begin{align}\label{eq:muV}
    \mu_V:=\sum_{\ell \in \mathcal S}\pi_\ell\EE(V'_\ell).
\end{align}
Observe that
\begin{align}
    \mu_V \overset{\text{a.s.}}{=} \lim_{N\to \infty} \frac{1}{N}\sum_{n=1}^N V_n \overset{\text{(i)}}{=} \lim_{N\to \infty} \frac{1}{N}\sum_{n=1}^N V_n^\mathrm{fut} \overset{\text{a.s.}}{=} \mu_{V^\mathrm{fut}},
\end{align}
where the first and last equality follows from Lemma \ref{lemma: main replacing random variable by its expectation in ergodic theorem} and (i) follows from Lemma~\ref{lemma: equivalence of W and W fut}.
We are now ready to derive the convergence result for the scaled numerator of \cref{eq:qber}.

\begin{lemma}\label{lemma:sqrtN}
For the scaled cumulative bit errors, the following result holds 
\begin{align*}
    \frac{1}{\sqrt{N}} \sum_{n=1}^N Q_n  \overset{\mathrm{a.s.}}{\longrightarrow} \frac{\mu_U- \mu_V}{2},   
\end{align*}
where $\mu_U$ and $\mu_V$ were defined in \cref{eq: definition mu U} and \cref{eq:muV} respectively. 
\end{lemma}


\begin{proof}
The main idea is to decompose the bit error $Q_n$ into its conditional mean and a noise term and then show that almost surely the scaled cumulative conditional mean converges to the desired limit, while the scaled cumulative noise term vanishes.
The first part follows from Lemma~\ref{lemma:sqrtLimit} and we use Fact.~\ref{fact:kronckerMDS} for the second part.

Let $\mathcal F_n$ denote the complete information (filtration) until time step $n$.
We write
\begin{align}
&Q_n
=
\EE(Q_n| \mathcal F_{n-1})
+
{D}_n, \quad \text{where} \label{eq:QnDecomp}\\
&{D}_n := Q_n - \EE(Q_n| \mathcal F_{n-1}). \label{eq:Dtilden}
\end{align}
Note that 
$\EE(Q_n| \mathcal F_{n-1}) = \EE(Q_n| S_n)$, i.e., the conditional mean of $Q_n$ given the information until the end of time step $n-1$ is equivalent to that of given the signed queue length $S_n$ at the beginning of time step $n$.
Also, $\EE(D_n| \mathcal F_{n-1}) = 0$, i.e., \(({D}_n,\mathcal F_n)\) is a martingale difference sequence.

\noindent \textbf{Mean term in}~\eqref{eq:QnDecomp}: Observe that
\begin{align*}
&\EE(Q_n| S_n\!=\!\ell)
=
\EE\!\left(\sum_{i=1}^{U_n}Q_{n,i} \mid S_n\!=\!\ell\right)
=
\EE\!\left(\sum_{i=1}^{U_n}
\EE(Q_{n,i}| I_{n,i},W_{n,i},S_n\!=\!\ell)
 \mid S_n\!=\!\ell\right)
=
\EE\!\left(\sum_{i=1}^{U_n}
I_{n,i}\frac{1\!-\!W_{n,i}}{2}
 \mid S_n\!=\!\ell\right), \\
\implies &\EE(Q_n| S_n\!=\!\ell)=
\EE\!\left(\sum_{i=1}^{U_n}
\EE(I_{n,i})
\frac{1\!-\!W_{n,i}}{2}
 \mid S_n\!=\!\ell\right)
=
\frac{1}{2\sqrt n}\,
\EE\!\left(\sum_{i=1}^{U_n}\frac{1\!-\!W_{n,i}}{2}
 \mid S_n\!=\!\ell\right)
= \frac{1}{4\sqrt n}\,
\EE\!\left(U_n-\sum_{i=1}^{U_n}W_{n,i}
 \mid S_n\!=\!\ell\right).
\end{align*}
Recalling that $\sum_{i=1}^{U_n}W_{n,i} = V_n$, we have $\EE(Q_n| S_n\!=\!\ell)= \EE(U_n-V_n| S_n\!=\!\ell) /(4 \sqrt{n}) = \EE(U'_\ell-V'_\ell) /(4 \sqrt{n})$.
We now introduce the shorthand $g(\ell):=\EE(U'_\ell-V'_\ell)$, which gives
\[
\EE(Q_n| \mathcal F_{n-1})
=
\frac{g(S_n)}{4\sqrt n}.
\]
Since $\sum_{\ell \in \mathcal S} \pi_\ell g(\ell) = \sum_{\ell \in \mathcal S} \EE(U'_\ell-V'_\ell) = \mu_U-\mu_V$, applying Lemma~\ref{lemma:sqrtLimit} we have
\begin{align}\label{eq:Qpredictable}
    \frac{1}{\sqrt N}\sum_{n=1}^N \frac{g(S_n)}{\sqrt n} \overset{\mathrm{a.s.}}{\longrightarrow} 2\mu_{(g)} \implies \frac{1}{\sqrt N}\sum_{n=1}^N \EE(Q_n| \mathcal F_{n-1})
\overset{\mathrm{a.s.}}{\longrightarrow}
\frac{\mu_U-\mu_V}{2}.    
\end{align}
\noindent \textbf{Noise term in}~\eqref{eq:QnDecomp}: 
Starting from the definition of ${D}_n$ in \cref{eq:Dtilden}, 
\begin{align}\label{eq:expectDn2}
    \EE({D}_n^2| \mathcal F_{n-1})
    \!=\! \EE((Q_n\!-\!\EE(Q_n| \mathcal F_{n-1}))^2 \mid \mathcal F_{n-1})
    \!=\! \EE(Q_n^2| \mathcal F_{n-1}) \!-\! (\EE(Q_n|\mathcal F_{n-1}))^2 \!\le\! \EE(Q_n^2| \mathcal F_{n-1}) \implies \EE({D}_n^2) \!\le\! \EE(Q_n^2).
\end{align}
We now aim to bound the RHS. 
Recalling the definition of $Q_{n,i}$ from \cref{eq:Qni},
\begin{align}
    Q_{n,i} \le I_{n,i} \overset{\eqref{eq:MnQn}}{\implies} Q_n \le B^\mathrm{z}_n \le U_n \implies \EE(Q_n^2 | U_n) \le U_n \EE(B^\mathrm{z}_n | U_n) \overset{\eqref{eq:Ini}}{=} \frac{U_n^2}{2\sqrt{n}} \implies \EE(Q_n^2) \le \frac{\EE(U_n^2)}{2\sqrt{n}} \le \frac{m^2}{2\sqrt{n}}.
\end{align}
From \cref{eq:expectDn2}, we then have
\begin{align}\label{eq:limDtilde}
\sum_{n=1}^\infty \frac{\EE({D}_n^2)}{n} < \infty \xRightarrow{\text{Fact.\,\ref{fact:kronckerMDS}}} \frac{1}{\sqrt N}\sum_{n=1}^N {D}_n \overset{\mathrm{a.s.}}{\longrightarrow} 0,
\end{align}
establishing the hypothesis. 


\end{proof}

In a similar fashion, we can show that
\begin{align}
    \frac{1}{\sqrt N}\sum_{n=1}^N  B^\mathrm{z}_n\overset{\mathrm{a.s.}}{\longrightarrow}
\mu_U.    
\end{align}
The claim in Thm.~\ref{lemma:qberN} is now immediate.
Scaling both numerator and denominator in \cref{eq:qber} by $1/\sqrt{N}$, we have
$$\bar{Q}_N = \frac{\frac{1}{\sqrt{N}}\sum_{n=1}^N Q_n}{\frac{1}{\sqrt{N}}\sum_{n=1}^N B^\mathrm{z}_n} \overset{\mathrm{a.s.}}{\longrightarrow} \frac{(\mu_U-\mu_V)/2}{\mu_U} \overset{\mathrm{a.s.}}{=} \frac{1-w}{2}.$$

\subsection{Secret Key Rate Heuristic and Corresponding Numerical Optimization over the Cutoff for Local Entanglement Generation Attempts}\label{section: secret key rate optimization}
We have so far shown that steady-state rate and estimated QBER are almost surely constant.
We now use them to provide a heuristic estimate of the secret key rate for the repeater as
\begin{equation}\label{eq: steady-state secret key rate}
    f_\mathrm{skr} := \max\bigg(1-2h\bigg(\frac{1-w}{2}\bigg),0\bigg)r,
\end{equation}
where $h$ denotes the binary entropy function.
Note that we have again adopted a policy-independent notation.
It is well-known that the secret key rate is positive only if the steady-state rate $r>0$ and the steady-state Werner parameter $w>w_\mathrm{QKD}\approx 0.780$.

In the multiplexed repeater described in \cref{section: system model}, the cutoff $c_\mathrm{local}$ on the number of local entanglement generation attempts during entanglement swapping can be used to trade off between steady-state rate and steady-state Werner parameter: a smaller cutoff improves the steady-state Werner parameter at the cost of reducing the steady-state rate. Here we show how the exact analysis and the OSS approximation of Appendix \ref{app: markov applied} can be used to find an approximation to the maximum achievable steady-state secret key rate by optimizing the rate/Werner tradeoff over the cutoff $c_\mathrm{local}$. The results for selected parameter configurations are shown in \cref{fig:NILN_eta_skr_lowerbound} in the main text.

Optimization of the steady-state secret key rate over $c_\mathrm{local}$ is a discrete optimization problem. Evaluating the secret key rate in the exact analysis for a particular value of $c_\mathrm{local}$ is numerically expensive, especially for large values of $\log_2(m)$, which is the regime of interest for positive secret key rates. Optimization by enumeration is therefore infeasible. However, using the OSS approximation, an approximation to the optimal cutoff value can be found, and using a single evaluation of the exact analysis this yields an approximation to the optimal secret key rate. The process is illustrated in \cref{fig:skr optimization} for a particular instance, and the steps are explained in more detail below.

\textit{Step 1:} We first restrict the range of $c_\mathrm{local}$ where the optimal cutoff could be achieved. Recall from Lemma \ref{lemma: upper bound on w} that the steady-state Werner parameter is at most $w_\mathrm{max}=\lambda_0/p_\mathrm{swap}$. Recalling the expressions of $\lambda_0$ and $p_\mathrm{swap}$ from \cref{tab:derived-quantities}, the maximum steady-state Werner parameter at cutoff $c_\mathrm{local}$ is thus at most 
\begin{equation} \label{eq: wmax clocal}
    w_\mathrm{max}(c_\mathrm{local}) = \frac{\lambda_\mathrm{static}\EE\left(\lambda_\mathrm{active}^{2N'_\mathrm{local}}\indicator(N'_\mathrm{local}\leq c_\mathrm{local})\right)}{\PP\left(N'_\mathrm{local}\leq c_\mathrm{local}\right)} =  \lambda_\mathrm{static}\EE\left(\lambda_\mathrm{active}^{2N'_\mathrm{local}}\mid N'_\mathrm{local}\leq c_\mathrm{local}\right),
\end{equation}
where $N'_\mathrm{local}\sim \Geo(p_\mathrm{local})$.
Since $\lambda_\mathrm{static}\leq 1$, it follows from Proposition \ref{proposition: range for n_coh_active to guarantee no key at n is c} below that no key can be produced for $c_\mathrm{local}\geq n_\mathrm{coh-active}$ when $5\leq n_\mathrm{coh-active}\leq 1/p_\mathrm{local}$. In particular, for local entanglement generation probability $p_\mathrm{local}=10^{-3}$, the active coherence time $n_\mathrm{coh-active}=1000$ considered in \cref{fig:NILN_eta_skr_lowerbound} is within this range.

\textit{Step 2:} The steady-state secret key rate in the OSS approximation is computed for $c_\mathrm{local}\in \{1, \dots, n_\mathrm{coh-active}-1\}$ using the closed-form expression of the steady-state rate and steady-state Werner parameter listed in \cref{tab:OSS approximations rate and werner}. The value $\widetilde c^*_\mathrm{local}$ where the steady-state secret key rate in the OSS approximation is maximized is an approximation for the optimal cutoff value.

\textit{Step 3:} An approximation to the optimal steady-state secret key rate after optimizing over $c_\mathrm{local}$ is obtained as
\begin{equation}
    f_\mathrm{skr}^* := f_\mathrm{skr}(r(\widetilde c_\mathrm{local}^*), w(\widetilde c_\mathrm{local}^*)),
\end{equation}
which is the heuristic estimate of the secret key rate given the exact steady-state rate $r$ and Werner parameter $w$ for cutoff $\widetilde c_\mathrm{local}^*$.

\begin{figure}
    \centering
    \includegraphics[width=0.5\linewidth]{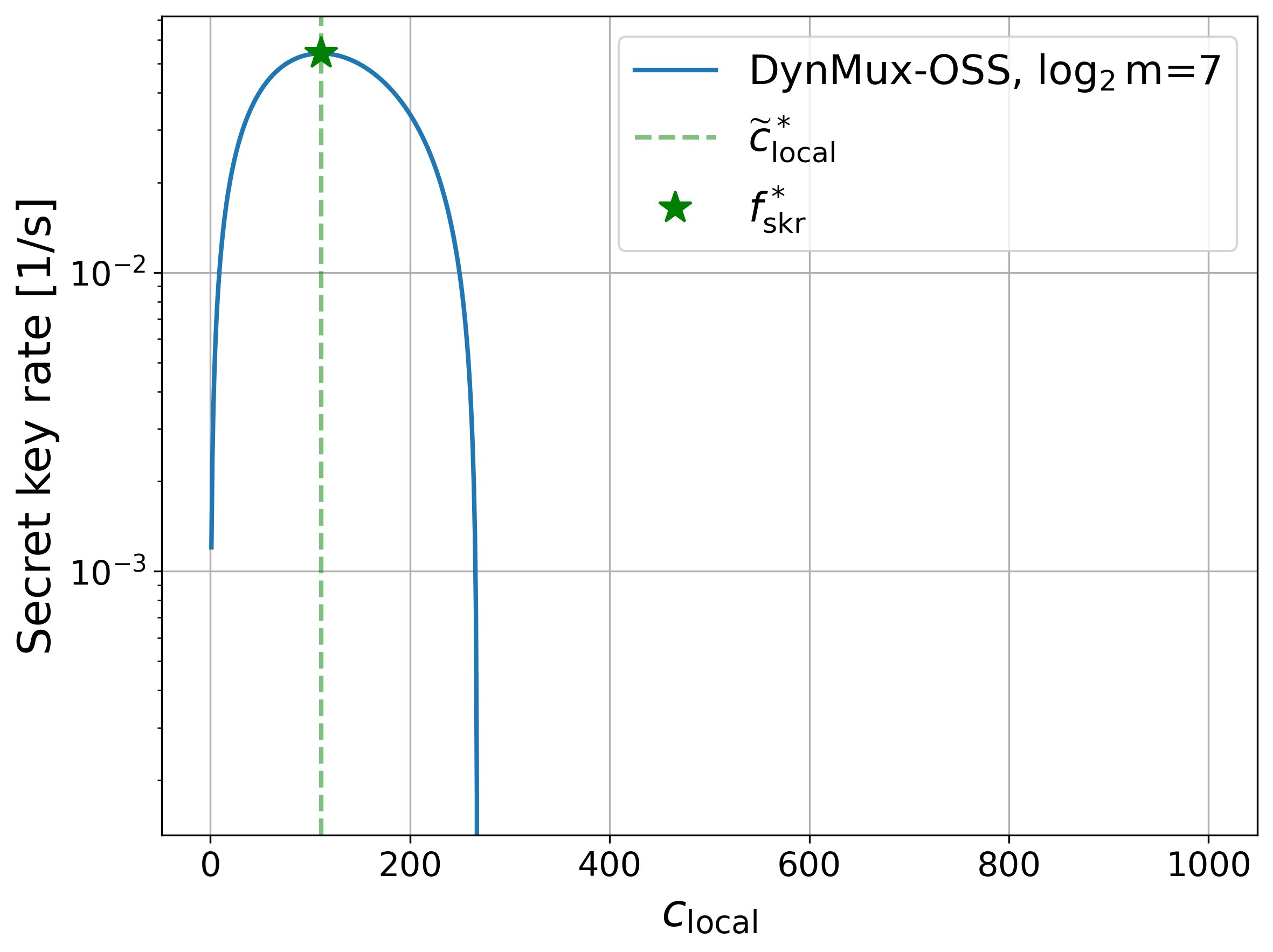}
    \caption{\textbf{Heuristic optimization of secret key rate over cutoff on the number of local entanglement generation attempts.} The curve shows the secret key rate as computed from the OSS approximation for the DynMux policy as a function of the cutoff $c_\mathrm{local}$ on the number of local entanglement generation attempts. The optimal value for the cutoff according to the OSS approximation is $\widetilde c_\mathrm{local}$ and is indicated by the green dashed line. The star marks the approximation $f_{\mathrm{skr}}^*$ to the maximally achievable secret key rate. The instance shown corresponds to the point in \cref{fig:NILN_eta_skr_lowerbound} at $\log_2(m)=7$ and $\eta_\mathrm{switch}=1$, where $p_\mathrm{local}=10^{-3}$ and $n_\mathrm{coh-active}=1000$.}
    \label{fig:skr optimization}
\end{figure}

In step 1 of the numerical optimization of the cutoff on the number of local entanglement generation attempts discussed above the following proposition is used. 
For simplicity, we only consider the case that $n_\mathrm{coh-active}$ is even.
\begin{proposition}\label{proposition: range for n_coh_active to guarantee no key at n is c} 
If $5\leq n_\mathrm{coh-active}\leq 1/p_\mathrm{local}$ is even, then the maximum steady-state Werner parameter $w_\mathrm{max}(c_\mathrm{local})=\lambda_\mathrm{static}\EE\left(\lambda_\mathrm{active}^{2N'_\mathrm{local}}\mid N'_\mathrm{local}\leq c_\mathrm{local}\right)$ with $N'_\mathrm{local}\sim \Geo(p_\mathrm{local})$ is below the QKD threshold $w_\mathrm{QKD}\approx 0.780$ for $c_\mathrm{local}\geq n_\mathrm{coh-active}$.
\end{proposition}

To prove the above proposition, we will use techniques from the theory of majorization \cite{marshall1979inequalities}. A vector $\xvec\in \RR^n$ \textit{is majorized by} a vector $\yvec\in \RR^n$, denoted $\xvec\prec \yvec$, if \cite[Definition 1.A.1]{marshall1979inequalities}
\begin{enumerate}
    \item $\sum_{i=1}^k x^\downarrow_i \leq \sum_{i=1}^k y^\downarrow_i$ for $k=1,\dots, n$
    \item $\sum_{i=1}^n x_i = \sum_{i=1}^n y_i$,
\end{enumerate}
where $\xvec^\downarrow$ is the vector in which the entries of $\xvec$ are sorted in decreasing order. 
Intuitively, $\xvec \prec \yvec$ if the entries of $\xvec$ are spread out more evenly than those of $\yvec$. Let $\mathcal{D}\subset \RR^n$ be the set of all vectors whose entries are sorted in decreasing order. A function $\phi:\mathcal{D}\to \RR$ is \emph{Schur convex on $\mathcal{D}$} (c.f. \cite[Definition 3.A.1]{marshall1979inequalities}) if for all $\xvec,\yvec \in \mathcal{D}$,
\begin{equation}
    \xvec\prec \yvec \implies \phi(\xvec)\leq \phi(\yvec).
\end{equation}
The particular fact that we will use from majorization theory is that taking the inner product with a fixed decreasing vector is a Schur convex function.
\begin{lemma}[Schur convexity of inner product with a fixed decreasing vector]\label{lemma: schur convexity of inner product with decreasing vectors}
    Let $\muvec\in \mathcal{D}$. The function $\phi:\mathcal{D}\to \RR$ defined through $\xvec \mapsto \langle \xvec, \muvec\rangle$, where $\langle-,-\rangle$ is the standard inner product on $\RR^n$,
    is Schur-convex on $\mathcal{D}$. 
\end{lemma}
\begin{proof}
    Let $\muvec \in \mathcal{D}$, then $\xvec\mapsto \langle \xvec, \muvec\rangle$ is smooth on $\RR^n$. By \cite[Theorem 3.A.3]{marshall1979inequalities} it thus suffices to show that $\nabla \phi(\xvec) \in \mathcal{D}$ for all $\xvec\in \mathcal{D}$. This is the case since for $k=1,\dots, n-1$ we have
    \begin{equation}
        (\nabla \phi(\xvec))_k = (\partial_k\phi)(\xvec) = \mu_k \geq \mu_{k+1}= (\partial_{k+1}\phi)(\xvec) = (\nabla \phi(\xvec))_{k+1},
    \end{equation}
    where we have used that $(\partial_k \phi)(\xvec) = \mu_k$ for all $\xvec\in \RR^n$ and the inequality follows since the entries of $\muvec$ are decreasing.
\end{proof}

We use the Schur convexity of the inner product to show that $w_\mathrm{max}(c_\mathrm{local})$ defined in \cref{eq: wmax clocal} is monotonically decreasing with $c_\mathrm{local}$. To simplify the notation, we let $q,\lambda\in (0,1)$ and consider the map
    \begin{equation}\label{eq: definition of map f}
        f_{(q,\lambda)}:\NN\to [0,1], \qquad c\mapsto \EE\left(\lambda^{N'}\mid N'\leq c\right)\quad\mathrm{with}\quad N'\sim \Geo(q),
    \end{equation}
which has the same structure as $w_\mathrm{max}(c_\mathrm{local})$ in \cref{eq: wmax clocal}.
\begin{lemma}\label{lemma: monotonicity of werner with c}
    The map $f_{(q,\lambda)}:\NN\to [0,1]$ defined in \cref{eq: definition of map f} is monotonically decreasing.
\end{lemma}
\begin{proof}
    Define the infinite dimensional vector $\muvec = (\mu_k)_{k\geq 1}$ with entries $\mu_k = \lambda^k$. Then, with $\zvec^{(c)} = (z^{(c)}_k)_{k\geq 1}$ defined by the conditional geometric probabilities
    \begin{equation}
        z^{(c)}_k = 
        \begin{cases}
            q(1-q)^{k-1}/\sum_{j=1}^{c}q(1-q)^{j-1} & \mathrm{if}\,k\leq c \\
            0 & \mathrm{otherwise}
        \end{cases},
    \end{equation}
    we can express $f_{(q,\lambda)}(c)$ as 
    \begin{equation}
        f_{(q,\lambda)}(c) = \langle \zvec^{(c)},\muvec\rangle. 
    \end{equation}
    We now restrict $\muvec$, $\zvec^{(c)}$, and $\zvec^{(c+1)}$ to the first $c+1$ entries. These restrictions are all part of $\mathcal{D}\subset \RR^{c+1}$. Moreover, $\zvec^{(c+1)}$ is majorized by $\zvec^{(c)}$. It thus follows from Lemma \ref{lemma: schur convexity of inner product with decreasing vectors} that
    \begin{equation}
        f_{(q,\lambda)}(c) =  \langle \zvec^{(c)},\muvec\rangle \geq \langle \zvec^{(c+1)},\muvec\rangle = f(c+1),
    \end{equation}
    and we conclude that $f_{(q,\lambda)}$ decreases monotonically with $c$.
\end{proof}
The above lemma shows that $w_\mathrm{max}(c_\mathrm{local})$ decreases monotonically with $c_\mathrm{local}$. 
The next step is to upper bound $w_\mathrm{max}(n_\mathrm{coh-active})$ when the cutoff is taken to be the active coherence time. We use the map $f_{(q,\lambda)}$ defined in \cref{eq: definition of map f} again.
\begin{lemma}\label{lemma: bound on truncated geometric expecation}
    Let $\lambda = e^{-2/n}$. If $1\leq n\leq 1/q$ is even, then
    \begin{equation}
        f_{(q,\lambda)}(n)\leq \frac{2}{n}\frac{1-e^{-1}}{1-e^{-2/n}}.
    \end{equation}
\end{lemma}
\begin{proof}
    Let $\zvec = (z_k)_{k=1}^{n}$ be the vector whose entries equal the conditional geometric probabilities
    \begin{equation}
        z_k = \frac{q(1-q)^{k-1}}{\sum_{j=1}^{n}q (1-q)^{j-1}} = \frac{q(1-q)^{k-1}}{1-(1-q)^n}
    \end{equation}
    and let $\muvec = (\mu_k)_{k=1}^{n}$ be the vector with entries
    \begin{equation}
        \mu_k = \lambda^{k}.
    \end{equation}
    Then, 
    \begin{equation}
        f_{(q,\lambda)} = \langle \zvec, \muvec\rangle.
    \end{equation}
    Now define the vector $\zvec'=(z'_k)_{k=1}^{n}$ with entries
    \begin{equation}
        z_k' = \begin{cases}
            2/n & \mathrm{if}\,k\leq n/2 \\
            0 & \mathrm{otherwise}
        \end{cases}.
    \end{equation}
    The entries of $\muvec$, $\zvec$, and $\zvec'$ are in decreasing order, i.e. they are all in $\mathcal{D}\subset \RR^n$. Moreover, by the assumption that $n\leq 1/q$, it follows from Lemma \ref{lemma: majorization assumption} that
    \begin{equation}\label{eq: majorization assumption application}
        z'_1=\frac{2}{n}\geq \frac{q}{1-(1-q)^n} =z_1,
    \end{equation}
    so that $\zvec$ is majorized by $\zvec'$. Since $\xvec \mapsto \langle \xvec,\muvec\rangle$ is Schur convex by Lemma \ref{lemma: schur convexity of inner product with decreasing vectors}, it follows that
    \begin{equation}
        f_{(q,\lambda)}(n) = \langle \zvec, \muvec\rangle \leq \langle \zvec',\muvec\rangle = \sum_{k=1}^{n/2} \frac{2}{n}\lambda^k \leq \frac{2}{n}\frac{1-e^{-1}}{1-e^{-2/n}},
    \end{equation}
    where in the final step we substituted $\lambda = e^{-2/n}$ and used that $\lambda \leq 1$.
\end{proof}
It follows from the above lemma that 
\begin{equation}\label{eq: upper bound w max}
    w_\mathrm{max}(n_\mathrm{coh-active}) \leq \frac{2}{n_\mathrm{coh-active}}\frac{1-e^{-1}}{1-e^{-2/n_\mathrm{coh-active}}}.
\end{equation}
The inequality in \cref{eq: majorization assumption application} in the proof of the lemma above is ensured by the following.
\begin{lemma}\label{lemma: majorization assumption}
    If $n\geq 1$, then $q\leq 1/n$ implies that
    \begin{equation}
        \frac{2}{n}\geq \frac{q}{1-(1-q)^n}.
    \end{equation}
\end{lemma}
\begin{proof}
    Let $n\geq 1$. We will show that for $q\leq 1/n$,
    \begin{equation}\label{eq: inverted majorization condition}
       \frac{n}{2}\leq \frac{1-(1-q)^n}{q}.
    \end{equation}
    The right hand side of \cref{eq: inverted majorization condition} is monotonically decreasing in $q$ for $q\leq 1$. This can be seen directly from the representation of the right hand side as the sum $\sum_{k=0}^{n-1}(1-q)^k$: if $q$ increases, then each term in the sum becomes smaller. Hence, for $q\leq 1/n$, the right hand side is minimized at $q=1/n$. At $q=1/n$, the inequality in \cref{eq: inverted majorization condition} is satisfied since
    \begin{equation}
        \frac{1-(1-q)^n}{q}\Big\lvert_{q=1/n} = \frac{1-(1-1/n)^n}{1/n} \geq n(1-e^{-1}) \geq \frac{n}{2}
    \end{equation}
    where we have used that $(1-1/n)^n\leq e^{-1}\approx 0.368$ \footnote{To see this, note that for any $x\in \RR$ it holds that $1-x\leq e^{-x}$. With $x=1/n$ this yields $1-1/n \leq e^{-1/n}$. Raising this inequality to the $n$-th power gives the result.}. We conclude that the inequality in \cref{eq: inverted majorization condition} holds for all $q\leq 1/n$, which proves the lemma.
    
\end{proof}

The final lemma shows that the upper bound \cref{eq: upper bound w max} found from Lemma \ref{lemma: bound on truncated geometric expecation} decreases with $n_\mathrm{coh-active}$.
\begin{lemma}\label{lemma: monotonic function}
    The function $g:[1,\infty)\to \RR$ defined by $n \mapsto \frac{2}{n}\frac{1-e^{-1}}{1-e^{-2/n}}$ is monotonically decreasing.
\end{lemma}
\begin{proof}
    We will show that the derivative $g'$ is nonpositive on $[1,\infty)$. We have
    \begin{equation}
        g'(n) = -\frac{2}{n^2} \frac{1-e^{-1}}{1-e^{-2/n}} + \frac{4}{n^3}\frac{(1-e^{-1})e^{-2/n}}{(1-e^{-2/n})^2}.
    \end{equation}
    Hence, for $n\in [1,\infty)$,
    \begin{equation}\label{eq: chain of equivalances monotonicity lemma}
        g'(n)\leq 0 \iff -1 + \frac{2}{n}\frac{e^{-2/n}}{1-e^{-2/n}}\leq 0 \iff \frac{2}{n}\sum_{k=1}^{\infty}e^{-2k/n} \leq 1.
    \end{equation}
    Bounding the latter sum by an integral it follows that
    \begin{equation}
        \frac{2}{n}\sum_{k=1}^{\infty}e^{-2k/n} \leq \frac{2}{n}\int_0^\infty dx\,e^{-2x/n}= 1.
    \end{equation}
     This verifies the last inequality in \cref{eq: chain of equivalances monotonicity lemma}. We conclude that $g'(n)\leq 0$ for $n\in [1,\infty)$ so that $g$ is monotonically decreasing on $[1,\infty)$. 
\end{proof}

\begin{proof}[Proof of Proposition \ref{proposition: range for n_coh_active to guarantee no key at n is c}]
    Let $n_\mathrm{coh-active}$ be such that $5\leq n_\mathrm{coh-active}\leq 1/p_\mathrm{local}$. By Lemma \ref{lemma: monotonicity of werner with c} the maximum Werner parameter $w_\mathrm{max}(c_\mathrm{local})$ decreases monotonically with $c_\mathrm{local}$. To show that $w_\mathrm{max}(c_\mathrm{local})\leq w_\mathrm{QKD}$ for all $c_\mathrm{local}\geq n_\mathrm{coh-active}$, it thus suffices to show it for $c_\mathrm{local}=n_\mathrm{coh-active}$.
    We assume that $n_\mathrm{coh-active}\leq 1/p_\mathrm{local}$ so that 
    \begin{equation}
        w_\mathrm{max}(n_\mathrm{coh-active}) \leq \frac{2}{n_\mathrm{coh-active}}\frac{1-e^{-1}}{1-e^{-2/n_\mathrm{coh-active}}}
    \end{equation}
    by Lemma \ref{lemma: bound on truncated geometric expecation}.
    The upper bound on the right hand side is monotonically decreasing with $n_\mathrm{coh-active}$ by Lemma \ref{lemma: monotonic function}, and for $n_\mathrm{coh-active}=5$ it evaluates to approximately $0.767$, which is below the QKD threshold $w_\mathrm{QKD}\approx 0.780$. Since $n_\mathrm{coh-active}$ was assumed to be at least $5$, it follows that $w_\mathrm{max}(n_\mathrm{coh-active})$ is below the QKD threshold.
\end{proof}

\section{Validation with NetSquid}\label{app: validation simulation}

In this appendix we discuss the validation of the exact analysis and the OSS approximation with simulation. We have simulated the multiplexed quantum repeater using NetSquid \cite{coopmans2021netsquid}, which is a discrete event simulator for quantum networks. The simulation code has been made available in the repository at \cite{grimbergen2026multiplexing}. 

\subsection{Description of NetSquid Simulation}
NetSquid \cite{coopmans2021netsquid} is a discrete event simulator for quantum networks. 
The efficiency of a discrete event simulator lies in the fact that it samples the random time between events, rather than sampling the random outcome of every time step. 
Internally, NetSquid tracks the density matrices of all (entangled) qubits in the network. Decoherence channels are applied to qubits every time an event occurs that involves these qubits.
For the multiplexed quantum repeater we simulate two main events:
\begin{enumerate}
    \item Multiplexed deliveries of links between quantum chips on the repeater node and one of the end nodes.
    \item Deliveries of internal links between two quantum chips on the repeater node.
\end{enumerate}
The distribution of the time between these events is known. NetSquid samples these distributions to schedule the next events into the future. The simulation then jumps to the first event that will occur. The \textsc{swap} gates and Bell state measurements that have to be performed after multiplexed deliveries and internal deliveries, respectively, are simulated as instantaneous. 
When an end-to-end link is created, the corresponding density matrix is immediately removed from the simulation and an \emph{end-to-end delivery} is registered in the simulation results. Each end-to-end delivery has the attributes \texttt{time} and \texttt{fidelity}, corresponding to the time when the concerned end-to-end link was generated and its fidelity to the ideal Bell state, respectively. 

The raw simulation data consists of lists of end-to-end deliveries corresponding to different \textit{runs} of the simulation. Each run starts with an empty repeater and ends when $N_\mathrm{links}$ end-to-end links have been generated. For each parameter configuration, we simulate a total of $N_\mathrm{runs}$ independent runs.
Below and in \cref{alg: raw data processing} we describe how the statistics are obtained from the raw data. 

For a single run, the steady-state rate is estimated as
\begin{equation}
    \hat r = \frac{U}{T},
\end{equation}
where $U$ is the total number of end-to-end links generated in the run, and $T$ is the total time needed to generate these links. Since the run stops when the number of end-to-end links reaches some predefined total number of links $N_\mathrm{links}$, we have that $U=N_\mathrm{links}$ for every run, and that $T$ is equal to the time at which the last link was generated. The simulation outputs fidelity which can be converted to Werner parameter using the relation $w(F)=(4F-1)/3$ for Werner states. The steady-state Werner parameter for a single run is estimated as
\begin{equation}
    \hat w = \frac{V_\mathrm{sum}}{U},
\end{equation}
where $V_\mathrm{sum}$ is the sum of all the Werner parameters of end-to-end links in the run. Statistics are obtained by estimating the steady-state rate and steady-state Werner parameter over $N_\mathrm{runs}$. The overall estimate of the simulation is taken to be the mean of the estimates for each run, and the error in the simulation is quantified by the standard error of the mean.



\begin{algorithm}[H]
\caption{\textsc{ProcessRawResults}(ResultsRaw)}
\begin{algorithmic}[1]
\State $\textit{rates} \gets []$
\State $\textit{werners} \gets []$
\ForAll{$(\textit{runId},\ \textit{deliveries}) \in \textit{ResultsRaw}$}
    \State $U \gets |\textit{deliveries}|$ \Comment{total number of links is total number of deliveries}
    \State $T \gets \textit{deliveries}[-1].\text{time}$
    \Comment{total time is time of final delivery}
    \State $V_{\text{sum}} \gets \displaystyle\sum_{d \,\in\, \textit{deliveries}} w(d.\text{fidelity})$ \Comment{Convert fidelity to Werner parameter and sum}
    \State $\hat{r} \gets U \,/\, T$ \Comment{rate estimate}
    \State $\hat{w} \gets V_{\text{sum}} \,/\, U$ \Comment{Werner estimate}
    \State Append $\hat{r}$ to \textit{rates}
    \State Append $\hat{w}$ to \textit{werners}
\EndFor
\State
\State $N_\mathrm{runs} \gets |\textit{rates}|$ \Comment{number of runs}
\State $\bar{r} \gets \text{mean}(\textit{rates})$
\State $\bar{w} \gets \text{mean}(\textit{werners})$
\State $\sigma_r \gets \text{std}(\textit{rates}) \,/\, \sqrt{N_\mathrm{runs}}$ \Comment{standard error of the mean}
\State $\sigma_w \gets \text{std}(\textit{werners}) \,/\, \sqrt{N_\mathrm{runs}}$
\State
\State \Return $\left(\, r_{\text{sim}} = \bar{r},\; w_{\text{sim}} = \bar{w},\; \sigma_r,\; \sigma_w \,\right)$
\end{algorithmic}
\label{alg: raw data processing}
\end{algorithm}

\subsection{Validation of Analytical Findings via NetSquid Simulation}
We compare the steady-state rate and Werner parameter from the exact analysis, the OSS approximation, and NetSquid simulation in \cref{fig:IL2_validation_full_sim} for the same ideal local entanglement generation parameters as in \cref{fig:IL2_rate_werner}.
For the simulation, we set the number of runs to $N_\mathrm{runs}=50$ and the number of end-to-end links per run to $N_\mathrm{links}=2000$. Validation results for more parameter configurations including non-ideal local entanglement generation can be found in the data repository \cite{grimbergen2026multiplexing_data}. In all parameter configurations tested in the regime $2mp<1$ the three methods of computation show good agreement, supporting our hypothesis 
that the OSS approximation is a good approximation for the exact analysis in this regime. 

\begin{figure*}[htbp]
    \centering
    \begin{subfigure}[t]{0.45\textwidth}
        \centering
        \includegraphics[width=\textwidth]{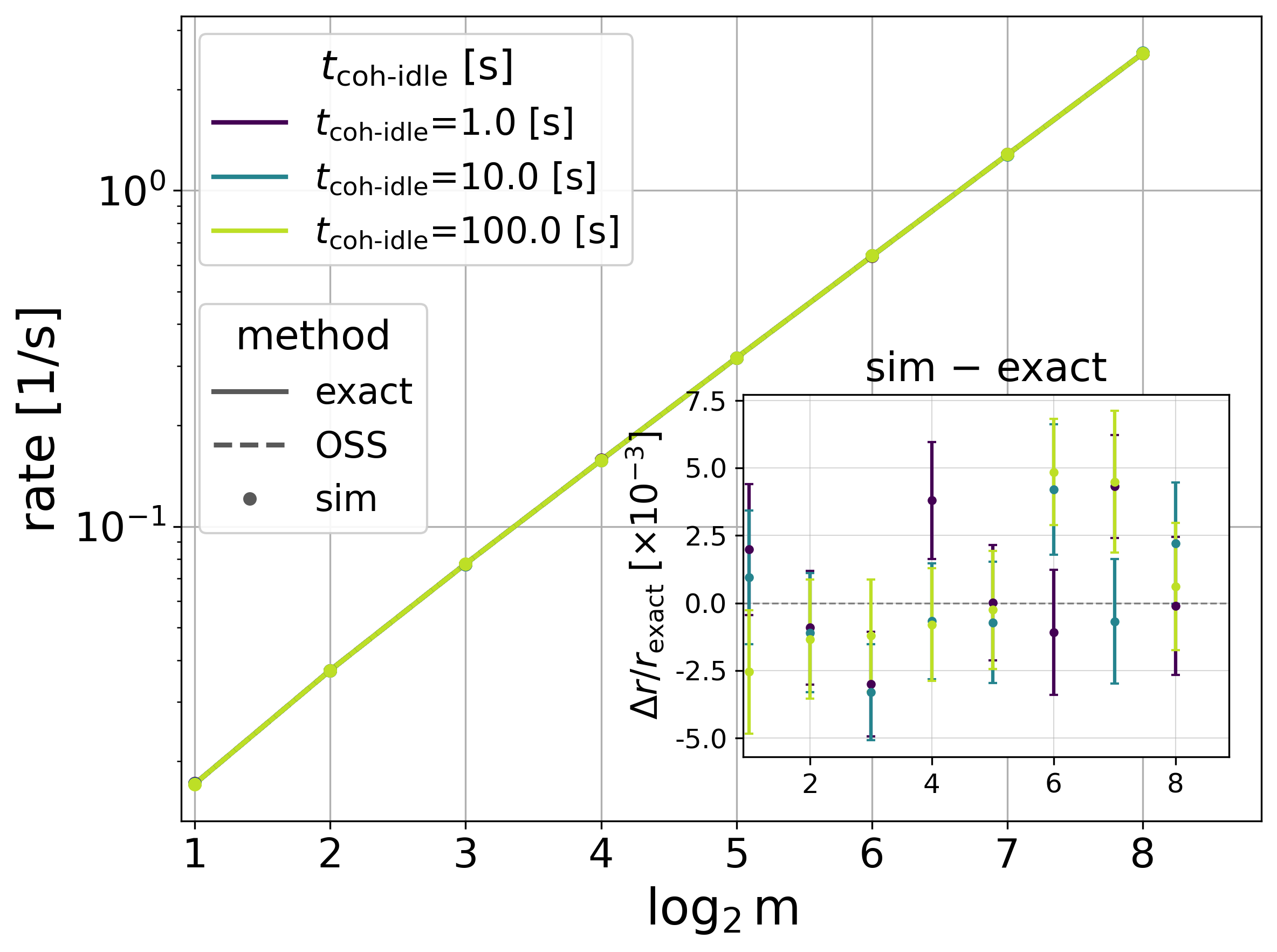}
        \caption{DynMux, rate}
        \label{fig:IL2_validation_rate2000_dyn}
    \end{subfigure}
    \hfill
    \begin{subfigure}[t]{0.45\textwidth}
        \centering
        \includegraphics[width=\textwidth]{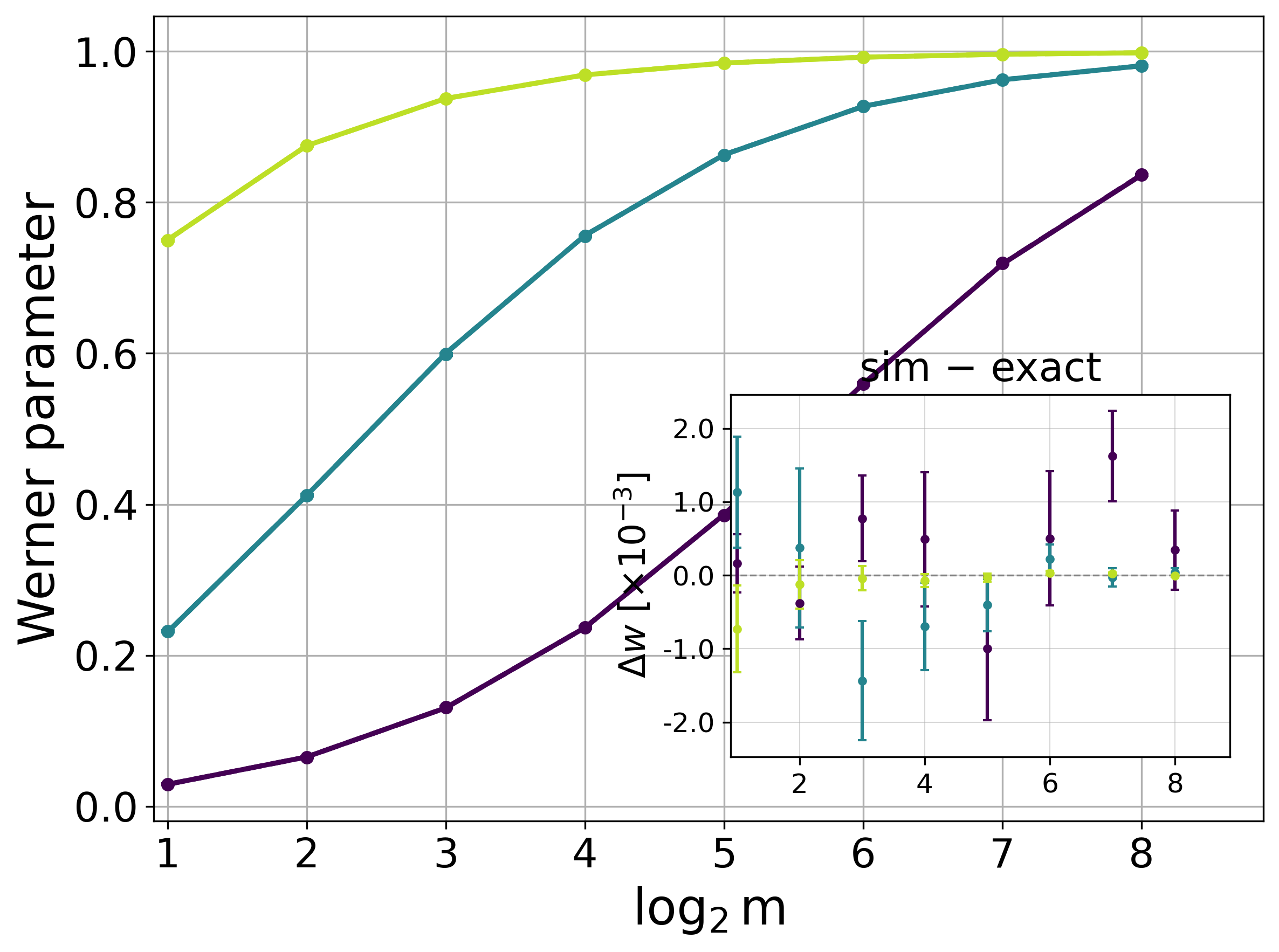}
        \caption{DynMux, Werner parameter}
        \label{fig:IL2_validation_Werner2000_dyn}
    \end{subfigure}
    \\
    \begin{subfigure}[t]{0.45\textwidth}
        \centering
        \includegraphics[width=\textwidth]{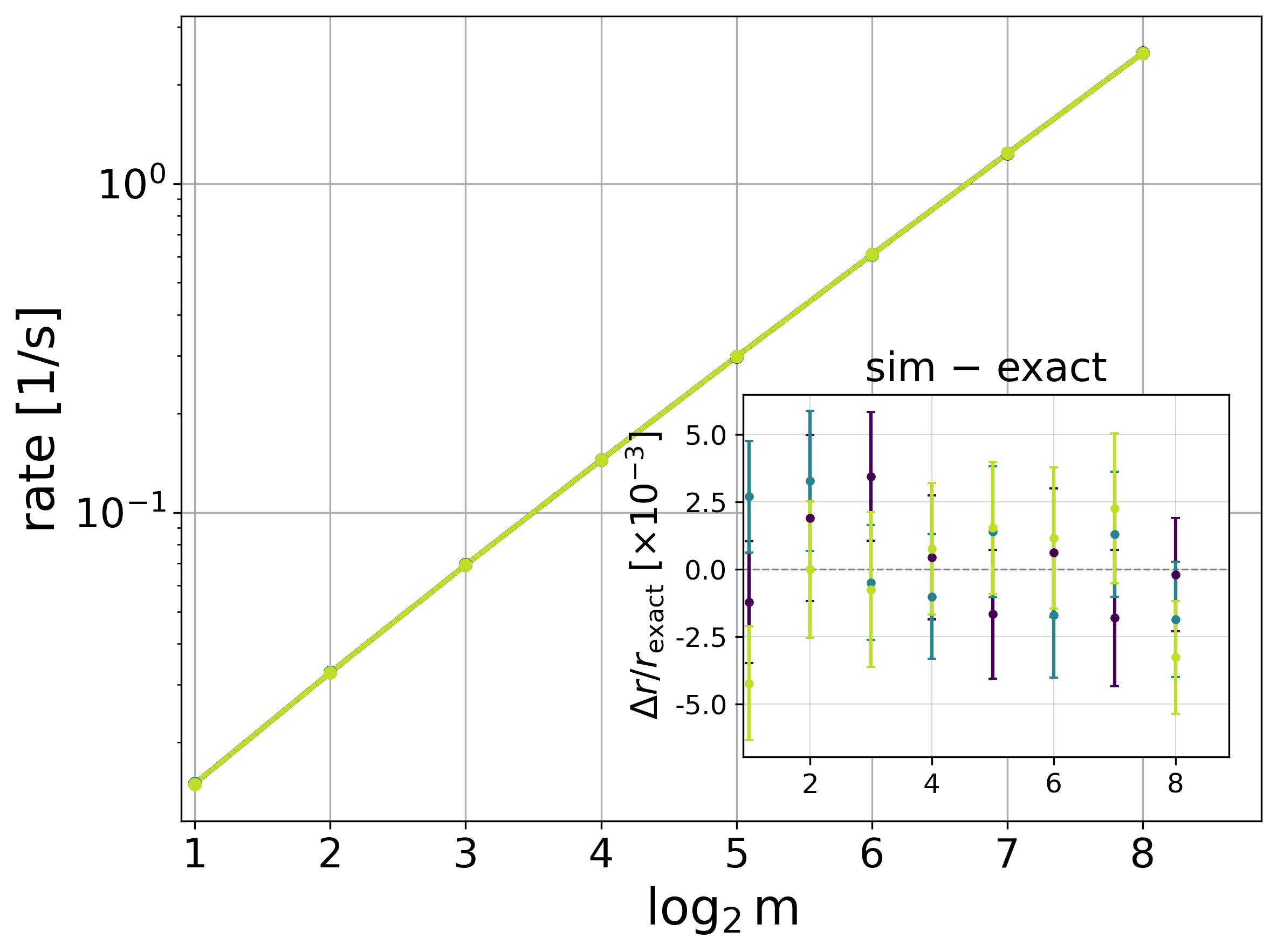}
        \caption{FxdMux, rate}
        \label{fig:IL2_validation_rate2000_fxd}
    \end{subfigure}
    \hfill
    \begin{subfigure}[t]{0.45\textwidth}
        \centering
        \includegraphics[width=\textwidth]{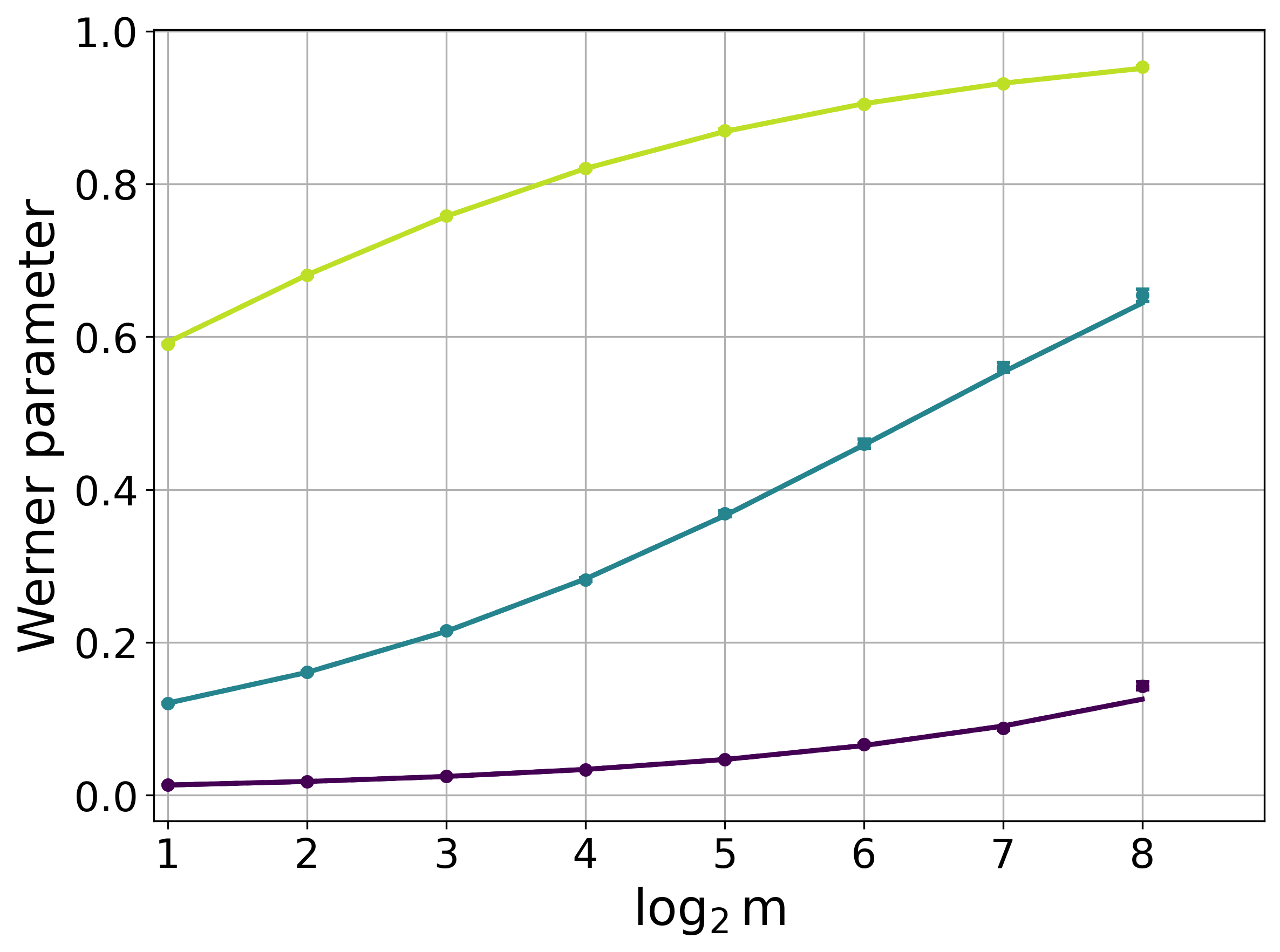}
        \caption{FxdMux, Werner parameter}
        \label{fig:IL2_validation_Werner2000_fxd}
    \end{subfigure}
    \caption{\textbf{Validation of Markov chain models with NetSquid simulation.} For the same parameters $2m=4,\dots, 512$ ($\log_2(m)=1,\dots,8$), $p=10^{-5}$, $t_\mathrm{coh-idle}=\SI{1}{\second},\SI{10}{\second},\SI{100}{\second}$, $t_\mathrm{long}=\SI{1}{\milli\second}$, $t_\mathrm{local}=\SI{1}{\micro\second}$, $p_\mathrm{local}=1$, and $n_\mathrm{coh-active}=\infty$ as in \cref{fig:IL2_rate_werner} \textbf{(a)} shows the steady-state rate for dynamic multiplexing computed with exact analysis (solid line), OSS approximation (dashed line), and NetSquid simulation (circles). Note that the lines for different $t_\mathrm{coh-idle}$ overlap since the rate is independent of $t_\mathrm{coh-idle}$. The lines for the exact analysis and the OSS approximation overlap, which shows that OSS approximates the exact analysis well in this regime. Simulation statistics are obtained over $50$ independent runs with each run spanning until $2000$ end-to-end links generation. Error bars show one standard error in the mean but are smaller than markers. The inset shows the relative difference in rate between the simulation and exact analysis as $\Delta r/r_\mathrm{exact}$ where $\Delta r=r_\mathrm{sim}-r_\mathrm{exact}$. Similarly, and for the same parameters, \textbf{(b)} shows the steady-state Werner parameter for DynMux, \textbf{(c)} the steady-state rate for FxdMux, and \textbf{(d)} the steady-state Werner parameter for FxdMux. Legend of (a) applies to (b), (c), and (d) as well. In (b), the inset for the Werner parameter shows the absolute error $\Delta w=w_\mathrm{sim}-w_\mathrm{exact}$. In (d) the inset is not shown because there is a visible difference between the simulation results and the exact analysis/OSS approximations for large values of $\log_2(m)$. 
    This is investigated in more detail in \cref{fig:IL2_validation_bias_decrease_fxdmux}, where it is shown that the bias decreases when the number of links per run is increased from $2000$ to $5000$, or the first $500$ links are discarded.}
    \label{fig:IL2_validation_full_sim}
\end{figure*}


However, there seems to be slight bias in the simulation to yield higher steady-state Werner parameters than the Markov chain analyses. This can be seen for example in \cref{fig:IL2_validation_full_sim} for the steady-state Werner parameter in the FxdMux policy at $\log_2(m)=8$. \cref{fig:IL2_validation_Werner2000_fxd} highlights the bias more clearly. There, the difference between the simulated steady-state Werner parameter and the steady-state Werner parameter from the exact analysis is seen to increase with the number of chips. 

Our hypotheses as to why the simulation overestimates the Werner parameter are that it does not completely reach steady-state or that the estimates of the long-time averages in \cref{alg: raw data processing} are still influenced by contributions from the transient phase. 
Both hypotheses are motivated by the fact that we expect that with more chips on the node, the system needs more end-to-end links to be generated before steady-state is reached. This could explain why the bias increases with $\log_2(m)$. In particular, the simulation might miss the rare events where the queue size becomes large. These events are rare because it is unlikely that many links are generated on a single segment in one step (in dynamic multiplexing) or that links are repeatedly generated on one segment without generating on the other segment (in fixed multiplexing). The impact of these rare events would be to reduce the Werner parameter, since it leads to links that have to be stored in queue for a long time. We have observed numerically that increasing the number of links per run (\cref{fig:IL2_validation_Werner5000_fxd}) or discarding the first links (\cref{fig:IL2_validation_Wernertc500_fxd}) both reduce the bias.

\begin{figure*}[htbp]
    \centering
    \begin{subfigure}[t]{0.3\textwidth}
        \centering
        \includegraphics[width=\textwidth]{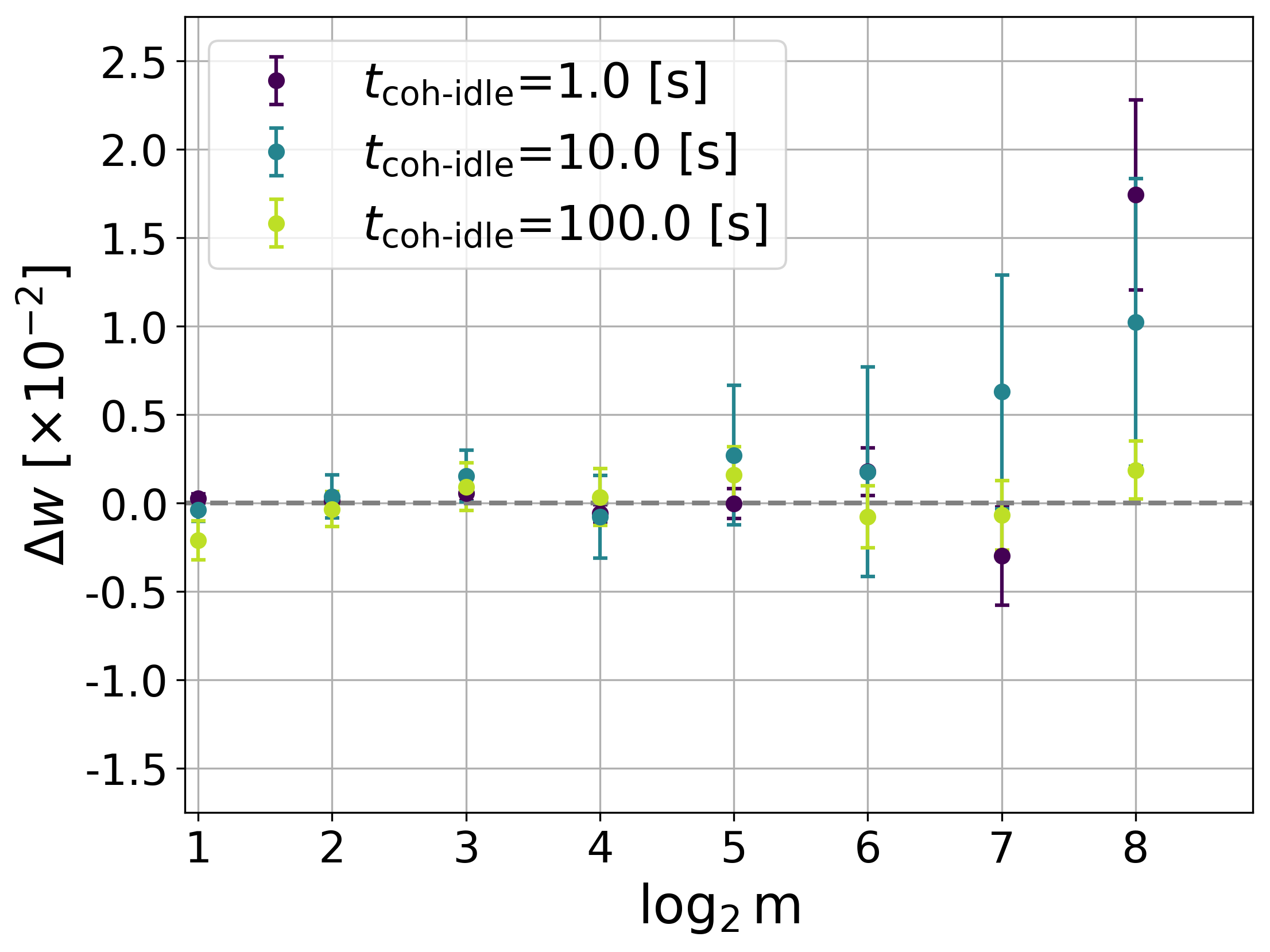}
        \caption{$N_\mathrm{links}=2000$}
        \label{fig:IL2_validation_Werner2000_fxd}
    \end{subfigure}
    \hfill
    \begin{subfigure}[t]{0.3\textwidth}
        \centering
        \includegraphics[width=\textwidth]{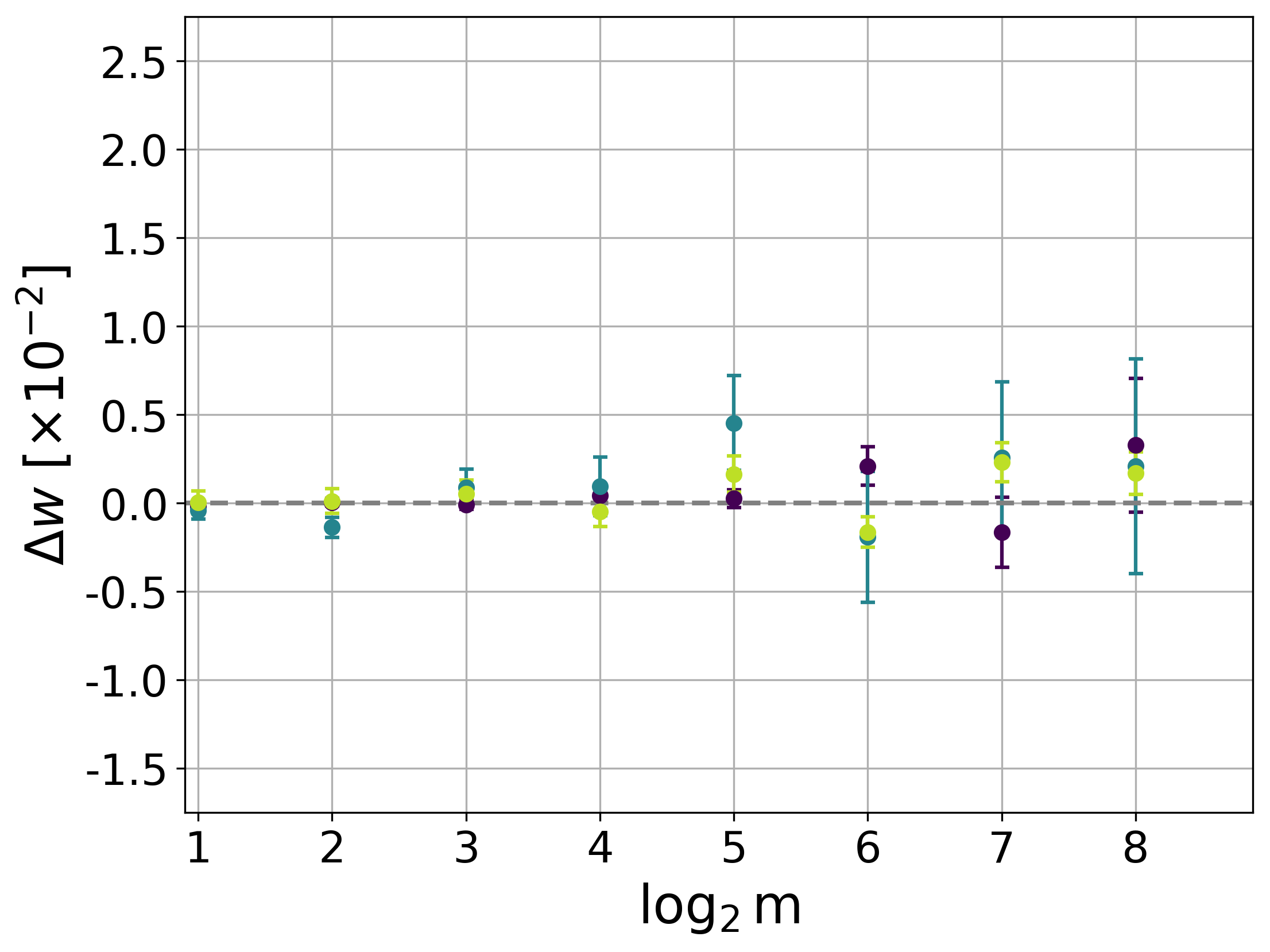}
        \caption{$N_\mathrm{links}=5000$}
        \label{fig:IL2_validation_Werner5000_fxd}
    \end{subfigure}
    \hfill
    \begin{subfigure}[t]{0.3\textwidth}
        \centering
        \includegraphics[width=\textwidth]{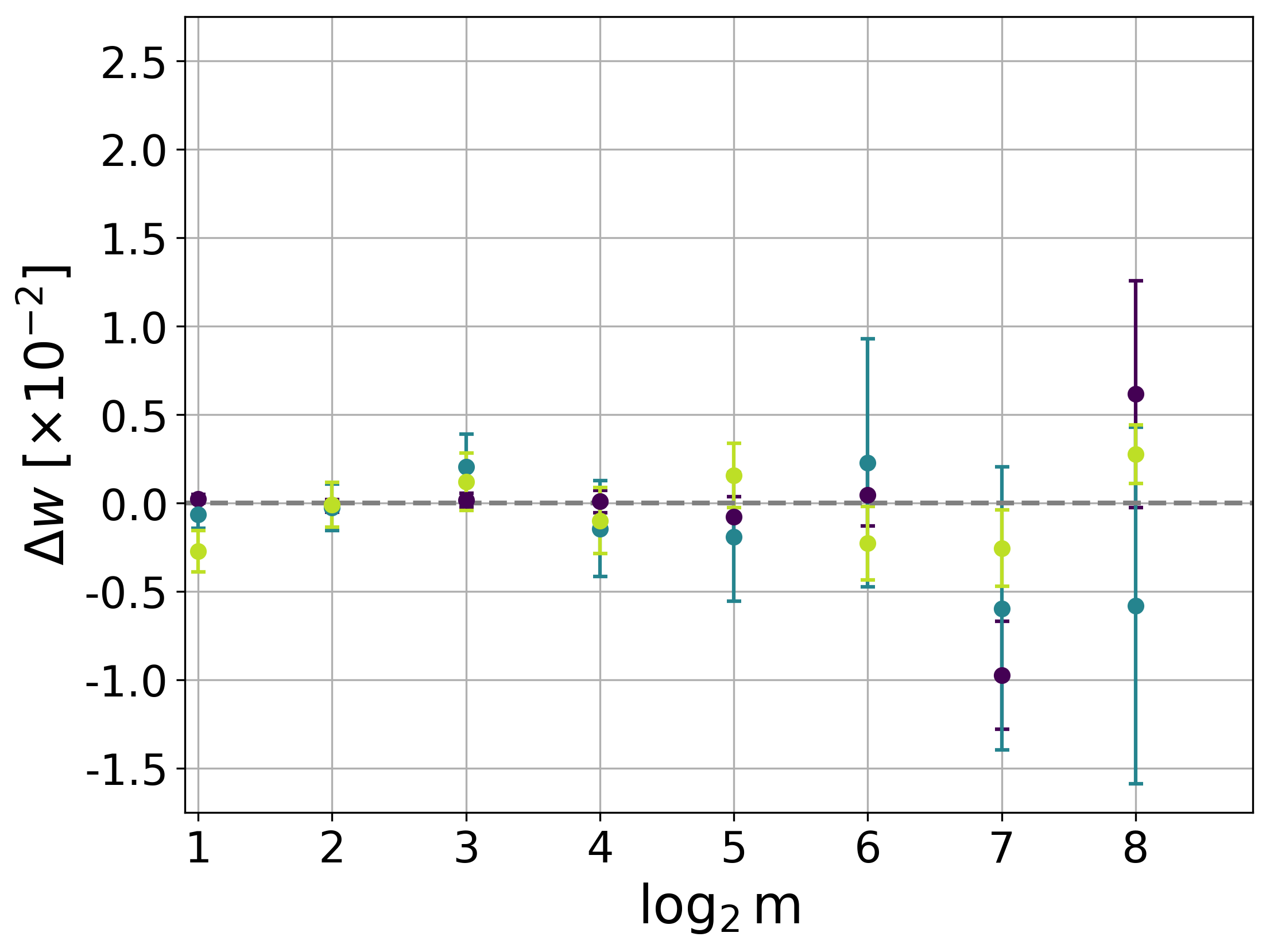}
        \caption{$N_\mathrm{links}=2000$, first $500$ links discarded}
        \label{fig:IL2_validation_Wernertc500_fxd}
    \end{subfigure}
    \caption{\textbf{Bias of NetSquid simulation decreases when using more links per run or when discarding the first links.} The difference in Werner parameter between NetSquid simulation and the exact analysis is shown for FxdMux with the same system parameters as in \cref{fig:IL2_validation_full_sim}. Legend of (a) applies also to (b) and (c). In \textbf{(a)} the difference is shown when the number of links per run is $2000$, as in \cref{fig:IL2_validation_full_sim}. The simulation consistently produces higher Werner parameter than the exact analysis, and this bias increases with the number of chips. In \textbf{(b)} the number of links per run is increased to $5000$, and the bias and magnitude of the error are reduced. In \textbf{(c)} the total number of links was $N_\mathrm{links}=2000$, but the first $500$ links were discarded while estimating the long-time averages in \cref{alg: raw data processing}. This reduces the bias, but the magnitude of the error stays similar.
    }
    \label{fig:IL2_validation_bias_decrease_fxdmux}
\end{figure*}

\end{document}